\documentclass[prx,twocolumn,showpacs,amsfonts,amsmath,floatfix,superscriptaddress]{revtex4-2}

\usepackage{natbib}
\usepackage{amsthm, amsmath,amssymb}
\usepackage{graphicx}
\usepackage{dcolumn}
\usepackage{hyperref}
\usepackage{relsize}
\usepackage{soul}
\usepackage{array}
\usepackage{tabularx}
\hypersetup{
    colorlinks=true,
    linkcolor=blue,
    filecolor=blue, 
    citecolor=blue,
    urlcolor=blue,
}

\usepackage{bm}
\usepackage{xcolor}
\usepackage{booktabs}
\usepackage{multirow}
\usepackage{graphicx}
\graphicspath{%
  {figs/}%
  {plots/}%
  {plots1/}%
  {plots2/}%
}
\usepackage{scrextend}

\usepackage{caption}
\usepackage{subcaption}

\newtheorem{lemma}{Lemma}

\DeclareMathOperator{\csch}{csch}
\newcommand{\tr}{\mathrm{Tr}}

\newcommand{\cC}{\mathcal{C}}
\newcommand{\cD}{\mathcal{D}}
\newcommand{\cE}{\mathcal{E}}

\newcommand{\cR}{\mathcal{R}}
\newcommand{\cS}{\mathcal{S}}

\newcommand{\BS}{{\rm BS}}

\newcommand{\dagg}{\dagger}

\begin{document}
\title{Bias-Preserving Gates and Quantum Error Correction With Dual-Rail Cat Codes}

\author{Debjyoti Biswas}

\thanks{These authors have contributed equally}

\affiliation{ Department of Microtechnology and Nanoscience (MC2),Chalmers University of Technology, SE-412 96 Göteborg, Sweden}

\author{Nikhil Sharma}

\thanks{These authors have contributed equally}
\affiliation{ Department of Microtechnology and Nanoscience (MC2),Chalmers University of Technology, SE-412 96 Göteborg, Sweden}
\affiliation{Department of Physics, Indian Institute of Science Education and Research, Pune 411008, India}

\author{Alberto Salvador}
\affiliation{ Department of Microtechnology and Nanoscience (MC2),Chalmers University of Technology, SE-412 96 Göteborg, Sweden}

\author{Rui Wang}
\affiliation{ Department of Microtechnology and Nanoscience (MC2),Chalmers University of Technology, SE-412 96 Göteborg, Sweden}

\author{Mats Granath}
\affiliation{Department of Physics, University of Gothenburg, Sweden}

\author{Adithi Udupa}
\affiliation{ Department of Microtechnology and Nanoscience (MC2),Chalmers University of Technology, SE-412 96 Göteborg, Sweden}

\author{Giulia Ferrini}
\affiliation{ Department of Microtechnology and Nanoscience (MC2),Chalmers University of Technology, SE-412 96 Göteborg, Sweden}
\begin{abstract}
 
Scalable fault-tolerant quantum computation requires quantum error-correcting codes that simultaneously support universal logical operations, suppress hardware-specific noise, and enable efficient handling of photon-loss error. Bosonic encodings such as the dual-rail and cat codes each offer attractive features but also exhibit important limitations when used in isolation. The dual-rail code enables efficient single-photon loss detection by converting leakage from the computational subspace induced by the loss error into an erasure error. In contrast, the cat code provides a resource-efficient, bias-tailored error-correction scheme with bias-preserving logical gate operations.
Here, we introduce the dual-rail cat code (DRCC), a concatenated bosonic encoding
that combines an inner cat code with an outer dual-rail structure, thereby inheriting and enhancing the advantages of both constituent codes.
We analyse the error-correction properties of the DRCC and propose a deterministic single-photon loss correction protocol by concatenating it with an outer repetition code. 
Exploiting the code’s intrinsic noise bias, we construct arbitrary single-qubit logical rotations using only beam-splitter interactions and complement them with entangling two-qubit gates, thereby obtaining a universal logical gate set. We demonstrate that the proposed logical gates preserve the erasure-biased noise structure.
The DRCC offers several distinctive advantages, including the absence of relative geometric phases during gate operations, deterministic erasure correction with an outer repetition code, and simultaneous syndrome extraction without interrupting stabilisation. These features make the DRCC a promising bosonic code for hardware-efficient, bias-preserving, and erasure-resilient fault-tolerant quantum computation.

\end{abstract}
\maketitle

\section{Introduction}

Quantum error correction (QEC) is essential for protecting quantum information from decoherence and forms the foundation of fault-tolerant quantum computation (FT-QC). Conventional architectures based on two-level systems typically require substantial hardware overhead, as many physical qubits are needed to encode a single logical qubit. Bosonic quantum error-correcting codes address this challenge by encoding logical information in the infinite-dimensional Hilbert space of a harmonic oscillator, thereby exploiting the intrinsic redundancy of a single physical mode~\cite{GKP_2001, Arne_GKP_Puri}.

An important feature of many bosonic encodings is that the dominant physical noise processes induce highly biased logical noise, whereby one class of logical errors occurs much more frequently than others~\cite{Puri_2020}. Such bias can be exploited to design tailored fault-tolerant protocols with significantly reduced overhead~\cite{Tuckett2018UltrahighNoise, Ataides2021XZZX,Srivastava2022xyzhexagonal, dua2022clifford}. In particular, syndrome-extraction and gate operations can be engineered to preserve the underlying noise asymmetry, thereby preventing the conversion of dominant errors into depolarising noise.

The two-component cat code provides a prominent example of this paradigm. Owing to its strong bias towards a single type of Pauli error induced by single-photon loss, several bias-preserving logical gate constructions have been proposed, including bias-preserving entangling gates~\cite{A_Laverrier_prl,hillmann2023quantum, Puri_2020}. To suppress leakage outside the computational manifold induced by gate operation or single-photon gains, cat qubits are commonly stabilised using two-photon dissipation~\cite{xu2023autonomous, Puri_2020,vv_albert_njp}. {Recent work has shown that logical bit-flip errors in dissipatively stabilised cat qubits can be converted into detectable erasures}~\cite{Ferrari2026}.
Bias-preserving gate protocols have been developed in Kerr-cat, dispersive-cat, and squeezed-cat architectures~\cite{Puri_2020, PRX_mazyar_mirrahimi, Putterman_2025_aws_cat,hillmann2023quantum}. A common drawback of these approaches, however, is that logical gate operations generally accumulate relative geometric phases that require additional compensation~\cite{Puri_2020,bp_gates_zoller, PRX_mazyar_mirrahimi}.


A complementary approach is provided by dual-rail encodings, originally introduced in the context of optical quantum computation~\cite{Chuang_1995_dual_rail, Chuang_dual_prl, KLM,knill2003bounds} { and subsequently implemented in photonic~\cite{OBrien2003} and microwave~\cite{Yale_dual, Kubica_erasure_threshold, S_puri_erasure_threshold} quantum-information-processing experiments}. Dual-rail qubits allow conversion of single-photon loss (or amplitude-damping) events into detectable leakage or erasure errors~\cite{Kubica_2023}, thereby transforming one of the dominant error mechanisms in several experimental platforms into an error model that is substantially easier to decode. Since the locations of erasures are known, a distance-$d$ code can correct up to $d-1$ erasure errors~\cite{Kubica_2023}. Leakage-to-erasure conversion has recently been demonstrated  experimentally~\cite{Yale_dual}. Nevertheless, dual-rail protocols rely on accurate erasure identification and may suffer from false syndrome flagging~\cite{Yale_dual, Kubica_erasure_threshold, S_puri_erasure_threshold}. Furthermore, the bias-preserving properties of the logical controlled-$Z$-gate can be compromised when leakage occurs during, rather than before, the gate operation~\cite{mehta2025bias}.

In this work, we combine the advantages of both approaches by introducing the dual-rail cat code (DRCC), obtained by concatenating an inner cat code with an outer dual-rail encoding. DRCCs possess a rotational symmetry distributed across two bosonic modes and can therefore be viewed as a class of multimode rotationally symmetric bosonic (RSB) codes~\cite{multimode_RSB}. The dual-rail structure allows for converting single-photon loss errors into detectable leakage events, while the cat-code encoding retains the favourable biased-noise properties that enable efficient fault-tolerant operations.

The constituent even- and odd-parity cat states have been deterministically prepared and characterised in superconducting cavities since early circuit-QED experiments~\cite{vlastakis2013deterministically,Sun2014}. Subsequently, a two-mode cat state corresponding to a particular logical superposition within the DRCC-1 codespace was deterministically prepared and coherently characterised~\cite{Wang2016}. More recently, a closely related two-cavity cat Bell state was employed in a quantum-communication protocol on a superconducting-cavity platform~\cite{dural_Cat_bs_yale}. Although these experiments demonstrate the physical accessibility of states relevant to the DRCC-1 encoding, they do not constitute an implementation of DRCC-1 as a quantum error-correcting code, and its error-correction capabilities and fault-tolerant properties have not yet been systematically investigated. Here, we analyse the QEC properties of DRCCs using two-component cat codes, examining the Knill--Laflamme conditions and the resulting logical noise bias. Building on techniques developed for dual-rail encodings~\cite{Yale_dual},
we construct arbitrary single-qubit rotations using beam-splitter interactions
and complement them with entangling gates and leakage-detection protocols. We show that all logical gates preserve the erasure-biased noise structure and therefore remain compatible with bias-tailored fault-tolerant protocols. Unlike previous gate constructions for multimode RSB codes~\cite{multimode_RSB}, our approach does not require gate teleportation. We then provide a scalable QEC protocol by concatenating DRCCs with an outer repetition code and introduce a deterministic leakage-detection and correction scheme.

The DRCC framework possesses several distinctive features. First, syndrome measurements can be performed without interrupting the dissipative stabilisation of the cat-code manifold. Second, logical gate operations do not accumulate relative geometric phases, in contrast to existing single-mode and pair-cat implementations. Third, leakage detection and correction can be implemented deterministically while maintaining the favourable erasure-biased error model. Together, these properties establish DRCCs as a promising bosonic code for hardware-efficient, bias-preserving, and leakage-resilient fault-tolerant quantum computation.

Our protocols employ auxiliary two-level systems (i.e., qubits) for measurement and control,  compatibly with a broad range of experimental platforms, including superconducting circuits~\cite{Putterman_2025_aws_cat}, trapped ions~\cite{trapped_ion_bosonic_qubit}, and neutral-atom architectures~\cite{neutral_atm_read, Bohnmann_2025}. We also discuss the relationship between DRCCs and other bosonic code constructions, including dual-rail, but also multimode cat-codes such as pair-cat codes ~\cite{Albert_pair_cat} and repetition cat codes~\cite{PRX_mazyar_mirrahimi}.
Multimode cat codes, distinct from DRCC, were also considered in the context of all-optical schemes~\cite{optical_multi_Cat} and in connection to a group-theoretic construction in Ref.~\cite{A_Laverrier_prl}. Finally, we show that a DRCC based on four-component cat states can simultaneously correct single-photon loss and dephasing errors, further extending the error-correction capabilities of the framework.

The paper is organised as follows. In Sec.~\ref{sec:drcc}, we review multimode RSB codes and introduce the DRCC.  We then investigate the QEC properties of the two-component dual-rail cat code and demonstrate how photon-loss-induced Pauli noise and dephasing noise acting on the single-mode cat code are converted into detectable leakage, i.e., erasure errors, in Sec.~\ref{sec:QEC_properties}. Sec.~\ref{sec:QEC_properties} also discusses the noise bias of the dual-rail two-component cat code. We then outline possible dissipative mechanisms for stabilising the dual-rail cat code with inner two-component cat states in Sec.~\ref{sec:logical gates}, and using these results, we construct logical gates. 
In Sec.~\ref{sec:conc_code}, we consider the concatenation of the dual-rail two-component cat code and outer repetition codes. In Sec.\ref{se:general-comparison-with-DR}, we present a comparison between the properties of DRCCs and the standard dual-rail code, with respect to error correctability and noise-bias.  We present the numerical analysis of the concatenated DRCC in Sec.~\ref{sec:numerics}, and also compare it with cat repetition codes. In Sec.~\ref{sec:drcc2}, we analyse the Knill--Laflamme conditions and QEC properties of the DRCC with a four-component inner cat code. Finally, we conclude with a discussion and outlook in Sec.~\ref{sec:discussion}. 
We summarise the abbreviations used throughout this work in Table \ref{tab:ab_tab}.

\begin{table}[t!]
\centering
\caption{Table of abbreviation}
\begin{tabular}{c c c}
\midrule
Abbreviation& &  Definition \\
\midrule
QEC  & & Quantum Error Correction \\
FT-QC & &  Fault-Tolerant Quantum Computation\\
DRC & & dual-rail code\\
DRCC & & dual-rail cat code\\

    RSB code & & Rotation-symmetric bosonic code\\
    KL condition & & Knill-Laflamme condition \\
    BS & & Beam splitter \\
    BP && Bias-preserving \\
    GP & &  Geometric phase\\
    TMS && Two-mode squeezing\\
    CSS &&  Calderbank–Shor–Steane\\
    LDPC && Low density parity check \\
    \midrule
\end{tabular}

\label{tab:ab_tab}

\end{table}

\section{Dual-rail cat codes}\label{sec:drcc}

Rotation-symmetric bosonic (RSB) codes constitute a broad class of quantum error-correction codes characterised by discrete rotational symmetries in phase space~\cite{RSB_grimsmo_PRX}. Prominent examples include cat codes, binomial codes, number-phase codes, and squeezed Fock-state codes~\cite{sq_fock_state_QEC}. Among these, cat codes with an $N$-fold rotational symmetry and $2N$ coherent-state components have attracted significant attention because of their noise-bias properties and compatibility with bosonic hardware platforms~\cite{Putterman_2025_aws_cat,bp_gates_zoller}.

Recently, Ref.~\cite{multimode_RSB} extended a group-theoretic framework~\cite{A_Laverrier_prl} for constructing a broad class of multimode RSB codes. 
Following this framework, one can express the generic logical states of a two-mode RSB code as
\begin{align}\label{eq:m-RSB}
    |0_L\rangle
    &=
    \hat{U}_{\rm BS}(\delta,\phi)
    \sum_{mn}
    f_{mn}
    |2mN\rangle|(2n+1)N\rangle
   \nonumber\\
    &=
    \hat{U}_{\rm BS}(\delta,\phi)|\bar{0}_N\rangle,
    \\
    |1_L\rangle
    &=
    \hat{U}_{\rm BS}(\delta,\phi)
    \sum_{mn}
    f_{mn}
    |(2n+1)N\rangle|2mN\rangle \nonumber\\
   & =
    \hat{U}_{\rm BS}(\delta,\phi)|\bar{1}_N\rangle,
\end{align}
where $\hat U_{\rm BS}(\delta,\phi)$ denotes a beam-splitter (BS) unitary,
\begin{align}\label{eq:bs}
    \hat U_{\rm BS}(\delta,\phi)
    =
    \exp\!\left[
    -j\delta
    \left(
    \hat G_{12}^{+}\sin\phi
    +
    \hat G_{12}^{-}\cos\phi
    \right)
    \right],
\end{align}
where $j= \sqrt{-1}$ and operators $\hat G^{\pm}_{12}$ are defined as 
\begin{align}
    \hat{G}_{12}^{+}
    &=
    \hat{a}_1\hat{a}_2^{\dagger}
    +
    \hat{a}_1^{\dagger}\hat{a}_2,
    \\
    \hat{G}_{12}^{-}
    &=
    j(\hat{a}_1\hat{a}_2^{\dagger}
    -
    \hat{a}_1^{\dagger}\hat{a}_2),
\end{align}
which transforms the mode operators $(\hat{a}_1, \hat{a}_2)$ as
\begin{align}\label{eq:BS_transformation}
    \begin{pmatrix}
        \hat{b}_1\\
        \hat{b}_2
    \end{pmatrix}
  &  =\begin{pmatrix}
        \cos\delta & e^{-j\phi}\sin\delta\\
        -e^{j\phi}\sin\delta & \cos\delta
    \end{pmatrix}  \begin{pmatrix}
        \hat{a}_1\\
        \hat{a}_2
    \end{pmatrix}.
\end{align}
Here, $(\hat{a}_i,\hat{a}_i^{\dagger})$ with $i\in\{1,2\}$ denote the annihilation and creation operators for the two bosonic modes.

The code vectors $\{|\bar{0}_N\rangle,|\bar{1}_N\rangle\}$ in Eq.\eqref{eq:m-RSB} can be understood as the concatenation of an arbitrary single-mode RSB code  with a dual-rail code, defined as
\begin{align}\label{eq:dual_rail_stdrd}
    |\bar 0_{\rm dr}\rangle
    &=
    |01\rangle,
    \\
    |\bar 1_{\rm dr}\rangle
    &=
    |10\rangle.
\end{align}
Choosing
\begin{align}\label{eq:f_mn_cat}
    f_{mn}
    =
    e^{-2|\alpha|^2}
    \frac{
    \alpha^{2mN}\alpha^{(2n+1)N}
    }{
    \sqrt{((2n+1)N)!}\sqrt{(2mN)!}
    }
\end{align}
in Eq.~\eqref{eq:m-RSB} yields the dual-rail cat code with rotation symmetry order $N$, with $2N$ cat components in each mode (DRCC-$N$).

Although the beam-splitter operation $\hat{U}_{\text{BS}}$ mixing the two modes as in Eq.\eqref{eq:m-RSB} does not play any role in protecting against photon loss error, as expected in view of the properties of the loss channel \cite{Brod_2020}, it can in some cases improve correctability and QEC performance against dephasing noise~\cite{multimode_RSB}, which we will corroborate in Sec.\ref{sec:drcc2} for the dual-rail cat code with inner four components.

\subsection{Code vectors for DRCC-1}

The main focus of this work is to study the error-correction properties of the dual-rail cat code with $N=1$, hereafter referred to as the DRCC-1 (Fig.~\ref{fig:drcc_1}). Its logical basis states are defined as
\begin{align}\label{eq:DRCC-10}
    |\bar{0}\rangle
    &=
    |C^{+}\rangle|C^{-}\rangle,
    \\
\label{eq:DRCC-11}
    |\bar{1}\rangle
    &=
    |C^{-}\rangle|C^{+}\rangle,
\end{align}
where
\begin{align}\label{eq:cat_even_odd}
    |C^{\pm}\rangle
    =
    \frac{1}{N_{\pm}}
    (|\alpha\rangle \pm |-\alpha\rangle),
\end{align}
with normalisation constants
\begin{align}\label{eq:N_pm}
    N_{\pm}
    =
    \sqrt{
    2(1\pm e^{-2|\alpha|^2})
    }.
\end{align}
In the large-$\alpha$ limit,
\begin{align}\label{eq:Np-Nm}
    \frac{N_+}{N_-}
    \approx
    \frac{N_-}{N_+}
    \approx
    1.
\end{align}
In the DRCC-1, the code vector consists, therefore, of a superposition of only two coherent states of opposite amplitude (Eq.\eqref{eq:cat_even_odd}). Mixing them on a 50-50 beam splitter would result in a coherent component twice as large in one mode and a vacuum in the other, which is not useful for QEC, as we also explicitly show in Appendix \ref{sec:DRCC-bs} for dephasing. Hence, unlike the construction in Eq.~\eqref{eq:m-RSB}, we do not explicitly include the BS in the definition of DRCC-1.
\begin{figure}
    \centering
    \includegraphics[width=1\columnwidth]{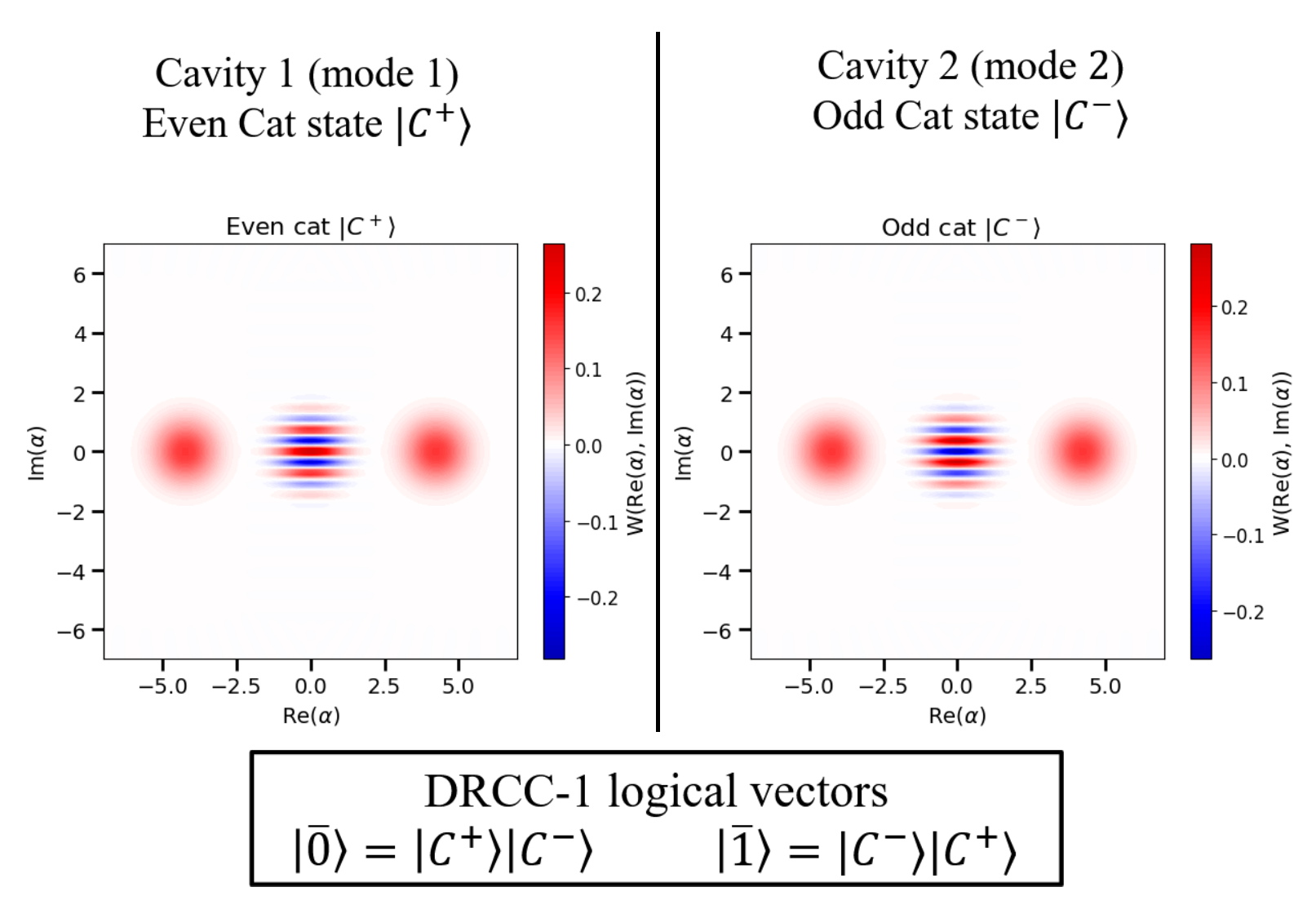}
   \caption{Wigner function of the even and odd cat states yielding the dual-rail cat code with $N=1$ (DRCC-1). The logical basis states are $|\bar{0}\rangle=|C^{+}\rangle|C^{-}\rangle$ (Eq.\eqref{eq:DRCC-10}) and $|\bar{1}\rangle=|C^{-}\rangle|C^{+}\rangle$ (Eq.\eqref{eq:DRCC-11}).}
    \label{fig:drcc_1}
\end{figure}

\subsection{Code vectors for DRCC-2}

Similarly, one can define the four-component dual-rail cat code corresponding to $N=2$ (DRCC-2). The code vectors for the DRCC-2 are given by 
\begin{align}\label{eq:DRCC-20}
    |0_4\rangle &= \hat{U}_{\rm BS}|C^{+}_4\rangle|C^{-}_4\rangle = U_{\rm BS}|\bar{0}_4\rangle,\\
\label{eq:DRCC-21}
    |1_4\rangle &= \hat{U}_{\rm BS}|C^{-}_4\rangle|C^{+}_4\rangle = U_{\rm BS}|\bar{1}_4\rangle,
\end{align}
where
\begin{align}
    |C^{\pm}_4\rangle
    &= \frac{1}{N_{4\pm}}
    \left(
    |\alpha\rangle + |-\alpha\rangle
    \pm (|j\alpha\rangle + |-j\alpha\rangle)
    \right),
\end{align}
with normalisation constants
\begin{align}
N_{4\pm }
= {\sqrt{\,4 + 4e^{-2|\alpha|^2} \pm 8e^{-|\alpha|^2}\cos(|\alpha|^2)\,}},
\end{align}
and
\begin{align}
|\bar{0}_4/\bar{1}_4\rangle
=
|C^{\pm}_4\rangle|C^{\mp}_4\rangle.
\end{align}
For DRCC-2, the presence of a beam splitter mixing the two modes can result in a QEC performance increase, as we will show in Section \ref{sec:drcc2} and in agreement with Ref.~\cite{multimode_RSB}.

\section{Knill--Laflamme Conditions and Noise Bias in the DRCC-1}
\label{sec:QEC_properties}

In this Section, we analyse the error-correction properties of DRCC-1 under photon-loss and dephasing noise. Photon loss is the dominant error channel in several bosonic hardware platforms, including superconducting cavities and microwave architectures~\cite{sc_qubit_adnoise}. The Kraus operator corresponding to $\ell$-photon loss in an $m$-mode system is
\begin{align}
    \hat{L}_{\ell}^{i}
    =
    \sqrt{
    \frac{(1-e^{-\kappa_i t})^{\ell}}{\ell!}
    }
    e^{-\kappa_i t \hat{n}_i/2}
    \hat{a}_i^{\ell},
\end{align}
where $\kappa_i t$ characterises the loss strength in the $i^{\rm th}$ mode. We also consider dephasing noise, whose Kraus operators are
\begin{align}
    \hat{D}_r^{i}
    =
    \sqrt{
    \frac{(\gamma_i t)^r}{r!}
    }
    e^{-\gamma_i t \hat{n}_i/2}
    \hat{n}_i^{r},
\end{align}
where $\gamma_i t$ denotes the dephasing strength in the $i^{\rm th}$ mode. As noted in Ref.~\cite{Albert_pair_cat}, $\hat{L}^i_0$ and the error operators $\{\hat{D}^i_r\}_{r=0}^{\infty}$ give rise to purely dephasing-type errors $\{\hat{n}^i_p\}_{p=0}^{\infty}$. In contrast, the operators $\{\hat{L}^i_{\ell}\}_{\ell \geq 1}$ describe $\ell$-photon loss from the $i$-th mode and are therefore referred to as loss errors.

To determine whether DRCC-1 corrects these errors, we check the Knill--Laflamme (KL) condition~\cite{KL-cond}. A code $\mathcal{C}$ with projector $\hat{P}$ corrects a set of errors $\{E_k\}$ if and only if
\begin{align}\label{eq:kl_cond}
    \hat{P}\hat{E}_k^{\dagger}\hat{E}_l\hat{P}
    =
    c_{kl}\hat{P},
\end{align}
for all $k,l$, where $c_{kl}$ forms a Hermitian matrix~\cite{chuang_nielsen}. In the following, we evaluate the projected operators $\hat{P}E_k^{\dagger}E_l\hat{P}$ for leading-order photon-loss and dephasing errors. If $\hat{P}\hat{E}_k^{\dagger}\hat{E}_l\hat{P}$ is  not proportional to $P$, it signifies a violation of the KL condition and  Eq.\eqref{eq:kl_cond} becomes 
\begin{align}\label{eq:non_KL}
   \hat{P}\hat{E}_k^{\dagger}\hat{E}_l\hat{P} & = c_{kl}\hat{P} + A_{X} \hat{X}_L+A_{Y} \hat{Y}_L+A_{Z} \hat{Z}_L,
\end{align}
where $\{\hat{X}_L,\hat{Y}_L,\hat{Z}_L\}$ are the logical Pauli operators, given for the DRCC-1 by
\begin{align}
    \hat{X}_L & = |\bar{0}\rangle\langle\bar{1}| + |\bar{1}\rangle\langle\bar{0}|\\
    \hat{Y}_L &=j|\bar{0}\rangle\langle\bar{1}| -j |\bar{1}\rangle\langle\bar{0}|\\
    \hat{Z}_L & = |\bar{0}\rangle\langle\bar{0}|- |\bar{1}\rangle\langle\bar{1}|.
\end{align}
 Furthermore, from Eq.\eqref{eq:non_KL}, we can infer the types of errors that are induced on the codespace by the physical noise process. For example, if $A_{X},A_{Y}$ are zero for some error operators $\{\hat{E}_k =\hat{I} , \hat{E}_l = \hat{O}\}$, we will conclude that the physical error $\hat{O}$ induces a logical $\hat Z$ error onto the logical qubit. If instead for some $\{\hat{E}_k =\hat I, \hat{E}_l = \hat{O}\}$, $A_{X}, A_{Y}, A_{Z}$ and $c_{kl}$ are zero, then we can conclude that the error $\hat{O}$ induces a detectable leakage which is also known as erasure~\cite{Yale_dual}. In the case where the dominant noise processes induce erasure errors,  we say that the code is biased towards the erasure error.

\subsection{Photon-loss errors}

We first analyse DRCC-1 under photon loss. 
We consider, in particular, single-photon loss errors. To leading order in the noise strengths $(\kappa_it)$, the relevant Kraus operators are

\begin{align}
   \hat{E}_0^L & = \hat{L}_0^1\hat{L}_0^2 \approx \hat I - \frac{1}{2}(\kappa_1t \hat{n}_1+\kappa_2t \hat{n}_2)\\
    \hat{E}_1^{L}
    &= \hat{L}_1^1\hat{L}_0^2\approx\sqrt{\kappa_1 t}\,\hat{a}_1,
    \\
    \hat{E}_2^{L}
    &= \hat{L}_0^1\hat{L}_1^2\approx\sqrt{\kappa_2 t}\,\hat{a}_2.
\end{align}

A single-mode two-component cat code cannot correct such first-order photon-loss errors because
\begin{align}
    \hat{a}|C^{\pm}\rangle
    =
    \alpha\frac{N_{\mp}}{N_{\pm}}
    |C^{\mp}\rangle.
\end{align}
Thus, a single-photon loss induces a logical bit-flip error in the single-mode cat-code basis. 

For DRCC-1, with projector
\begin{align}
   \hat{ \bar{P}}
    =
    |\bar{0}\rangle\langle\bar{0}|
    +
    |\bar{1}\rangle\langle\bar{1}|,
\end{align}
looking at $\hat E_k^{\dagger}\hat E_l$ gives us the following relations, to first order in $\kappa_it$: 
\begin{align}\label{eq:Pa_iP}
\hat{\bar{P}}\hat{a}_i
\hat{\bar{P}}&= 0\\
\label{eq:a1a2_w/obs}
    \hat{\bar{P}}\hat{a}_2^{\dagger}\hat{a}_1\hat{\bar{P}}
    &=
    A_X \hat{X}_L
    +
    jA_Y \hat{Y}_L,\\
    \label{eq:n1_w/obs}
    \hat{\bar{P}}\hat{a}_1^{\dagger}\hat{a}_1\hat{\bar{P}}
    &=
    A_0\hat{\bar{P}}
    +
    A_Z \hat{Z}_L,
    \\
\label{eq:n2_w/obs}
    \hat{\bar{P}}\hat{a}_2^{\dagger}\hat{a}_2\hat{\bar{P}}
    &=
    A_0\hat{\bar{P}}
    -
    A_Z \hat{Z_L},
\end{align}
where $A_{0},A_X,A_Y$ and $A_Z$ are given by
\begin{align}\label{eq:A_X}
    A_{0}= A_{X}
    &=
    \frac{|\alpha|^2}{2}
    \left(
    \frac{N_-^2}{N_+^2}
    +
    \frac{N_+^2}{N_-^2}
    \right) ,
    \\
\label{eq:A_Z}
    A_{Z}=A_{Y}
    &=
    \frac{|\alpha|^2}{2}
    \left(
    \frac{N_-^2}{N_+^2}
    -
    \frac{N_+^2}{N_-^2}
    \right).
\end{align}
In the large-$\alpha$ limit, $N_+/N_-\approx N_-/N_+\approx 1$
[see Eqs.~\eqref{eq:N_pm}--\eqref{eq:Np-Nm}]. Consequently,
$A_Z=A_Y\simeq-2|\alpha|^2e^{-2|\alpha|^2}\rightarrow 0$, whereas
$A_X\simeq|\alpha|^2$. Hence, from Eqs.~\eqref{eq:a1a2_w/obs}--\eqref{eq:n1_w/obs} we see that the error operators $\hat{a}_1,\hat{a}_2$ do not satisfy the KL condition even in the large $\alpha$ limit. Therefore, DRCC-1 cannot correct single-photon loss errors. 

 However, such errors are detectable by virtue of Eq.\eqref{eq:Pa_iP}. Indeed, a single-photon loss event induces detectable leakage into the codespace (see Fig.\ref{fig:DRCC_code_pic}, upper panel). In other words, unlike the single-mode cat code, where a single-photon loss error induces a logical $\hat{X}$ error,   with DRCC-1, the single-photon losses  map the logical states out of the DRCC-1 subspace into the even-parity leakage manifold  spanned by
\begin{align}\label{eq:cap_Phi}
    |\Phi^{\pm}\rangle
    =
    |C^{\pm}C^{\pm}\rangle.
\end{align}
This is analogous to the dual-rail code, where single-photon loss maps the code vectors onto a detectable leakage state, see  Fig.~\ref{fig:DRCC_code_pic}, lower panel (see also Sec.~\ref{se:general-comparison-with-DR} for a detailed comparison of the dual-rail code with DRCC-1). 
One can detect these leakage events using a joint parity check (JPC), as discussed in Ref.~\cite{Yale_dual}.
Although the JPC does not deterministically detect the cavity that suffered photon loss, it enables leakage detection and erasure flagging, which we later exploit when constructing concatenated QEC protocols in Sec.~\ref{sec:conc_code}.

For a detailed discussion on the higher-order photon loss, we refer to Appendix \ref{sec:noise_bias}.

\begin{figure}[t!]
    \centering
    \includegraphics[width=0.75\columnwidth]{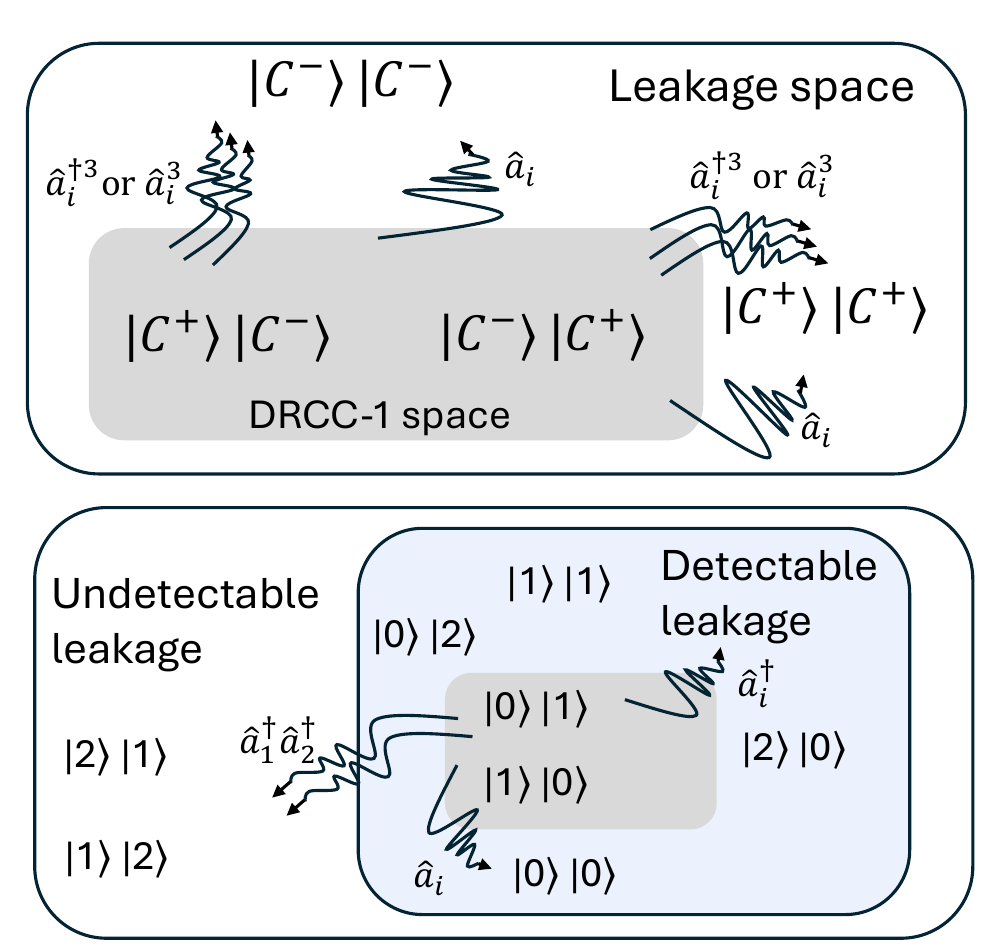}
    \caption{For DRCC-1 (upper panel), any single-photon loss or gain, as well as arbitrary higher-order photon-loss errors $\hat{a}_i^m$ or photon-gain errors ${\hat{a}_{i}^{\dagger m}}$ on a single mode $(m\geq 2)$, either map the logical states to the even-parity leakage subspace or leave them within the odd-parity codespace. For the conventional dual-rail code (lower panel), while single-photon loss or gain on either mode map the logical states to the even-parity leakage states $\{|00\rangle,|11\rangle,|02\rangle,|20\rangle\}$ (detectable as erasures), higher-order photon gain processes ${\hat{a}_{i}^{\dagger m}}$ ($m\geq 2$) can also populate other odd-parity states, such as $\{|12\rangle,|21\rangle\}$, leading to false-negative syndrome outcomes~\cite{Yale_dual}.   }
    \label{fig:DRCC_code_pic}
\end{figure}
\subsection{Dephasing errors}

Next, we analyse the dephasing channel. To first order in the noise strengths $(\gamma_it)$, the dephasing Kraus operators are 

\begin{align}\label{eq:d0}
    \hat{E}_0^{D}
    &= \hat{D}_1^0\hat{D}_2^0\approx
    I
    -
    \frac{t}{2}
    (\gamma_1 \hat{n}_1+\gamma_2 \hat{n}_2),
    \\
\label{eq:d1}    \hat{E}_1^{D}
    &=\hat{D}_1^1\hat{D}_2^0\approx
    \sqrt{\gamma_1 t}\,\hat{n}_1,
    \\
\label{eq:d2}    \hat{E}_2^{D}
    &=\hat{D}_1^0\hat{D}_2^1\approx
    \sqrt{\gamma_2 t}\,\hat{n}_2.
\end{align}

Therefore, analysing the effect of $\hat{E}^{D\dagger}_k\hat{E}_\ell^{D}$ onto the codespace up to the first order of the noise strength  $\gamma_it$, we obtain the following
\begin{align}
   \hat{ \bar{P}}\hat{n}_i\hat{\bar{P}}
    &=
    A_0\hat{\bar{P}}
    +
    (-1)^{1 \oplus i}A_Z\hat{Z}_L,\\
    \hat{\bar{P}}\hat{n}^2_i\hat{\bar{P}}
    &=
    (|\alpha|^4+A_0)\hat{\bar{P}}
    +
    (-1)^{1 \oplus i}A_Z\hat{Z}_L,\\
     \hat{\bar{P}}\hat{n}_2\hat{n}_1\hat{\bar{P}}& = |\alpha|^4\hat{\bar{P}}.
\end{align}
All the above equations show that the first-order dephasing errors do not satisfy the KL condition in the low-$\alpha$ limit.
We also note that first-order dephasing error on either cavity induces a logical $\hat{Z}$ error to the DRCC-1 qubits in the low $\alpha$ limit. However, as $
\alpha$ increases, the logical errors $\hat{Z}$ are being suppressed exponentially due to Eq.\eqref{eq:A_Z}.  

For the QEC properties under higher-order dephasing, see Appendix~\ref{sec:noise_bias}. 

\section{Universal quantum operations for the DRCC-1}
\label{sec:logical gates}

Achieving a universal set of quantum operations requires arbitrary single-qubit rotations implemented alongside at least one entangling two-qubit gate~\cite{DiVincenzo_2000}. Having previously characterised the noise properties and error-correction capabilities of the DRCC-1 under photon-loss and dephasing channels, we now turn to the construction of logical operations within the codespace defined by Eqs.~\eqref{eq:DRCC-10}--\eqref{eq:DRCC-11}. In this Section, we present the implementation of logical gates, analyse their noise-bias properties, and establish measurement protocols tailored for the DRCC-1. These implementations rely on fundamental circuit elements, with particular emphasis on beam-splitter interactions, number-operator couplings, and parity-sensitive measurements.

To construct the logical gates, we adopt the framework of Zeno dynamics restricted to the codespace, as outlined in Ref.~\cite{vv_albert_njp} for single-mode cat codes and similarly formalised in Refs.~\cite{Albert_pair_cat, PRX_mazyar_mirrahimi}. The central principle of Zeno-dynamics-assisted gate construction relies on evaluating whether the perturbative projection of an interaction Hamiltonian $\epsilon \hat H_I$, with $\epsilon$ being the strength of the perturbation,  onto the codespace, given by $\hat {\bar P} \hat H_I \hat {\bar P}$, generates a non-trivial logical operation of the form
\begin{align}
    \hat {\bar P} \hat H_I \hat {\bar P} = A_0 \hat {\bar P} + \sum_i A_i \hat {\Sigma}_i,  
\end{align}
where $\hat \Sigma_i \in \{\hat X_L,\hat Y_L,\hat Z_L\}$ and $A_i \in \{A_X, A_Y,A_Z\}$.  For example, consider the interaction Hamiltonian that comprises terms such as $\hat H_I \sim (\hat O_i^{\dagger}\hat O_j + \text{h.c.})$. The constituent operators $\hat O_i\hat O_j$ must violate the Knill-Laflamme condition; otherwise, $\hat H_I$ acts trivially as an identity operation on the codespace, failing to induce logical rotations.

However, the projected Hamiltonian $\hat{\bar P}\hat H_I\hat{\bar P}$ only describes the effective logical action within the codespace. The full evolution generated by $\hat{H}_I$, namely $\hat{U}_I=\exp(-i\hat{H}_I t)$, generally maps the logical vectors outside the codespace, leading to leakage during the gate evolution~\cite{vv_albert_njp}. 
Therefore, an additional mechanism is required to suppress these unwanted transitions and confine the dynamics to the codespace. Following the dissipative gate-construction framework of Refs.~\cite{Albert_pair_cat, PRX_mazyar_mirrahimi}, we achieve this confinement through continuous engineered dissipation, resulting in the Lindblad dynamics
\begin{align}
\hat {\dot{\rho}}
=-j \epsilon [\hat H_I,\hat \rho]+d_s \cD(\hat \rho),
\end{align}
where $d_s$ denotes the stabilisation rate and is assumed to satisfy $d_s \gg \epsilon$. In this regime, the dissipative dynamics continuously project the system back onto the codespace, thereby suppressing leakage through the quantum Zeno effect. For a more detailed discussion of dissipative Zeno dynamics, see Refs.~\cite{zeno_zanardi, Zeno_facci, symm_albert, uni_cont_noise, Albert_pair_cat, PRX_mazyar_mirrahimi}. We now discuss the dissipative stabilisation of the DRCC-1. After that, we will show how to engineer the logical gates utilising the Zeno dynamics discussed above, and finally discuss the bias-preserving properties of the gates.

\subsection{Dissipative stabilisation of DRCC-1}
\label{sec:dissipator}

 The dissipative stabilisation of bosonic codes has been extensively studied in the context of two-component cat codes~\cite{xu2023autonomous, BP_gates_liang_jiang_pair_cat, Puri_2020,vv_albert_njp}. Related developments include bias-preserving gates and autonomous stabilisation of squeezed cat codes~\cite{xu2023autonomous,hillmann2023quantum}, as well as dynamically protected logical operations in cat manifolds~\cite{PRX_mazyar_mirrahimi}. Here, we introduce a dissipative-stabilisation scheme for the DRCC-1.

The logical basis states of the DRCC-1 are simultaneous eigenstates of the operators $\hat a_i^2$ ($i\in\{1,2\}$), satisfying
\begin{align}\label{eq:cat_eig_state}
\begin{split}
     \hat{a}_i^{2}|C^+C^-\rangle
    &=
    \alpha^2 |C^+C^-\rangle,\\
    \hat{a}_i^{2}|C^-C^+\rangle
    &=
    \alpha^2 |C^-C^+\rangle.
\end{split}
\end{align}
Based on the above observation, we note that  the dissipator
\begin{align}\label{eq:disp_dyn_2}
    \mathcal{\hat{D}}[\hat{a}_1^2\hat{a}_2^2-\alpha^4],
\end{align}
stabilises the manifold 
\begin{align}\label{eq:cat_manifold}
\{
|C^{+}_{\alpha_1}C^{+}_{\alpha_2}\rangle,
|C^{-}_{\alpha_1}C^{-}_{\alpha_2}\rangle,
|C^{+}_{\alpha_1}C^{-}_{\alpha_2}\rangle,
|C^{-}_{\alpha_1}C^{+}_{\alpha_2}\rangle
\},
\end{align}
provided that $\alpha_1^2\alpha_2^2=\alpha^4$. Here, $\hat{\cD}[\hat A]$ denotes the Lindblad dissipator acting on a density matrix $\hat \rho$ as
\[
\hat \cD[\hat A](\hat \rho)
=
\hat A \hat \rho \hat A^{\dagger}
-\frac{1}{2}
\{\hat A^{\dagger}\hat A,\hat \rho\},
\]
with $\hat A=\hat a_1^2\hat a_2^2-\alpha^4$.

The same dissipator also stabilises the pair-cat code, as shown in Ref.~\cite{Albert_pair_cat}. To understand how it can stabilise the DRCC-1, we analyse the kernel of the corresponding jump operator. Any state $|\xi\rangle$ stabilized by the dissipation, also called a  dark state \cite{hillmann2023quantum}, must satisfy
\begin{align}\label{eq:diss_2_op}
    (\hat a_1^2 \hat a_2^2-\alpha^4)|\xi\rangle=0.
\end{align}
Expanding the state in the two-mode Fock basis,
\begin{align}
    |\xi\rangle
    =
    \sum_{m,n}
    C_{mn}|m\rangle|n\rangle,
\end{align}
and projecting Eq.~\eqref{eq:diss_2_op} onto the Fock basis yields the recurrence relation
\begin{align}\label{eq:reccurence}
    C_{m+2,n+2}
    =
    \frac{\alpha^4}
    {\sqrt{(m+1)(m+2)(n+1)(n+2)}}
    C_{m,n}.
\end{align}

Equation~\eqref{eq:reccurence} admits multiple classes of solutions. One family corresponds to the pair-cat code~\cite{Albert_pair_cat}. Another class of solutions  results in entangled states \footnote{\label{foot_1} Another class of solutions is obtained by choosing $C_{00}=C_{11}=0$, which leads to
\begin{align}\label{eq:new_dual}
    |\xi\rangle
    =
    \begin{cases}
        \displaystyle
        \sum_n
       h_{kn}
        |2n,2n+1\rangle,
        & C_{01}\neq 0,\; C_{10}=0,\\[2ex]
        \displaystyle
        \sum_n
        h_{kn}
        |2n+1,2n\rangle,
        & C_{01}=0,\; C_{10}\neq 0,
    \end{cases}
\end{align}
where $\hat H_{kn} =  \frac{\alpha^{4k}}
        {\sqrt{(2n)!(2n+1)!}}$. These states are distinct from the logical states of DRCC-1. Apart from one cavity containing even numbers, the other has odd numbers of photons; the DRCC-1, or, in general, multi-mode RSB code (Eq.~\eqref{eq:m-RSB}), exhibits no correlation between photon numbers in different modes, such that $n$ and $m$ in Eqs.~\eqref{eq:m-RSB} and \eqref{eq:DRCC-10} are independent. In contrast, the states in Eq.~\eqref{eq:new_dual} satisfy the constraint $m=2n+1$.}.
In order to obtain the logical state $|\bar 0\rangle$ of DRCC-1, we impose the following conditions to the coefficients of  Eq.~\eqref{eq:reccurence}:
\begin{itemize}
    \item[(i)] The coefficients satisfy
    \[
    C_{00}=C_{11}=0,
    \qquad
    C_{10}\neq0,
    \qquad
    C_{01}\neq0.
    \]

    \item[(ii)] Restricting the first mode to the even-parity sector and the second mode to the odd-parity sector reduces Eq.~\eqref{eq:reccurence} to
    \begin{align}\label{eq:mod_recc}
        C_{2m+2,2n+3}
        =
        \frac{\alpha^4\,C_{2m,2n+1}}
        {\sqrt{(2m+1)(2m+2)(2n+2)(2n+3)}}.
    \end{align}

    \item[(iii)] Together with the boundary conditions
    \begin{align*}
    \begin{split}
         C_{0,m}
        &\propto
        \frac{\alpha^m}{m!},\\
         C_{n,1}
        &\propto
        \frac{\alpha^{n+1}}{(n+1)!},
    \end{split}
    \end{align*}
    Eqs.~\eqref{eq:reccurence} and \eqref{eq:mod_recc} uniquely generate the logical state $|\bar 0\rangle$ defined in Eq.~\eqref{eq:DRCC-10}.
\end{itemize}

Similarly, the logical state $|\bar 1\rangle$ is obtained by exchanging the parity sectors of the two modes, such that the first mode contains only odd photon numbers and the second mode only even photon numbers. The corresponding boundary conditions are
\begin{align*}
\begin{split}
     C_{0,m}
    &\propto
    \frac{\alpha^{m+1}}{(m+1)!},\\
     C_{n,1}
    &\propto
    \frac{\alpha^n}{n!}.
\end{split}
\end{align*}

We therefore conclude that the DRCC-1 logical manifold can be realised as a particular dark subspace of the dissipator in Eq.~\eqref{eq:disp_dyn_2}, which will be adopted throughout the logical-gate constructions presented in the following Sections. Physical implementations of this dissipator can be realised using asymmetrically threaded SQUID architectures~\cite{maryam_diss}. In Appendix~\ref{sec:alt_diss}  we discuss alternative dissipative-stabilisation schemes for the DRCC-1. 

Having identified relevant stabilisation mechanisms, we can now discuss gate implementation.

\subsection{Single-qubit gates}\label{sec:1q-gates}

Here, we show that an arbitrary single-qubit rotation can be engineered using beam-splitter transformations, as for the standard dual-rail qubit Eq.~\eqref{eq:dual_rail_stdrd} discussed in Ref.~\cite{Yale_dual}. We first consider the Hamiltonian
\begin{align}\label{eq:gen_H}
     \hat{H}= K_{\rm BS}(\hat{a}_2\hat{a}_1^{\dagger}
    +
    \hat{a}_2^{\dagger}\hat{a}_1),
\end{align}
where $K_{\rm BS}$ is the strength of the beam-splitter interaction. The beam-splitter unitary $\hat U_{\BS}(\delta,\phi)$ in Eq.~\eqref{eq:bs} transforms the Hamiltonian in Eq.~\eqref{eq:gen_H} as
\begin{align}
\label{eq:BS-hamiltonian}
     \hat{ U}_{\BS}(\delta,\phi) \hat{H} \hat{U}_{\BS}(\delta,\phi)^{\dagger}& \equiv  \hat{\tilde{H}}(\delta,\phi)\\
    &=
   K_{\rm BS} \bigg[(\hat{a}_2^{\dagger}\hat{a}_1A_{\delta,\phi}+h.c.)
    \nonumber\\
    &\quad+
    (\hat{n}_2-\hat{n}_1)\sin2\delta\cos\phi\bigg],
\end{align}
where the coefficients $A_{\delta,\phi}$ depend on the beam-splitter parameters $(\delta,\phi)$ as follows: 
\begin{align}
    A_{\delta,\phi}&  = \cos^2\delta - e^{- 2j \phi}\sin^2\delta.
\end{align}

Projecting this Hamiltonian onto the DRCC-1 subspace yields
\begin{align}\label{eq:gen_gate}
    \hat{\bar{P}}\hat{\tilde{H}}(\delta,\phi)\hat{\bar{P}}
    &=
    2K_{\rm BS}(A_X{\rm Re}(A_{\delta,\phi})\hat{X}_L
   +
   A_Y {\rm Im}(A_{\delta,\phi})\hat{Y}_L
    \nonumber\\
    &\quad-
    A_Z\sin2\delta\cos\phi\, \hat{Z}_L).
\end{align}
Therefore, by tuning the beam-splitter parameters $(\delta,\phi)$, one can generate rotations about different axes on the logical Bloch sphere, as summarised in Table~\ref{tab:rot_gate_tab}.

\begin{table}[t!]
    \centering
    \begin{tabular}{c|c|c}
    \midrule
       Rotation gate  &  $\delta$ & $\phi$\\
       \midrule
       $\hat{R}_X$  & $\delta = 0$  & $\phi = \pi/2$\\
       $\hat{R}_Y$  & $\delta = \pi/2$  & $\phi = \pi/4$\\
       $\hat{R}_Z$  & $\delta = \pi/4$  & $\phi = 0$\\
       \midrule
    \end{tabular}
    \caption{Different beam-splitter angles of the unitary $\hat{U}_{\text{BS}}(\delta,\phi)$ generate different logical gates, corresponding to rotations about different axes on the DRCC-1 Bloch sphere as per Eq.(\ref{eq:gen_gate}). Here, $\hat{R}_X,\hat{R}_Y,\hat{R}_Z$ are the rotations around the logical $\hat{X},\hat{Y}$, and $\hat{Z}$ axes, respectively. }
     \label{tab:rot_gate_tab}
\end{table}
\subsubsection{Logical $X$ and $Z$-rotations}

For $\delta=0$ and $\phi= \pi/2$, the beam-splitter interaction acts as the identity and we retrieve Eq.(\ref{eq:gen_H}),
\begin{align}\label{eq:HX}
 \hat{H}_X
 &   =
    K_{\BS}
    \left(
    \hat{a}_2\hat{a}_1^{\dagger}
    +
    \hat{a}_2^{\dagger}\hat{a}_1
    \right).
\end{align} 
 Projecting Eq.~\eqref{eq:HX} onto the DRCC-1 subspace gives
\begin{align}
\label{eq:BS-Rx}
 \hat{\bar{H}}_X \equiv   \hat{\bar{P}}\hat{H}_X\hat{\bar{P}}
    =
    2K_{\BS}A_X\hat{X}_L.
\end{align}
Thus, this Hamiltonian generates a logical $\hat{X}$-axis rotation gate $\hat{R}_X(\theta_X) = \exp({-j \hat{\bar{H}}_X t}) = \exp({-j \theta \hat{X}_L})$, with 
\begin{align*}
    \theta_X = 2K_{\rm BS}A_Xt,
\end{align*}
and $t$ the interaction time.

A similar interaction has previously been employed to generate logical entangling operations for single-mode cat codes~\cite{Puri_2020,vv_albert_njp}.

Analogously, fixing $(\delta,\phi)$ in Eq.~\eqref{eq:gen_gate} to $(\pi/4,0)$ gives the Hamiltonian
\begin{align}\label{eq:H_Z}
   \hat{ H}_Z & = K_{\BS}(\hat{n}_1-\hat{n}_2).
\end{align} 
Using Eqs.~\eqref{eq:n1_w/obs}--\eqref{eq:n2_w/obs}, the projection of $\hat H_Z$ onto the DRCC space gives
\begin{align}\label{eq:BS_RZ}
  \hat{ \bar{H}}_Z \equiv  \hat{\bar{P}}\hat{H}_Z\hat{\bar{P}}
    &=
    K_{\BS}\hat{\bar{P}}(\hat{n}_1-\hat{n}_2)\hat{\bar{P}}
    \nonumber\\
    &=
    2K_{\BS}A_Z\hat{Z}_L.
\end{align}
Therefore, $\hat{H}_Z$ generates a logical $\hat{Z}$ rotation  gate $\hat R_Z(\theta_Z) = \exp({- j \hat{\bar{H}}_Z}t)= \exp(- j \theta_Z \hat{Z}_L)$ on the DRCC-1 Bloch sphere, with
\begin{align}
    \theta_Z & = 2A_ZK_{\rm BS}t .
\end{align}
We will later use this interaction for designing the joint parity measurement.
Ref.~\cite{Albert_pair_cat} uses a similar Hamiltonian in the context of syndrome measurement for pair-cat codes.  

The BS-based implementation, however, is not the only way to realise the logical $\hat R_X$ gate for the DRCC-1. An alternative approach is to engineer the logical $\hat{X}$ gate using the following two-mode squeezing (TMS) Hamiltonian, analogous to that used for the pair-cat code~\cite{Albert_pair_cat}:
\begin{align*}
\hat{H}_X^{\rm alt}
&=
J(\hat{a}_1\hat{a}_2 + h.c.),
\end{align*}
whose projection onto the DRCC-1 subspace yields
\begin{align*}
\hat{\bar{P}}\hat{H}_X^{\rm alt}\hat{\bar{P}}
=
2{\rm Re}(\alpha^2)J\hat{X}_L.
\end{align*}

Despite providing an alternative implementation of $\hat R_X(\theta)$, the TMS-based approach requires higher-order non-linear interactions. Consequently, it is generally more challenging to engineer than the BS-based $\hat R_X$ gate generated by the Hamiltonian in Eq.~(\ref{eq:BS-Rx}), which can be implemented using passive linear optical elements (see, e.g., Chapter~7 of Ref.~\cite{Gerry_Knight_2004}). We note that the $\hat R_X$ gate is trivially bias-preserving.

Finally, we note that although the order of rotational symmetry for DRCC-1, $N=1$, is odd, some of the gate constructions developed for even-$N$ two-mode RSB codes in Ref.~\cite{multimode_RSB} still apply to DRCC-1. In particular, the construction of the logical operator $\hat{Z}_L$ remains valid, as shown in Appendix~\ref{sec:proof_gate} (Lemma~\ref{lem:ZL_alt}). However, whereas the method in Ref.~\cite{multimode_RSB} requires gate teleportation to implement non-Clifford gates, our gate constructions are teleportation-free. 

{ In Appendix~\ref{sec:leak_app}, we quantify the off-codespace contribution for the logical gate Hamiltonians considered using the Hilbert--Schmidt norm. A high off-codespace contribution is associated with a slow gate time, and this is the case for the $\hat Z_L$ gate that we have introduced in Eq.\eqref{eq:H_Z}. A faster implementation of the $\hat Z_L$ gate (albeit not general $\hat R_Z(\theta)$ rotations) can be achieved by using the methods in Appendix \ref{sec:proof_gate}. A further fast implementation of $\hat R_Z(\theta)$ gates can also be achieved by using the parity check operation from  Ref.\cite{Puri_2020} together with an auxiliary qubit, as we discuss in Appendix \ref{sec:fast_Rz}. See also the extended discussion in the concluding section.}


\subsubsection{Erasure-bias preserving properties of single qubit gates}\label{sec:bpg_1q}

Here we discuss the effect of photon-loss errors during interaction $\hat{\tilde H}$ in Eq.\eqref{eq:BS-hamiltonian}. If a single-photon loss error occurs in one of the cavities during or before the beam-splitter operation, the beam splitter transforms the mode operators according to Eq.~\eqref{eq:BS_transformation}. After this transformation, the error remains a detectable leakage error since
\begin{align}
    \hat{\bar{P}} \hat{b}_i\hat{\bar{P}}& = 0. 
\end{align}
In other words, all detectable leakages, i.e., erasures, before the action of the BS, persist as erasures. 
Therefore, gates generated through BS operations preserve the erasure bias. 

\subsection{Two-qubit gates}\label{sec:two_qubit_gate}

To achieve universal quantum computation, along with the single-qubit gate, we need a two-qubit gate~\cite{DiVincenzo_2000}. Here, we discuss two relevant classes of two-qubit operations for DRCC-1: the controlled-$X$ gate ($CX$) and an $XX$-type entangling interaction.

We first consider the $CX$ gate, implemented via a controlled beam splitter (CBS) interaction, controlled on the cavity mode. The corresponding Hamiltonian is given by
\begin{align}\label{eq:H_cx}
    \hat{H}_{\rm CX}
    \approx
    g\,\hat{n}_2
    (\hat{a}_4\hat{a}_3^{\dagger}+\hat{a}_4^{\dagger}\hat{a}_3).
\end{align}

Therefore the action of $\hat H_{\rm CX}$ on the codespace of DRCC-1 is given by 
\begin{align}
    \hat{\bar{P}} \hat H_{\rm CX} \hat{\bar{P}} & = 2g (A_0\hat{\bar{P}} - A_Z \hat Z_L ) \otimes (A_X \hat X_L).
\end{align}

In the low $\alpha$ limit,  $\exp(- j  \hat{\bar{P}} \hat H_{\rm CX} \hat{\bar{P}} )$ hence generates a controlled-$\hat R_X$ gate up to a single qubit $\hat R_Z$ gate on the controlling mode.

Alternatively, universal computation can be achieved using entangling $XX$ interactions. For DRCC-1, we consider
\begin{align}\label{eq:H_XX}
    \hat{H}_{XX}
    =
    \chi
    (\hat{a}_1\hat{a}_2^{\dagger}+h.c.)
    \otimes
    (\hat{a}_3\hat{a}_4^{\dagger}+h.c.).
\end{align}
 Projecting onto the logical subspace gives 

\begin{align}
   \hat{\bar{H}}_{XX}= \hat{\bar{P}}^{\otimes2}\hat{H}_{XX}\hat{\bar{P}}^{\otimes2}&=4\chi\,A_X^2
    (\hat{X}_L\otimes \hat{X}_L).
\end{align}
Choosing $ t_{XX} = \frac{\pi}{16A_X^2\chi}$, the gate
\begin{align}\label{eq:XX_gate}
     X X (\theta)=e^{- j \hat{\bar{H}}_{XX}t_{XX}}
\end{align}
transforms the product state $|\bar{0}\bar{0}\rangle$ into the maximally entangled state
\begin{align*}
    \frac{|\bar{0}\bar{0}\rangle-j|\bar{1}\bar{1}\rangle
    }{\sqrt{2}}.
\end{align*}
 Similar two-qubit gates have been studied previously for pair-cat codes, albeit with distinct interactions.~\cite{Albert_pair_cat}.

The entangling interaction $\hat H_{XX}$ can be generated using the cross-Kerr interaction as follows 
\begin{align}
\hat{ \tilde H}_{XX}(\delta,\phi)
    & = \chi\hat{U}_{BS}^{\dagger}(\delta,\phi)
    (\hat{n}_1\otimes \hat{n}_3)
    \hat{U}_{BS}(\delta,\phi).
\end{align}
For $\delta=\pi/4$ and $\phi=0$, the transformed cross-Kerr interaction contains the desired $\hat{X}_L\otimes\hat{X}_L$ coupling after projection onto the DRCC-1 codespace.

Interactions of the form $\hat{H}\sim \hat{n}_1\hat{n}_3$ are instead useful for implementing logical $CZ$ gates in RSB encodings~\cite{RSB_grimsmo_PRX,multimode_RSB}. Similarly, quartic interactions of the form
\begin{align*}
  \hat H^{\rm alt}_{XX}
    \sim
    \hat{a}_1\hat{a}_2\hat{a}_3\hat{a}_4+h.c.
\end{align*}
generate effective $XX$ interactions for DRCC-1 qubits. We note that $\hat{ \tilde H}_{XX}(\delta,\phi)$ also generates the $\hat X\hat X$ interaction with pair-cat codes~\cite{Albert_pair_cat}. 

\subsubsection{Erasure-bias preserving properties}

We now discuss the noise-bias properties of the proposed two-qubit gates with respect to photon loss. We find that both the $CX$ gate generated by $\hat H_{CX}$ in Eq.~\eqref{eq:H_cx} and the $\hat X\hat X(\theta)$ gate generated by $\hat H_{XX}$ in Eq.~\eqref{eq:H_XX} preserve the erasure bias.

We first consider the $CX$ gate. The Hamiltonian in Eq.~\eqref{eq:H_cx} inherently preserves the erasure bias because it consists of a number operator acting on the control mode and a beam-splitter interaction acting on the target modes. Consequently, any leakage into the target modes occurring before or during the gate operation persists afterwards. Similarly, leakage events in the control mode remain leakage events after the gate operation since $[\hat a,\hat n]=\hat a$. We also note from \cite{multimode_RSB} that the $CZ$ gate for the DRCC-1 can be generated via $\hat H_{CZ} \sim \hat n_1 \hat n_3$. {Following the leakage-bias-preserving argument for the $CX$ gate, we conclude that the $CZ$ gate  preserves the leakage bias too, irrespective of whether the leakage occurs on the target or the control logical qubit.}

Next, we consider the gate generated by $\hat H_{XX}$ in Eq.~\eqref{eq:H_XX}. Since $\hat H_{XX}$ is entirely based on beam-splitter interactions, the argument follows directly from the discussion of erasure-bias preservation in Sec.~\ref{sec:bpg_1q}. We therefore conclude that the $\hat X\hat X(\theta)$ gate in Eq.~\eqref{eq:XX_gate} also preserves the erasure bias.

Since erasures occurring before or during either two-qubit gate persist throughout the evolution, they can subsequently be corrected using an erasure-correction protocol such as that of Ref.~\cite{Yale_dual}, or through a standard Pauli-correction procedure implemented with an outer repetition code, as discussed in Sec.~\ref{sec:conc_code}.
\subsection{Geometric phase}\label{sec:geom_phase}

During logical gate operations, cat-code vector states generally evolve along trajectories in phase space and accumulate phases associated with this evolution. One contribution is the geometric phase, which depends only on the path traversed in phase space and not on the rate at which the path is followed. For closed trajectories, this geometric phase reduces to the Berry phase~\cite{berry_phase}. The role of geometric phases in logical gate operations has been extensively studied for single-mode cat qubits, both in dissipatively stabilised and Kerr-cat architectures~\cite{PRX_mazyar_mirrahimi, Puri_2020}. In multimode cat codes, geometric phases have also been exploited to engineer logical gates; for example, Ref.~\cite{Albert_pair_cat} uses this mechanism to realise logical operations in the pair-cat code. In this Section, we first review geometric-phase accumulation in the single-mode cat code and then show that the logical states of any general dual-rail RSB code~\cite{multimode_RSB}, including the DDRC-1 and the DRCC-2, do not acquire a relative geometric phase under the gate operations considered here.

To illustrate the origin of the geometric phase, we consider a logical $Z$-rotation of a single-mode cat qubit generated by the unitary operator
\begin{align}\label{eq:time_ev_U}
\hat U_{Z}(\phi)& =  e^{j\hat n\phi(t)} ,
\end{align}
under which a coherent state evolves as
\[
|\alpha(t)\rangle
=
|\alpha_0 e^{j\phi(t)}\rangle.
\]
When $\phi(t)=\pi$, the evolution in Eq.~\eqref{eq:time_ev_U} implements a logical $Z$-rotation on the cat-qubit Bloch sphere. Starting from the logical state
\[
|\psi_c(0)\rangle
=
c_0|C^+\rangle
+
c_1|C^-\rangle,
\]
the evolution produces
\begin{align*}
|\psi_c(0)\rangle
\rightarrow
|\psi(t)\rangle
=
c_0 e^{j\Phi_s^+}
|C^+_{\alpha(t)}\rangle
+
c_1 e^{j\Phi_s^-}
|C^-_{\alpha(t)}\rangle,
\end{align*}
where $\Phi_s^\pm$ denote the geometric phases accumulated by the states $|C^\pm\rangle$. The geometric-phase $\Phi_s^\pm$ is given by 
\begin{align}\label{eq:GP_def}
\Phi_s^\pm &=j \int \ dt
\left\langle
C^{\pm}_{\alpha_0 e^{j\phi(t)}}
\left|
\frac{d}{dt}
\right|
C^{\pm}_{\alpha_0 e^{j\phi(t)}}
\right\rangle \nonumber \\
&=
j \int \ dt\left\langle
C^{\pm}_{\alpha_0 e^{j\phi(t)}}
\left|
j\dot{\phi}\,\hat n
\right|
C^{\pm}_{\alpha_0 e^{j\phi(t)}}
\right\rangle
\nonumber\\
&=
- \int \ d\phi
\left\langle
C^{\pm}_{\alpha_0 e^{j\phi(t)}}
\right|
\hat n
\left|
C^{\pm}_{\alpha_0 e^{j\phi(t)}}
\right\rangle .
\end{align}

Equation~\eqref{eq:GP_def} shows that the geometric phase accumulated along the trajectory is proportional to the mean photon number of the corresponding cat state. Since the even and odd cat states have different mean photon numbers,
\[
\langle \hat n\rangle_+
\neq
\langle \hat n\rangle_-,
\]
they accumulate different geometric phases during the evolution. Consequently, the logical basis states acquire a relative phase
\begin{align}\label{eq:relt_GP_phase}
\Delta\Phi
=
\Phi_s^+
-
\Phi_s^-
=
-\int d\phi\,
\left(
\langle \hat n\rangle_+
-
\langle \hat n\rangle_-
\right).
\end{align}

This relative geometric phase can be exploited to implement logical $Z$-rotations in single-mode cat qubits. For example, using the relative GP, we can design a universal gate set~\cite{Albert_2016}. However, the same mechanism can also contribute to an unintended relative phase during gate operations that involve motion in phase space, as noted in Refs.~\cite{Puri_2020, PRX_mazyar_mirrahimi}. In such cases, the accumulated geometric phase acts as an additional logical $Z$-rotation and must be accounted for or compensated during gate design~\cite{Puri_2020, PRX_mazyar_mirrahimi}.

The same consideration applies to the repetition-cat code. For the repetition cat-code vectors $\{|C^{+}C^{+}\rangle, |C^{-}C^{-}\rangle\}$, the evolution of the state $c_0 |C^{+}C^{+}\rangle+c_1 |C^{-}C^{-}\rangle$ in phase space accumulates a relative phase since
\begin{align*}
    \Delta \Phi &= \langle C^{+}C^{+}|\bar{n}|C^{+}C^{+}\rangle - \langle C^{-}C^{-}|\bar{n}|C^{-}C^{-}\rangle \nonumber\\
    & = 2|\alpha|^2A_Z.
\end{align*}
Thus, the gate Hamiltonians for the two-mode repetition cat code $\{|C^{+}C^{+}\rangle, |C^{-}C^{-}\rangle\}$ require an additional term to cancel the appearance of $\Delta\Phi$. In this regard, the code $\{|C^{+}C^{+}\rangle, |C^{-}C^{-}\rangle\}$ behaves similarly to the single-mode cat code with respect to GP accumulation.

In contrast, the DRCC cat-code vectors have identical mean photon numbers,
\begin{align}
\langle\bar{0}|\hat{\bar{n}}|\bar{0}\rangle = 2A_0 = \langle\bar{1}|\hat{\bar{n}}|\bar{1}\rangle,
\end{align}
where $\bar{n}$ is $n_1+n_2$ and $A_0$ is given by Eq.~\eqref{eq:A_X}. Consequently, any gate operation that generates a rotation of the logical states $\{|\bar{0}\rangle,|\bar{1}\rangle\}$ in phase space causes both logical vectors to accumulate the same GP. If we use a coherent-drive-induced $X$-axis rotation, it evolves an arbitrary logical state as 
\[
|\psi(t=0)\rangle
=
c_0 |C^{+}C^-\rangle
+
c_1|C^{-}C^+\rangle
\]
to
\begin{align}
    |\psi_t\rangle= c_0 e^{j \Psi_d^+}|C^+_{\alpha(t)}C^-_{\alpha(t)}\rangle + e^{j \Psi_d^+}c_1 |C^-_{\alpha(t)}C^+_{\alpha(t)}\rangle,
\end{align}
where $\Psi_d^+$ is a global phase with $\Psi_d^+ \propto |\alpha|^2A_0$. Thus, the evolved logical state under the $\hat R_X$ gate operation does not acquire any relative phase. Unlike the single-mode and repetition cat codes, the DRCC does not require additional correction terms to compensate for geometric-phase-induced relative phase accumulation.

The absence of relative GP accumulation is not unique to the DRCC-1. From the description of the general dual-rail RSB code in Eq.~\eqref{eq:m-RSB}, we note that the photon-number expectation value is identical for the two code vectors,
\begin{align}
    \langle 0_L|\bar{n}|0_L\rangle=\langle 1_L|\bar{n}|1_L\rangle = \sum_{mn}|f_{mn}|^2(2n+2m+1)N.
\end{align}

Thus, in general, none of the dual-rail RSB codes, including the DRCC-2, is affected by relative geometric phases during logical gate operations.
\subsection{Measurements}
\label{sec:meas}

A universal quantum computing architecture additionally requires state preparation and measurement capabilities~\cite{DiVincenzo_2000}. Having discussed logical gate operations for DRCC-1, we now describe protocols for logical readout and syndrome extraction. Here, we discuss two types of measurements: the joint-parity-check measurement, which we use for syndrome measurement, and the computational-basis measurement.

\subsubsection{Joint parity check}\label{sec:JPC}
For the syndrome detection of leakage errors with the standard dual-rail code, Ref.\cite{Yale_dual} uses a joint parity check scheme with an auxiliary (transmon) qubit and the following unitary operation: 
\begin{align}\label{eq:jpc_yale}
     \hat{U}_{\rm PC}
    &=\hat{I}\otimes|g\rangle\langle g|
    +e^{-j\pi(\hat{n}_1+\hat{n}_2)}
    \otimes |e\rangle\langle e|.
\end{align}
However, using this unitary would require stopping the stabilisation, even if we consider the dissipator $\hat \cD$ in Eq.\eqref{eq:disp_dyn_2}; only then can we perform the leakage detection with $\hat U_{\rm PC}$. To circumvent this issue, we suggest using the following unitary for leakage detection:
\begin{align}\label{eq:U_JPC}
    \hat{U}_{\rm JPC}
    &=
    \hat{I}\otimes|g\rangle\langle g|
    +
    e^{-j\pi(\hat{n}_1-\hat{n}_2)}
    \otimes
    |e\rangle\langle e|
    \nonumber\\
    &=
    \hat{I}\otimes|g\rangle\langle g|
    +
    \hat{\Pi}\otimes|e\rangle\langle e|,
\end{align}
where $\hat{\Pi}=e^{-j\pi(\hat{n}_1-\hat{n}_2)}$. From Table~\ref{tab:rot_gate_tab}, $\hat{Z}_L \propto (\hat{n}_1-\hat{n}_2)$ and therefore naturally commutes with $\hat \Pi$. Therefore, a logical $\hat{Z}$ error on the DRCC qubits commutes through, leaving the JPC interaction unchanged. In contrast, logical $\hat{X}$ errors anti-commute with $\hat{\Pi}$, giving
\begin{align}
    \hat{U}_{\rm JPC}
    (\hat{X}_L\otimes I_t)
    \hat{U}_{\rm JPC}
    &=
    \hat{X}_L\otimes|g\rangle\langle g|
    +
    \hat{\Pi} \hat{X}_L\hat{\Pi}
    \otimes
    |e\rangle\langle e|
    \nonumber\\
    &=
    \hat{X}_L\otimes\hat{\sigma}_Z,
\end{align}
where we used $\{\hat{X}_L,\hat{\Pi}\}=0$.

Thus, logical $\hat{X}$ errors induce phase flips on the auxiliary qubit. By preparing the auxiliary qubit in the $(|g\rangle+|e\rangle)$ basis and subsequently measuring it in the $\{|g\rangle\pm|e\rangle\}$ basis, one can detect effective logical $\hat{X}$ errors arising from correlated photon-loss events.

Similarly, if we have a bit-flip error in the auxiliary qubit, it propagates as
\begin{align}
    \hat{U}_{\rm JPC}
    (\hat{\bar{P}}\otimes\hat{\sigma}_X)
    \hat{U}_{\rm JPC}
    =
    \hat{\Pi}\otimes\hat{\sigma}_X,
\end{align}
where $\hat{\sigma}_X=|e\rangle\langle g|+ |g\rangle\langle e|.$ Since $(|g\rangle+|e\rangle)$ is an eigenstate of $\hat{\sigma}_X$, this type of error on the auxiliary qubit does not affect the measurement protocol. However, a phase-flip error affecting the auxiliary qubit can propagate to the DRCC qubits. Addressing the protection of the auxiliary qubit against phase-flip noise is left as an open direction for future work.

\begin{figure}[t!]
    \centering
    \includegraphics[width=1\columnwidth]{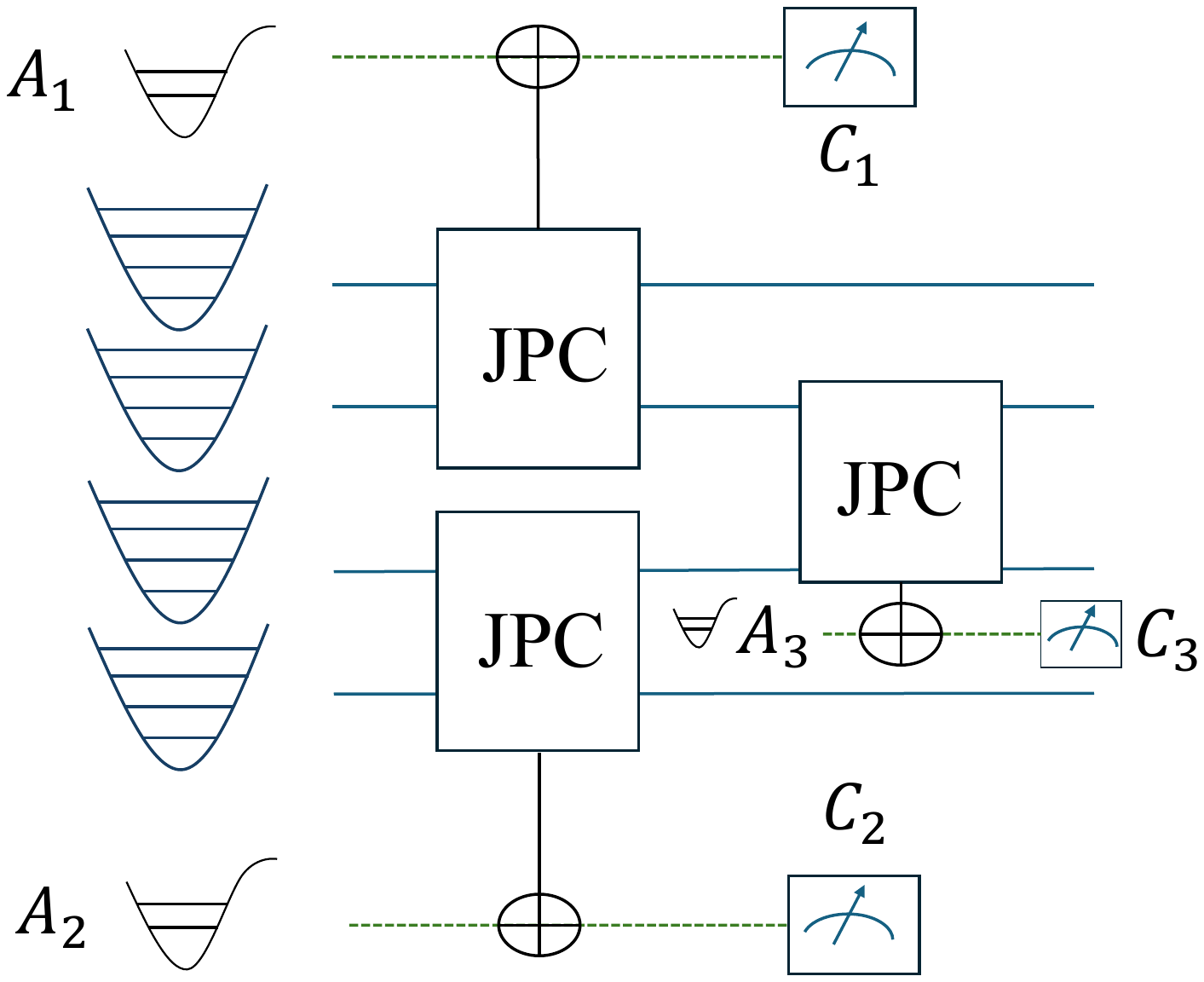}
    \caption{Schematic for syndrome measurement and error correction with two-qubit repetition code with DRCC-1. For the syndrome measurement, we use a two-level system as an auxiliary qubit. Note that the JPC operation does not introduce entanglement between the auxiliary two-level system and the DRCC-1.}
    \label{fig:DRCC_alt}
\end{figure}

\subsubsection{Computational basis measurement}

To complete the discussion on the universal control of DRCC-1 qubits, we discuss measurements in the computational ($Z$) basis, denoted by $\mathcal{M}_Z$.

 Following Ref.~\cite{practical_cat_XZZX_puri}, we first consider the resonator-assisted interaction Hamiltonian 
\begin{align}
   \hat{ H}_{M_Z}^{r}
    =  \frac{J}{2}
    (\hat{n}_1-\hat{n}_2)
    (\hat{a}_r+\hat{a}_r^{\dagger}),
\end{align}
where $\hat{a}_r$ is the annihilation operator of the auxiliary resonator. Using Eqs.~\eqref{eq:n1_w/obs}--\eqref{eq:n2_w/obs}, the projected interaction becomes
\begin{align}
    \hat{\bar{P}}\hat{H}_{M_Z}^{r}\hat{\bar{P}}
    \propto
    \hat{Z}_L(\hat{a}_r+\hat{a}_r^{\dagger}).
\end{align}
Thus, the resonator response depends on the logical state of the DRCC-1 qubit and can be detected via homodyne measurement~\cite{practical_cat_XZZX_puri}.

The structure of $\hat{H}_{M_Z}^{r}$ closely resembles the parity-measurement Hamiltonian used for pair-cat codes~\cite{Albert_pair_cat}. Importantly, this interaction also distinguishes the odd-parity DRCC-1 subspace from the even-parity leakage manifold spanned by $|C^{\pm}C^{\pm}\rangle$. Hence, it naturally enables a joint-parity check, similar to parity measurements in conventional dual-rail encodings~\cite{Yale_dual}.

Instead of an auxiliary resonator, one may also use an auxiliary qubit. The corresponding interaction Hamiltonian in this case is
\begin{align}
    \hat{H}_{M_Z}^{t}
    =
    \frac{J}{2}
    (\hat{n}_1-\hat{n}_2)\hat{\sigma}_Z,
\end{align}
where $\hat{\sigma}_Z$ acts on the auxiliary qubit. Similar oscillator-qubit couplings have previously been proposed for bosonic QEC architectures~\cite{Putterman_2025_aws_cat, Albert_pair_cat}.

Having discussed the required operations, we analyse, in the following Sections, a deterministic QEC protocol for DRCC-1 concatenated with an outer repetition code.

\section{Concatenation with repetition code}\label{sec:conc_code}

In the presence of single-photon losses, the single-mode two-component cat code exhibits a dominant logical noise channel together with weaker residual error components~\cite{Puri_2020}. Similar biased-noise behaviour has been reported for squeezed cat codes~\cite{hillmann2023quantum} and multimode cat codes such as the pair-cat code~\cite{Albert_pair_cat}. To suppress the residual errors, several works have explored concatenating cat codes with outer codes, including repetition, CSS, surface, Golay, and LDPC codes~\cite{Puri_2020, steane_cat, Putterman_2025_aws_cat, practical_cat_XZZX_puri, steane_cat, two_cat_ldpc, ruiz_ldpc_4cat}, as well as employing such concatenated architectures in quantum algorithms~\cite{126_133_cat}.

In this Section, we investigate the concatenation of the DRCC-1 with outer repetition codes. We first consider a two-qubit repetition code and show that, in the large-$\alpha$ limit, it is sufficient to correct the dominant single-mode single-photon loss errors. We then study a three-qubit repetition code and show that it can additionally correct logical $\hat X_L$ errors induced by joint photon-loss events associated with operators such as $\hat a_1\hat a_2$, which arise at second order in $\kappa t$. Finally, we generalise these results to an $n$-qubit repetition code.

\subsection{Error correction with DRCC-1 and outer two-qubit repetition code}\label{sec:two_rep_drcc}

From the QEC conditions in Sec.~\ref{sec:QEC_properties}, we see that with the DRCC-1 alone, we can detect single-photon loss events in the individual cavities, but we cannot distinguish which mode the error occurred in. Previous work on the standard dual-rail code~\cite{Yale_dual, Kubica_2023, Kubica_erasure_threshold} has discussed erasure correction. We show that a simple stabiliser-measurement-based QEC with a repetition code is sufficient to correct erasures caused by single-photon loss. 

 To correct the single-photon loss errors on the individual cavities, we concatenate the DRCC-1 with a two-qubit repetition code, yielding the code vectors
\begin{gather}\label{eq:two_reps}
    \begin{aligned}
     |\bar{0}_2\rangle & = |\bar{0}\bar{0}\rangle= |C^+C^-C^+C^-\rangle\\
    |\bar{1}_2\rangle & = |\bar{1}\bar{1}\rangle= |C^-C^+C^-C^+\rangle
\end{aligned}
\end{gather}
with stabilizer generator $\cS_2= {\hat{Z}_L\hat{Z}_L}$ and projector $\hat{\bar{P}}_2 = |\bar{0}_2\rangle\langle\bar{0}_2|+|\bar{1}_2\rangle\langle \bar{1}_2|$ onto the codespace.

The code ${|\bar{0}_2\rangle,|\bar{1}_2\rangle}$ satisfies the KL conditions for the first-order photon-loss errors $\hat{a}_i$s as follows:
\begin{align}\label{eq:kl1_rep_drcc}
    \hat{\bar{P}}_2\hat{a}_i^{\dagg}\hat{a}_i\hat{\bar{P}}_2 & =\begin{cases}
     (A_0 \hat{\bar{P}} + A_Z \hat{Z}_L) \otimes \hat{\bar{P}} & {\rm for \,} i=1\\
     (A_0 \hat{\bar{P}} - A_Z \hat{Z}_L) \otimes \hat{\bar{P}} & {\rm for \,} i=2
    \end{cases},\\
  \label{eq:kl2_rep_drcc}   \hat{\bar{P}}_2\hat{a}_i^{\dagg}\hat{a}_i\hat{\bar{P}}_2 & =\begin{cases}
   \hat{\bar{P}} \otimes  (A_0 \hat{\bar{P}} + A_Z \hat{Z}_L)  & {\rm for \,} i=3\\
    \hat{\bar{P}} \otimes  (A_0 \hat{\bar{P}} - A_Z \hat{Z}_L)& {\rm for \,} i=4,
    \end{cases}
\end{align}
where $A_0,A_Z$ are given by Eqs.\eqref{eq:A_X} and \eqref{eq:A_Z}. Moreover, for photon loss in different modes, we have
\begin{align}\label{eq:loss_rep_diff}
\begin{split}
     \hat a_1 |\bar{0}_2\rangle & \propto |\Phi_-\rangle|\bar{0}\rangle\\
   \hat a_1 |\bar{1}_2\rangle & \propto |\Phi_+\rangle|\bar{1}\rangle\\
   \hat a_2 |\bar{0}_2\rangle & \propto |\Phi_+\rangle|\bar{0}\rangle\\
   \hat a_2 |\bar{1}_2\rangle & \propto |\Phi_-\rangle|\bar{1}\rangle,
\end{split}
\end{align}
where $|\Phi_{\pm}\rangle$ are given in Eq.\eqref{eq:cap_Phi}. The states $|\Phi_{\pm}\rangle$ are orthogonal to each other and also orthogonal to the code vectors. Therefore, from Eq.\eqref{eq:loss_rep_diff}, we obtain $\hat{\bar{P}}_2 \hat a_1^{\dagger}\hat a_2\hat {\bar{P}}_2= 0 $. Similarly, one can verify that $\hat{\bar{P}}_2\hat{a}_i^{\dagger}\hat a_j\hat{\bar{P}}_2= 0$ for all $i \neq j$. Hence, the code in Eq.\eqref{eq:two_reps} approximately satisfies the KL condition for all first-order photon-loss errors. Consequently, the DRCC-1 concatenated with the two-qubit repetition code corrects all first-order photon-loss errors in the large $\alpha$ regime.

\begin{table}[t]
    \centering
    \caption{Syndrome table for the photon loss errors acting on the individual modes of the two-qubit repetition code with an inner DRCC-1 with the code vectors $|\bar{0}_2\rangle$ and  $|\bar{1}_2\rangle$  defined in Eq.\eqref{eq:two_reps}.}
    \begin{tabular}{c c c | c | c}
    \midrule
       $C_1$  & $C_2$  & $C_3$ & Error & Correction\\
       \midrule
       $e$  & $g$      & $g$ & $\hat a_1$ & $\hat X_c$ on the first mode \\
        &     &   & &\\
       $e$ & $g$    & $e$ & $\hat a_2$ & $\hat X_c$ on the second mode\\
       \midrule
        $g$  & $e$      & $e$ & $\hat a_3$ & $\hat X_c$ on the third mode\\
        &     &   & &\\
       $g$ & $e$  &   $g$ & $\hat a_4$ & $\hat X_c$ on the fourth mode\\
       \midrule
        $g$ & $g$  &  $g$ & No errors & Identity \\
        \midrule
    \end{tabular}
    
    \label{tab:syn_tab_alt}
\end{table}
We now focus on a simplified QEC protocol for syndrome extraction and feedback, shown schematically in Fig.~\ref{fig:DRCC_alt}. 
The QEC protocol proceeds as follows.

\begin{itemize}

\item [S1:] We begin with the state $|\psi_0 \rangle = c_0|\bar{0}\bar{0}\rangle+c_1|\bar{1}\bar{1}\rangle$.
\item[S2:] Suppose a single-photon loss error occurs on the $i{-\rm th}$ mode. After the error, the state becomes
\begin{align}
    |\psi_1\rangle &\sim \begin{cases}
        (c_0|\Phi_-\rangle|\bar{0}\rangle+ \coth |\alpha|^2c_1|\Phi_+\rangle|\bar{1}\rangle) & \mbox{if the error is } \, a_1\\
        (\coth |\alpha|^2c_0|\Phi_+\rangle|\bar{0}\rangle+ c_1|\Phi_-\rangle|\bar{1}\rangle) & \mbox{if the error is } \, a_2\\
        (c_0|\bar{0}\rangle|\Phi_-\rangle+ \coth |\alpha|^2c_1|\bar{1}\rangle|\Phi_+\rangle) & \mbox{if the error is } \, a_3\\
        (\coth |\alpha|^2c_0|\bar{0}\rangle|\Phi_+\rangle+ c_1|\bar{1}\rangle|\Phi_-\rangle) & \mbox{if the error is } \, a_4.
    \end{cases}
\end{align}

\item[S3:] \textit{Joint-parity check:} Since a single-photon loss maps the DRCC-1 states from the odd-photon-number-parity subspace to the even-photon-number-parity subspace, we perform a joint-parity check using an auxiliary qubit, following an approach analogous to that proposed in Ref.~\cite{Yale_dual}. The JPC unitary is defined in Eq.~\eqref{eq:U_JPC}.

Following the discussion in Sec.~\ref{sec:meas} and Ref.~\cite{Yale_dual}, each auxiliary qubit is first prepared in the state $|g\rangle-|e\rangle$ by applying a Hadamard gate to $|e\rangle$. In the absence of photon loss, the controlling oscillators remain within the odd-parity DRCC-1 subspace. The action of $\hat U_{\rm JPC}$ then transforms the auxiliary-qubit state from $|g\rangle-|e\rangle$ to $|g\rangle+|e\rangle$. A second Hadamard gate applied before readout maps this state to $|g\rangle$. Therefore, the measurement outcome $|g\rangle$ indicates that no photon loss has occurred.

In contrast, if a single-photon-loss error from the set ${\hat a_1,\hat a_2,\hat a_3,\hat a_4}$ occurs, the encoded state is transferred to an even-parity subspace. In this case, the auxiliary-qubit state $|g\rangle-|e\rangle$ remains unchanged under the action of $\hat U_{\rm JPC}$. The second Hadamard gate then maps this state to $|e\rangle$. Therefore, the measurement outcome $|e\rangle$ signals the occurrence of a photon-loss error. The corresponding measurement results are recorded in the classical registers $C_1$, $C_2$, and $C_3$.

We emphasise that higher-order photon-loss processes can produce false syndrome flags. Moreover, a phase-flip error on an auxiliary qubit can interchange the outcomes $|g\rangle$ and $|e\rangle$, resulting in either a false-positive or a false-negative syndrome.

\item[S4:] The outcomes $(C_1, C_2, C_3)$ determine which error among  $\{\hat a_1,\hat a_2,\hat a_3,\hat a_4\}$ has occurred, and the corresponding correction is applied according to Table~\ref{tab:syn_tab_alt}. Once we detect the photon loss, we apply the operation
    \begin{align}\label{eq:X_c}
        \hat X_c = |C_+\rangle\langle C_-|+|C_-\rangle\langle C_+|
    \end{align}
to the corresponding cavity.

The operator $\hat X_c$ corresponds to the logical $X$-gate for the single-mode cat code and the Hamiltonian $\hat H_c = J(\hat a+ \hat a^{\dagg})$~\cite{Puri_2020} generates the $\hat X_C$. The action of $\hat H_c$ on the single-mode cat-code manifold spanned by $\{|C^{\pm}\rangle\}$ is
    \begin{align*}
    \begin{split}
         \hat P_c\hat H_c\hat P_c &= J \hat P_c(\hat a+\hat a^{\dagg})\hat P_c\\
         & = 2J \mathbf{Re}(\alpha) \hat X_c, 
    \end{split}
        \end{align*}
where $\hat P_c = |C^{+}\rangle\langle C^{+}|+|C^{-}\rangle\langle C^{-}|$ is the projector onto the single mode cat-codespace. Therefore, the operation $\hat X_c$ is generated by evolving under $\hat H_c$ for a period of time
\begin{align*}
    t_{\hat X_c} \sim \frac{\pi}{4J\mathbf{Re}(\alpha)}.
\end{align*} 
\end{itemize}

We also note that our syndrome extraction scheme in Fig.\ref{fig:DRCC_alt} allows us to distinguish the loss events of two-modes $\{\hat a_i\hat a_j\}$ from the identity error, as summarised in Table~\ref{tab:syn_tab_alt_2}.

\begin{table}[t!]
    \centering
    \caption{Syndrome table for the photon loss errors acting on the two different modes of the two-qubit repetition code with an inner DRCC-1 with the code vectors $|\bar{0}_2\rangle$ and  $|\bar{1}_2\rangle$  defined in Eq.\eqref{eq:two_reps}.}
    \begin{tabular}{c c c | c }
    \midrule
       $C_1$  & $C_2$  & $C_3$ & Error \\
       \midrule
       $g$  & $g$      & $e$ & $\hat a_1\hat a_2$ \\
        $g$ & $g$  &  $e$ & $\hat a_3\hat a_4$  \\
      \midrule
        $e$  & $e$      & $g$ & $\hat a_1\hat a_4$ \\
       $e$ & $e$  &   $g$ & $\hat a_2\hat a_3$ \\
       \midrule
        $e$ & $e$  &  $e$ & $\hat a_2\hat a_4$  \\
        $e$ & $e$    & $e$ & $\hat a_1\hat a_3$ \\
        \midrule
    \end{tabular}
    
    \label{tab:syn_tab_alt_2}
\end{table}

Tables~\ref{tab:syn_tab_alt} and \ref{tab:syn_tab_alt_2} show that the syndromes associated with single-photon loss and two-photon-loss errors are disjoint. However, two-photon-loss events cannot be corrected since, as shown in Table~\ref{tab:syn_tab_alt_2}, the corresponding syndromes do not allow us to distinguish which pair of cavities suffered the photon loss. An alternative approach  for correcting single-photon loss with the DRCC-1 concatenated with an outer two-qubit repetition code is discussed in Appendix~\ref{sec:alt_qec}.

To estimate the logical error probability for the two-qubit repetition code concatenated with the DRCC-1 in Eq.~\eqref{eq:two_reps}, we identify the lowest-order uncorrectable error processes:
\begin{enumerate}

\item[(a)] \textit{Simultaneous leakage in two out of four modes:}
Simultaneous photon-loss events on multiple modes are not distinguishable within the two-qubit repetition code. Although we can detect all these events, the affected pair of modes cannot be uniquely identified, thereby preventing deterministic recovery using the Petz map or other optimal recovery procedures. The probability of such events is
\begin{align}
p_{\text{leak}} = 4\kappa^2A_X^2t^2,
\end{align}
where we assume identical noise strengths across all modes.

\item[(b)] \textit{Logical bit-flip errors:}
Logical bit-flip errors arise from simultaneous photon-loss events on adjacent modes, generated by operators such as $\hat a_1\hat a_2$ or $\hat a_3\hat a_4$. The corresponding probability is
\begin{align}
p_{\rm bit-flip} = 2 A_X^2 \kappa^2t^2.
\end{align}

\item[(c)] \textit{Residual logical $\hat Z$ and $\hat Y$ errors:}
The QEC against photon loss itself generates a dephasing error, which can be understood as follows. The recovery operator of the Petz map \cite{liang_jiang_near_opt} recovering against a single-photon loss is given by 
\begin{align}
    \hat R_k &=\sum\limits_{\mu ,\nu, l} (M^{-1/2})_{[\mu k],[\nu l]}|\mu\rangle\langle\nu| \hat L_k^{\dagger} \nonumber\\
    &\approx\sqrt{\kappa_kt}  \sum\limits_{\mu ,\nu, l} (M^{-1/2})_{[\mu k],[\nu l]}|\mu\rangle\langle\nu| \hat a_k^{\dagger}.
\end{align}
When we apply this recovery to correct the single-photon loss on the $k^{th}$-mode, we obtain the following state:
\begin{align}
\begin{split}
    \hat R_k\hat L_k|\bar{m}\rangle & \approx \kappa_k t  \sum\limits_{\mu,\nu, l} (M^{-1/2})_{[\mu k],[\nu, l]}|\bar \mu\rangle\langle\bar \nu| \hat a_k^{\dagger} a_k |\bar{m}\rangle\\
   &  =\sum\limits_{\mu,\nu, l}(M^{-1/2})_{[\mu k],[\nu l]}M_{[\nu k], [m k]}|\bar \mu\rangle.
\end{split}
\end{align}
Now, from Eq.\eqref{eq:kl1_rep_drcc}-\eqref{eq:kl2_rep_drcc} we have 
\begin{align}
     M_{[\nu k],[m k]} & = \begin{cases}
       \kappa_kt ( A_0\bar   + A_Z )& \nu = m=0 \\
       \kappa_kt ( A_0\bar   - A_Z ) & \nu= m =1\\
        0 & \nu \neq m .
    \end{cases}
\end{align}
\begin{figure}
    \centering
    \includegraphics[width=1\columnwidth]{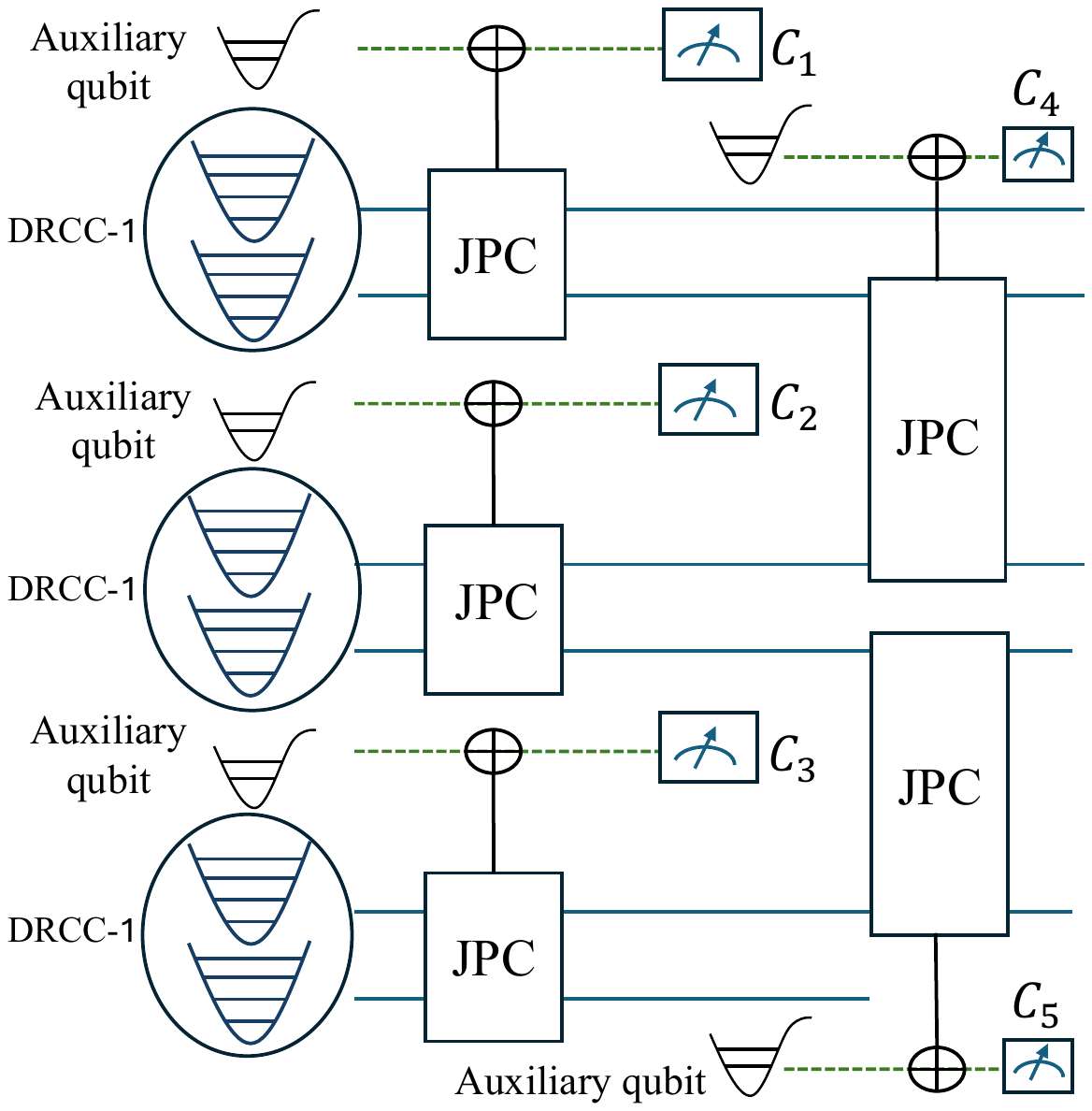}
    \caption{Syndrome measurement scheme with the three-qubit repetition code with an inner DRCC-1. }
    \label{fig:three_reps_syn}
\end{figure}
Using the above equation, we calculate the recovered state as
\begin{align}
    \hat R_k\hat L_k|\bar{m}\rangle 
   & = \sqrt{\kappa_kt} \left(\sum (\sqrt{M})_{[mk],[mk]}|\bar m\rangle +|\psi_k\rangle \right)\\
   &  =\sqrt{ \kappa_kt} \left(\sum (A_0 +(-1)^m A_z)|\bar m\rangle + |\psi_k\rangle\right ),
\end{align}
where $|\psi_k\rangle=  \sum\limits_{k \neq l} M_{[mk],[ml]}|\bar m\rangle$. 
Thus, we see the correction itself introduces a dephasing error. The leading-order contribution of such errors is
\begin{align}
p_{\rm other} = 4 A_Z^2 \kappa t.
\end{align}

\end{enumerate}

An upper bound to the total probability of uncorrectable logical errors, up to second order in photon loss, can therefore be estimated as
\begin{align}
p_{\rm logical} \sim p_{\text{leak}} +p_{\rm bit-flip}+p_{\rm other}.
\end{align}
Fig.~\ref{fig:upper_bound_logical} shows the dependence of $p_{\rm logical}$ on $\alpha$ and the noise strength.

\subsection{Error correction with DRCC-1 and outer three-qubit repetition code}

To correct the logical bit-flip errors arising from two-photon-loss events, we use the three-qubit repetition code
\begin{gather}\label{eq:three-reps}
\begin{aligned}
    |\bar{0}_3\rangle & = |\bar{0}\bar{0}\bar{0}\rangle\\
     |\bar{1}_3\rangle & = |\bar{1}\bar{1}\bar{1}\rangle,
\end{aligned}
\end{gather}
with stabiliser generator $\cS_3 = \langle \hat Z_L\hat Z_L\hat I, \hat I\hat Z_L\hat Z_L\rangle$ and projector
\begin{align}
    \hat P_3=|\bar{0}_3\rangle\langle\bar{0}_3|+|\bar{1}_3\rangle\langle \bar{1}_3|
\end{align}
onto the codespace.

The correction protocol follows steps analogous to those discussed for the two-qubit repetition code. In Appendix~\ref{sec:three-qubit_rep_qec}, we show that the three-qubit repetition code with the inner DRCC-1 corrects all single-photon loss errors as well as the two-photon-loss events. The corresponding syndrome-extraction procedure is shown in Fig.~\ref{fig:three_reps_syn}, and the complete syndrome table is given in Table~\ref{tab:three_reps_syn}.

\begin{table}[t!]

    \centering
    \begin{tabular}{c c c c c |c}
    \midrule
         $C_1$& $C_2$ & $C_3$ & $C_4$ & $C_5$ & Errors \\
         \midrule
        $e$  & $g$  & $g$ & $g$ & $g$ & $\hat a_1$\\
         $e$  & $g$  & $g$ & $e$ & $g$ & $\hat a_2$\\
          $g$  & $e$  & $g$ & $e$ & $g$ & $\hat a_3$\\
           $g$  & $e$  & $g$ & $g$ & $e$ & $\hat a_4$\\
           $g$ & $g$ & $e$ & $g$  & $e$ & $\hat a_5$\\
           $g$ & $g$& $g$& $g$& $e$ &$\hat a_6$ \\
           $g$ & $g$& $g$& $e$& $g$ &  $\hat{a}_1\hat{a}_2$\\
            $e$& $e$& $g$& $e$& $g$ &  $\hat{a}_1\hat{a}_3$\\
           $e$ & $e$& $g$& $g$& $e$ &  $\hat{a}_1\hat{a}_4$\\
           $e$ &$g$ & $e$& $g$& $e$ &  $\hat{a}_1\hat{a}_5$\\
           $e$ & $g$& $e$& $g$&  $e$&  $\hat{a}_1\hat{a}_6$\\
           $e$ & $e$ & $g$& $g$& $g$ &  $\hat{a}_2\hat{a}_3$\\
           $e$& $e$& $g$& $e$& $e$ &  $\hat{a}_2\hat{a}_4$\\
           $e$& $g$& $e$& $e$& $e$ &  $\hat{a}_2\hat{a}_5$\\
           $e$& $g$& $e$& $e$&  $g$&  $\hat{a}_2\hat{a}_6$\\
           $g$& $g$& $g$& $e$& $e$ &  $\hat{a}_3\hat{a}_4$\\
           $g$& $e$& $e$& $e$& $e$ &  $\hat{a}_3\hat{a}_5$\\
           $g$& $e$& $e$& $e$&  $g$&  $\hat{a}_3\hat{a}_6$\\
           $g$& $e$& $e$& $g$&  $g$&  $\hat{a}_4\hat{a}_5$\\
           $g$& $e$& $e$& $g$&  $e$&  $\hat{a}_4\hat{a}_6$\\
           $e$& $e$& $e$&$g$ &  $e$&  $\hat{a}_5\hat{a}_6$\\
           \bottomrule
    \end{tabular}
    \caption{Syndrome table for the three-qubit repetition code with an inner DRCC-1 code (Eq.~\eqref{eq:three-reps}).}
    \label{tab:three_reps_syn}
\end{table}

Unlike the two-qubit repetition code, which corrects only first-order photon-loss events, the three-qubit repetition code can correct loss errors up to second-order. These include logical bit-flip errors within the codespace as well as double-leakage events affecting two different qubits.

Consequently, the lowest-order uncorrectable logical errors consist of the following:
\begin{enumerate}
\item[(i)] three simultaneous leakage events ($p_{\rm leak_3}$) $\sim \hat a_i\hat a_{i+1}\hat a_{i+2}$,

\item[(ii)] four simultaneous leakage events ($p_{\rm leak_4}$) $\sim \hat a_i\hat a_{i+1} \hat a_j\hat a_{j+1}$,
\item[(iii)] Similar to the two-qubit outer code, the QEC introduces a residual logical $\hat Z$-type errors arising from dephasing. 
\end{enumerate}

Following the analysis presented in Sec.~\ref{sec:numerics}, the leading-order upper bound on the logical error probability has the following form
\begin{align}
p_{\rm logical}
&\sim p_{\rm leak_3}+p_{\rm leak_4}+p_{\rm other}.
\end{align}

Substituting the corresponding leading-order contributions, we obtain
\begin{align}
p_{\rm logical}&=15 \left(A_X^2 \kappa^2 t^2\right)^2+20A_X^3 \kappa^3 t^3+6A_Z^2 \kappa t.
\end{align}

Here, the numerical prefactors arise from the different combinatorial ways in which the corresponding error processes can occur. The logical error ($p_{\rm logical}$) as a function of the cat amplitude $\alpha$ is shown in Fig.\ref{fig:upper_bound_logical}.

\begin{widetext}
    \begin{figure*}
        \centering
        \includegraphics[width=\linewidth]{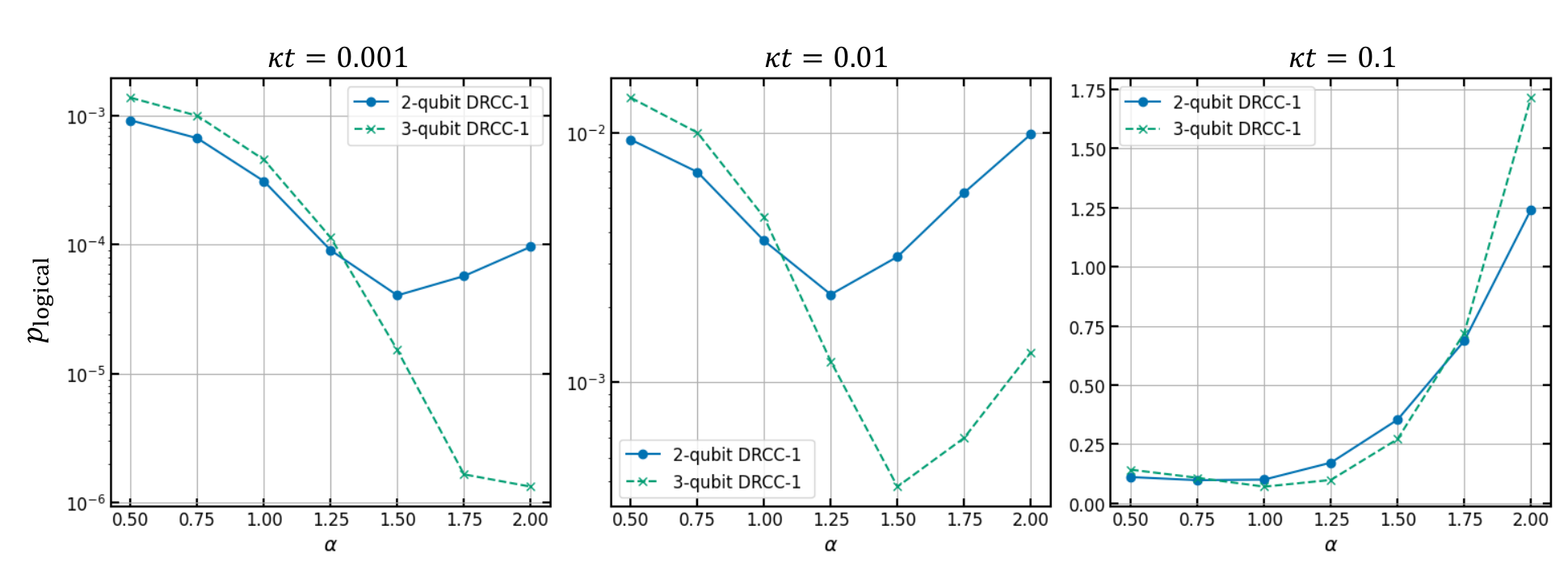}
        \caption{Upper bound of the probability of logical error for two- and three-qubit repetition code with inner DRCC-1, as a function of the coherent state amplitude in different noise strength regimes. The optimal regime, corresponding to the minimum value of $p_{\rm logical}$, occurs in the range $1\leq\alpha\leq 1.75$ for different values of $\kappa t$ considered. Furthermore, the location of the optimal point depends on the loss strength $\kappa t$.}
        \label{fig:upper_bound_logical}
    \end{figure*}
\end{widetext}

\subsection{Quantum error correction with outer $n$-qubit repetition code}

Having analysed the QEC performance of the DRCC-1 concatenated with two- and three-qubit repetition codes under photon-loss, we now consider a general $n$-qubit repetition code as the outer code. Our goal is to determine how the error-correction capabilities scale with $n$ and whether the resulting concatenated code can correct a broader class of errors than the standard dual-rail code.

 As discussed in Sec.~\ref{sec:two_rep_drcc}, the two-qubit repetition code corrects all first-order photon-loss events and detects all second-order photon-loss events through distinct syndromes. However, because the syndrome information does not identify which DRCC-1 block experienced the leakage, errors of the form $\hat a_i\hat a_j$ remain detectable but uncorrectable. Therefore, correcting such second-order loss events requires an additional DRCC-1 block.

This observation generalises to larger repetition codes. An $n$-qubit repetition code, requiring $2n$ cavities, can correct all $(n-1)$-th order photon-loss errors. In other words, correcting $t-1$ erasures with the DRCC-1 requires at least a $t$-qubit outer repetition code, corresponding to a code distance $d=t+1$.

\begin{widetext}
\begin{center}
 \begin{table}[t!]
 \centering
    \begin{tabular}{|c|c|c |c | c| }
    \midrule
         & \multirow{2}{8em }{dual-rail code \{$|01\rangle, |10\rangle$\}} & \multirow{2}{10em }{Repetition cat code $\{|C^{+}C^+\rangle,|C^{-}C^-\rangle\}$} & \multirow{2}{10em}{Pair cat code  ~~~~~~ (Eq.(3.12) in Ref.\cite{Albert_pair_cat})} & \multirow{2}{8em}{DRCC-1 $\{|C^{+}C^-\rangle,|C^{-}C^+\rangle\}$}  \\
         & & & &\\
         \midrule
         Geometric relative phase & In general present & Present  & Present  & Absent\\
         & & & & \\
         \multirow{2}{8em }{Bias preservation of two-qubit gates} & \multirow{3}{12em }{ Erasure bias not preserved: control-qubit erasure $\rightarrow$ target $Z$-error}  & \multirow{3}{12em }{Preserves bias against $\hat Z$ errors from single-photon loss during $CX$.}   &\multirow{2}{8em }{$CX$ preserves the erasure during or before the gates.} &\multirow{3}{10em }{ All two-qubit gates preserve erasure bias before and during the gates.}\\ 
          & & & & \\
          & & & & \\
          & & & & \\
        
         \multirow{2}{10em}{ Erasure correction (with two-qubit repetition code)} & \multirow{2}{10em}{Does not correct (Appendix \ref{sec:qec_drc}).} & \multirow{2}{10em}{Deterministic correction} & \multirow{2}{10em}{Deterministic correction for a particular $\alpha$, without any outer code.} &\multirow{2}{10em}{Deterministic correction}\\
         & & & & \\
          & & & & \\
           & & & & \\
            & & & & \\
          Dissipator and parity check & Not applied  & Non commuting & Commute & Commute \\
         \hline
          \multicolumn{5}{|c|}{Other features}\\
          \hline
           \multirow{2}{8em}{Two-qubit gate implementation} & \multirow{2}{8em}{Extra qubit required} &  \multirow{2}{8em}{Quartic gates (e.g., cross-Kerr)} &\multirow{2}{8em}{ Quartic gates (e.g., cross-Kerr)} & \multirow{2}{10em}{Quartic gates (e.g., cross-Kerr, Eq.~\eqref{eq:H_XX}) }\\
             & & & & \\
        
            & & & & \\
          Parity measurement & $\hat U_{\rm PC} $ Eq.\eqref{eq:jpc_yale}& $e^{- j\chi \hat n_i \otimes \hat \sigma_Z}$ Ref.\cite{Puri_2020} & $e^{-j (\hat n_1-\hat n_2)\hat \sigma_Z}$ Ref.\cite{Albert_pair_cat}& $\hat U_{\rm JPC}$ Eq.\eqref{eq:U_JPC}\\
            & & & & \\
          \multirow{2}{8 em }{Undetectable errors} & \multirow{2}{8 em }{$\hat a^{\dagger2}_i$, $\hat a^{\dagger}_1\hat a^{\dagger}_2$ ,$\hat a_1^{\dagger}\hat a_2$, $\hat a_2^{\dagger}\hat a_1$ }& \multirow{2}{8 em }{$\hat a_1\hat a_2$ and $\hat a^{\dagger}_1\hat a^{\dagger}_2$} & \multirow{2}{8 em }{$\hat{a} _i^2$, $\hat{a} _1\hat{a} _2$  Ref.\cite{Albert_pair_cat} } & \multirow{2}{8 em }{$\hat a_1\hat a_2$ , $\hat a^{\dagger}_1\hat a^{\dagger}_2$, $\hat a_1^{\dagger}\hat a_2$, $\hat a_2^{\dagger}\hat a_1$}\\
            & & & &\\
            & & & &\\
         \midrule
    \end{tabular}
    \caption{Comparison of DRCC-1 with other bosonic codes. 
    For assessing erasure correction, we consider the codes in the table concatenated with a two-qubit repetition code. The concatenated DRCC-1, pair-cat code and repetition-cat code enable deterministic correction of the dominant single-photon-loss errors. For the corresponding DRC, block-level erasure information alone cannot distinguish whether the photon was lost from the first or the second mode of the
erased DRC block. Consequently, the DRC concatenated with this two-qubit
repetition code does not correct the complete single-photon-loss error set
(see Appendix~\ref{sec:qec_drc}). Neither DRCC-1 nor
DRCC-2 accumulates a relative geometric phase. Unlike the DRC $CZ$ gate,
which preserves the leakage bias only when leakage occurs before the gate
operation, the DRCC-1 $CZ$ and $CX$ gates preserve the leakage bias
regardless of whether leakage occurs before or during the gate operation.}
    \label{tab:comp_tab}
\end{table}   
\end{center}
\end{widetext}
\section{General comparison of features of DRCC-1 with other codes}
\label{se:general-comparison-with-DR}

Having discussed the gate implementation and the syndrome measurements via joint parity-check for DRCC-1, we can now compare with the standard dual-rail code.
The DRCC-1 shares several structural similarities with the dual-rail code (DRC); namely, both detect but do not correct single-photon loss, and both support beam-splitter-based single-qubit control. However, they differ in several key aspects.

First, the parity-check mechanism is fundamentally modified. As discussed in  Sec. \ref{sec:meas}, for leakage detection of the DRC, we generally use the unitary $\hat U_{\rm PC}$ in Eq.~\eqref{eq:jpc_yale}, whereas the DRCC employs the unitary $\hat U_{\rm JPC}$. Since we aim to keep the dissipative stabilisation active during the leakage-detection process, the unitary $\hat U_{\rm PC}$ cannot be used directly in the DRCC-1.

The second difference lies in the structure of the leakage space and the resulting error-detection capabilities. In the DRC, the codespace is embedded within a larger odd-parity manifold, so joint-parity measurements cannot detect errors that map the state to other odd-parity subspaces, as noted in Ref.\cite{Yale_dual}. For example, two-photon-gain processes such as $\hat{a}_i^{\dagger 2}$ or $\hat{a}_1^{\dagger}\hat{a}_2^{\dagger}$ preserve the odd parity of the state while taking it outside the codespace (See Fig.\ref{fig:DRCC_code_pic}). Since the joint-parity measurement cannot distinguish such states from the logical subspace, these errors remain undetected~\cite{Yale_dual}. In contrast, the DRCC-1 partitions the Hilbert space into two distinct sectors: an odd-parity codespace and an even-parity leakage space. Consequently, the joint-parity check $\hat U_{\rm JPC}$ in Eq.~\eqref{eq:U_JPC} detects any error that transfers a logical state from the odd-parity codespace to the even-parity leakage space. Errors such as $\hat a_1\hat a_2$, $\hat a_1^{\dagger}\hat a_2^{\dagger}$ or $\hat a_i^{\dagger}\hat a_j$  preserve instead the parity and map the logical states back into the codespace, generating a logical $\hat X$ rotation  (see Sec.\ref{sec:logical gates}). These errors are therefore undetectable by the joint-parity check.

Third, DRCC-1 two-qubit gates preserve erasure bias irrespective of when the error occurs and prevent leakage propagation between blocks. By comparison, in the DRC, erasures during gate operations can induce residual errors or propagate between qubits. For example, an erasure occurring on either the control or target block before or during a $CX$ or $CZ$ gate can propagate to the other block, as discussed in Ref.~\cite{S_puri_erasure_threshold}.

Fourth, with the two-qubit repetition outer code considered here, DRCC-1
enables deterministic correction of the dominant single-photon-loss events
in the large-$\alpha$ regime. By contrast, for the standard DRC concatenated
with the same outer repetition code, block-level erasure information alone
is insufficient to distinguish between photon-loss events generated by
$\hat a_1$ and $\hat a_2$ within the erased block. Consequently, this
concatenation does not correct the complete single-photon-loss error set
when only block-level erasure information is available. For more details on the QEC with dual rail and repetition code, see Appendix \ref{sec:qec_drc}. 


We now turn to the comparison with cat codes. For the repetition cat code, certain two-qubit gates, such as the $CX$ gate, preserve the dominant logical $\hat Z$ errors arising from single-photon loss. However, bias preservation is generally gate dependent, and not all two-qubit gates preserve the $\hat Z$-biased noise channel. In contrast, single-photon loss in the DRCC-1 induces erasures rather than logical $\hat Z$ errors, and all logical gates preserve this erasure bias. Another distinction arises in syndrome extraction. For repetition cat codes, syndrome measurements generally require temporarily turning off the dissipative stabilisation process~\cite{PRX_mazyar_mirrahimi, Albert_pair_cat}. In contrast, the DRCC-1 enables syndrome measurements without interrupting the stabilisation dynamics.

For the pair cat code, there is likewise no entanglement between the data and auxiliary qubit during QEC. The $CX$ operation preserves erasure either during or before the gate, and the code uses a cross-Kerr interaction to implement the two-qubit gate. Its parity-check operations do commute with the dissipator, and the parity measurement is implemented via $e^{-i (\hat n_1-\hat n_2) \hat \sigma_Z}$ as in~\cite{Albert_pair_cat}. Undetectable errors include $\hat a_1\hat a_2, \hat a_1^{\dagger}\hat a_2, \hat a_2^{\dagger}\hat a_1$ and $\hat a_1^{\dagger}\hat a_2^{\dagger}$. 

Finally, DRCC-1 avoids the geometric phase accumulation inherent to gate constructions in other codes, such as the dual-rail code \cite{Tsunoda_PRXQuantum}, pair cat code \cite{Albert_pair_cat} and the repetition cat code \cite{Puri_2020, PRX_mazyar_mirrahimi}. 

A concise comparison of these properties is provided in Table~\ref{tab:comp_tab}.

\section{Numerical analysis and comparison of performance with related cat-code constructions}\label{sec:numerics}

In this Section, we analyse the performance of DRCC-1 concatenated with an outer two-qubit repetition code and compare it with that of a three-qubit repetition code employing inner single-mode cat states.

For this study, we use the Petz recovery map, which provides a near-optimal recovery procedure. We adopt the Petz map because it is computationally less complex than the optimal recovery map, whose evaluation typically requires a numerically intensive optimisation procedure~\cite{liang_jiang_near_opt,fletcher2008channel}. At the same time, the Petz map provides a reliable bound on the QEC performance, allowing us to infer the expected behaviour of the corresponding optimal recovery scheme~\cite{liang_jiang_near_opt, hkn_pm2010,barnum2002}.

Furthermore, we note that under certain assumptions on the code and the noise model, the Petz recovery can coincide with the optimal recovery, as discussed in~\cite{liang_jiang_opt_code}. We note that the Petz recovery admits two different formulations: the state-specific Petz map~\cite{mm_wilde} and the code-specific Petz map~\cite{liang_jiang_near_opt,barnum2002,hkn_pm2010}. The code-specific map is particularly useful in the context of QEC~\cite{biswas2024noise}, and it has recently been implemented on superconducting platforms~\cite{biswas2025universalsyndromebasedrecoverynoiseadapted}. In contrast, a state-specific Petz map has been proposed for trapped-ion processors~\cite{trapped_ion_petz} and nuclear magnetic resonance processors~\cite{Singh_2026}.

We focus here on the code-specific Petz recovery map. For a code $\cC$ and a noise process $\cE$, such map is defined through the Kraus operators
\begin{align}
\cR_{P,\cE}(. ) = \sum\limits_{k} \hat P\hat E^{\dagger}_{k}\cE(\hat P)^{-1/2} (.) \cE(\hat P)^{-1/2} \hat E_k \hat P,
\end{align}
where the inverse of $\cE(\hat P)$ is taken on the support of $\cE(\hat P)$.

To quantify the performance of the QEC protocol, we use the channel fidelity, also referred to as entanglement fidelity~\cite{nielsen1996entanglementfidelityquantumerror,liang_jiang_near_opt}. For a quantum channel
\begin{align}
\Lambda(.) \overset{\rm def}{=} \sum_i \hat \Lambda_i (.) \hat \Lambda_i^{\dagger},
\end{align}
the entanglement fidelity is given by 
\begin{align} \label{eq:ent_fid_def}
F_{\rm ent} = \frac{1}{d^2} \sum_i|\tr (\hat \Lambda_i)|^2,
\end{align}
where $d$ denotes the dimension of the codespace.

Following Ref.~\cite{liang_jiang_near_opt}, the entanglement fidelity corresponding to the Petz recovery map $\cR_{P,\cE}$, which is also known as the transpose channel, can be expressed as
\begin{align}
F_{\rm ent} = \frac{1}{d^2} \bigg \vert\bigg\vert \sqrt{\hat M}\bigg \vert\bigg\vert_F^2,
\end{align}
where $||(.)||_F$ denotes the Frobenius norm and
\begin{align}
\hat M \equiv \{M_{[\mu k],[\nu\ell]}\}= \langle\mu|\hat E_k^{\dagg}\hat E_\ell|\nu\rangle.
\end{align}
Here, ${|\mu\rangle,|\nu\rangle}$ denote the code vectors spanning the codespace $\cC$, and ${\hat E_k}$ are the Kraus operators corresponding to the noise channel $\cE$.

In what follows, to assess the usefulness of the DRCC-1 relative to other multi-qubit two-component cat codes, we consider the three-mode and four-mode cat repetition codes and examine their performance under photon-loss error with the near-optimal recovery via the Petz map. 

\subsection{Comparison with three-qubit single-mode cat repetition code (three modes)}

We first compare the performance of DRCC-1 with the three-mode cat repetition code ${|C^{+}\rangle^{\otimes 3},|C^{-}\rangle^{\otimes 3}}$ against photon loss errors with the near-optimal recovery provided by the Petz map. Figure~\ref{fig:ent_plot_diff_rec} shows that, for DRCC-1, the infidelity $1-F_{\rm ent}$ remains below the break-even threshold over a wide range of noise strengths $\kappa t$. We observe a comparable behaviour for the cat-repetition code, particularly at low values of $\alpha$.

For both codes, the performance initially improves as $\alpha$ increases and reaches an optimal regime around a sweet-spot $\alpha=\alpha_{\text{opt}}$, where $\alpha_{\text{opt}}$ depends on the noise strength $\kappa t$. Beyond this regime, the infidelity increases again. 


For the DRCC-1, the optimal operating point can be understood from the leading-order contributions to the logical error probability discussed in Sec.~\ref{sec:two_rep_drcc}. In particular, the presence of $\hat Z_L$ errors induced by first-order dephasing in each cavity, together with the logical $\hat Y_L$ error induced by joint photon loss on the pair of modes, leads to a minimum in $p_{\rm logical}$ in between $1.0\leq \alpha \leq 1.75$; this optimal point also depends on the noise strength. 
Indeed, as $\alpha$ increases, the probabilities of logical phase-flip and $Y$-type errors decrease exponentially. At the same time, the probabilities of logical $\hat X$ errors and leakage events increase. Therefore, under recovery procedures such as the Petz map, the competition between suppressing $\hat Z$- and $\hat Y$-type errors and enhancing $\hat X$-type and leakage errors yields an optimal operating point.

Finally, we note that as $\alpha$ increases, the DRCC-1 performs better than the repetition code, with a more pronounced reduction in infidelity, especially in the higher-noise regime.

\begin{widetext}
\begin{center} 
\begin{figure}[t!] 
\centering 
\includegraphics[width=1.0\textwidth]{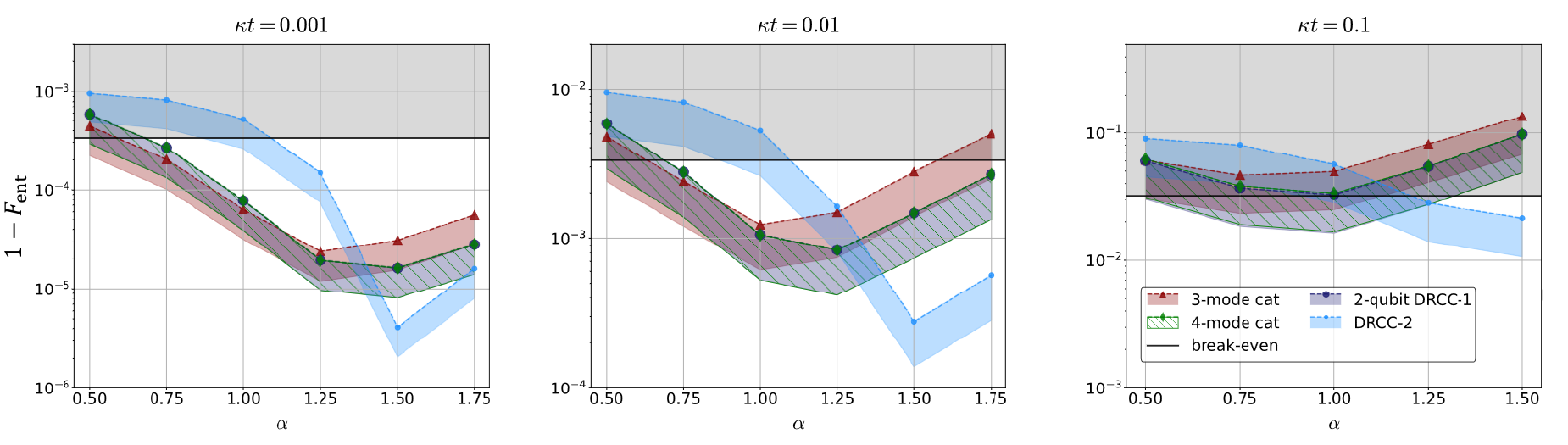}
\caption{Performance of the DRCC-1 and DRCC-2 under photon-loss noise with identical noise rates in all cavities, $\kappa_i t=\kappa t$. We compare the performance of four different encodings: the three-mode repetition cat code, the DRCC-1, the four-mode repetition cat code, and the DRCC-2, using the near-optimal Petz recovery map. The shaded regions for the DRCC-1 and the four-mode repetition cat code overlap almost perfectly in the parameter regime considered. For the DRCC-2, we do not consider mixing the two modes with beam splitter because photon-loss errors commute with the beam-splitter operation, as discussed in Ref.~\cite{multimode_RSB}; consequently, the inclusion of a beam splitter would not alter the performance of the code under photon-loss noise.}
\label{fig:ent_plot_diff_rec}
\end{figure}
\end{center} 
\end{widetext}

\subsection{Comparison with two-qubit two-mode cat repetition code (four modes)}

We now compare DRCC-1 with another two-mode cat-code construction with $N=1$, obtained by concatenating the even two-mode repetition cat code $|\Phi^{\pm}\rangle$ in Eq.\eqref{eq:cap_Phi} with an outer two-qubit repetition code:
\begin{align}\label{eq:another_two_rep_code}
\begin{split}
     |\bar{0}_c\rangle & = |\Phi^+\Phi^+\rangle,\\
    |\bar{1}_c\rangle & = |\Phi^-\Phi^-\rangle.
\end{split}   
\end{align}

This code satisfies the KL conditions for first-order photon-loss errors, as follows:
\begin{align}\label{eq:loss_rep_phi}
\begin{split}
    \hat a_1 |\bar{0}_c\rangle & \propto |\bar{1}\rangle|\Phi_+\rangle,\\
    \hat a_1 |\bar{1}_c\rangle & \propto|\bar{0}\rangle |\Phi_-\rangle,\\
    \hat a_2 |\bar{0}_c\rangle & \propto |\bar{0}\rangle|\Phi_+\rangle,\\
    \hat a_2 |\bar{1}_c\rangle & \propto |\bar{1}\rangle|\Phi_-\rangle.
\end{split}
\end{align}

\begin{table}[t]
    \centering
    \caption{Syndrome table for error correction with the two-qubit repetition code concatenated with the inner repetition cat code defined in Eq.~\eqref{eq:another_two_rep_code}. }
    \begin{tabular}{c c c | c | c}
    \midrule
       $C_1$  & $C_2$  & $C_3$ & Error & Correction\\
       \midrule
       $e$  & $g$      & $g$ & $\hat a_1$ & $\hat X_c$ on the first mode \\
        &     &   & &\\
       $e$ & $g$    & $e$ & $\hat a_2$ & $\hat X_c$ on the second mode\\
       \midrule
        $g$  & $e$      & $g$ & $\hat a_3$ & $\hat X_c$ on the third mode\\
        &     &   & &\\
       $g$ & $g$  &   $e$ & $\hat a_4$ & $\hat X_c$ on the fourth mode\\
       \midrule
        $g$ & $g$  &  $g$ & No errors & Identity \\
        \midrule
    \end{tabular}
    
    \label{tab:syn_tab_alt_rep}
\end{table}

The DRCC-1 logical vectors span the error space, or leakage space, associated with this code, see Fig.~\ref{fig:DRCC_code_pic}. Consequently, the syndrome values $C_1C_2C_3$ for the even-parity repetition code in Eq.~\eqref{eq:another_two_rep_code} and for DRCC-1 concatenated with the repetition code differ only by an addition modulo two, as seen from Tables~\ref{tab:syn_tab_alt} and \ref{tab:syn_tab_alt_2}.

The error-correction circuit shown in Fig.~\ref{fig:DRCC_alt} is also applicable to this code. As in the DRCC-1 with an outer repetition code, we can use the JPC unitary to measure the syndrome with the four-mode repetition cat codes. Owing to these similarities in code structure and error-correction procedure,  the performance of these two codes under photon-loss noise using the Petz recovery channel is nearly identical throughout the parameter regime considered (Fig.~\ref{fig:ent_plot_diff_rec}). 

A full numerical evaluation of the three-qubit repetition code with DRCC-1 as the inner code, using the Petz recovery map (which implies six cavity modes in the protocol), is computationally hard, and we leave such analysis to future work.



\section{QEC properties for Dual-rail four component cat code  (DRCC-2)}\label{sec:drcc2}

\begin{figure}[t!]
    \centering
    \includegraphics[width=0.8\columnwidth]{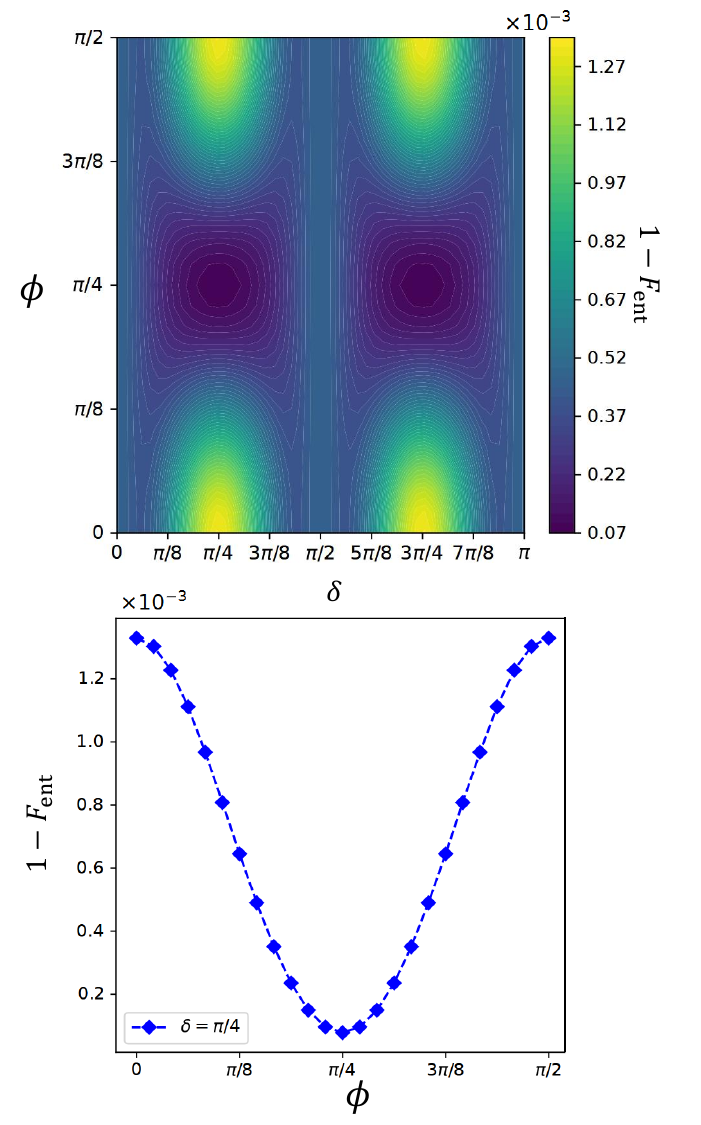}
    \caption{Dependence of the performance of the DRCC-2 with beam-splitted code vectors Eq.\eqref{eq:DRCC-20} under dephasing noise on the beam-splitter angles $\delta$ and  $\phi$. Here, we consider $\alpha=1.538$ and a dephasing strength $ \gamma_it =10^{-3}$ in both modes.}
    \label{fig:cat-4-deph}
\end{figure}

Dual-rail codes with an inner four-component cat code, which can intrinsically correct single-photon loss errors on the individual modes, were previously studied in Ref.~\cite{Plenio_dual_rail}. Here, we show that such codes can correct dephasing and photon-loss errors more effectively than their single-mode counterparts while also suppressing coherent errors~\cite{Constant_excitation_code_Ouyang_2021}, when we consider the logical states obtained by mixing the modes with a beam-splitter operation $\hat U_{\rm BS}$, yielding Eq.\eqref{eq:DRCC-20} (see also Ref.~\cite{multimode_RSB}).

We start by analysing the KL conditions for DRCC-2. The KL conditions can be studied using the same first-order photon-loss and first-order dephasing operators considered for the DRCC-1. With the first-order photon loss and dephasing operators and for the BS angle $\delta= \phi=\pi/4$, we obtain the following relations 
\begin{align}
\hat P_4\hat a_1^\dagger \hat a_1\hat P_4
&=
\hat{U}_{\rm BS}\hat{\bar P }_4(\hat U_{BS}^\dagger\hat a_1^\dagger \hat a_1\hat U_{BS} )\hat{\bar P}_4 \hat{U}_{\rm BS}^{\dagger} \nonumber\\
&=
\frac{1}{2}\hat U_{\rm BS}\hat {\bar P}_4 (\hat a_1^\dagger+\hat a_2^\dagger)(\hat a_1 +\hat a_2)\hat{\bar P}_4 \hat U_{\rm BS}^{\dagger} \nonumber\\
&=
\frac{1}{2}(\langle \hat n_1\rangle + \langle \hat n_2\rangle)\hat P_4, \label{eq:n_1_DRCC-2}\\
\hat P_4\hat a_2^\dagger \hat a_2\hat P_4
&=
\frac{1}{2}(\langle \hat n_1\rangle +\langle \hat n_2\rangle)\hat P_4, \label{eq:n_2_DRCC-2}\\
\hat P_4\hat a_1^\dagger \hat a_2\hat P_4
&=
\frac{1}{2}(\langle \hat n_1\rangle - \langle \hat n_2\rangle) \hat Z_L, \label{eq:a1a2_DRCC-2}\\
\hat P_4 \hat n_1\hat n_2 \hat P_4 & \propto \hat P_4 \label{eq:n1n2_DRCC-2}
\end{align}
where  $\bar P_4 = |C^+_4\rangle\langle C^+_4|+ |C^-_4\rangle\langle C^-_4|$ is the projector onto the DRCC-2 codespace, and $\hat Z_{L}$ denotes the logical Pauli-$\hat Z$ operator acting on the DRCC-2 subspace. To arrive at Eq.\eqref{eq:n_1_DRCC-2} we have also used that $\bar P_4 \hat a_i^{\dagger}\hat a_j\bar P_4=0$ for $i \neq j$.

The relation in Eq.\eqref{eq:a1a2_DRCC-2} shows that photon-loss events acting on different modes map the codespace to mutually orthogonal subspaces, and also to subspaces orthogonal to the original codespace, when
\begin{align}
\langle \hat n_1\rangle - \langle \hat n_2\rangle =0.
\end{align}
Equivalently, since $\langle \hat n_1\rangle - \langle \hat n_2\rangle$ is proportional to ${\rm Tr}(\hat Z_L \hat n_1)$, the above condition can be written as
\begin{align}\label{eq:sweet_spot}
{\rm Tr}(\hat Z_L \hat n_1) =0.
\end{align}
Furthermore, we note from Eqs. \eqref{eq:n_1_DRCC-2}, \eqref{eq:n_2_DRCC-2}, and \eqref{eq:n1n2_DRCC-2} that, for the beam-splitter parameters $\delta=\phi=\pi/4$, the DRCC-2 can satisfy the relevant QEC conditions for dephasing errors. 

To evaluate the performance of the DRCC-2 against photon-loss and dephasing noise, we use the Petz recovery map. We demonstrate the performance under photon-loss noise in Fig.~\ref{fig:ent_plot_diff_rec}. We observe that the infidelity reaches a minimum near $\alpha \approx 1.5$. A similar optimal value of $\alpha$ has also been observed for the single-mode four-component cat code in Refs.~\cite{VV_albert_bosonic_performance, Albert_pair_cat}. Around this value of $\alpha$, the condition ${\rm Tr}(\hat Z_L \hat n_1)=0$ is also satisfied. Therefore, near $\alpha \approx 1.5$, the DRCC-2 satisfies the QEC conditions for the photon-loss channel.

Similar to DRCC-1, a BS operation that mixes the two modes, as in Eq.\eqref{eq:m-RSB}, does not play any role in protecting against photon-loss errors. However, for dephasing noise occurring before the beam-splitter, in analogy with what was observed for dual-rail binomial codes ~\cite{multimode_RSB}, the beam-splitter operation plays an important role in improving the performance of the code. As shown in Fig.~\ref{fig:cat-4-deph}, the infidelity reaches a minimum for $
\delta=\phi=\pi/4.$

In summary, under the conditions $\alpha \approx 1.5$, $\delta=\phi=\pi/4$, the DRCC-2 achieves improved performance against both the loss and dephasing noise channels for sufficiently small noise strengths. A full numerical analysis with DRCC-2 concatenated with an outer code, e.g, the repetition code, is left for future work.

\section{Discussion and outlook}\label{sec:discussion}

In this work, we have introduced and analysed dual-rail cat codes (DRCCs), a family of bosonic encodings that combine the bias-preserving properties of cat codes with the leakage-detection capabilities of dual-rail architectures. We investigated DRCCs constructed from both two- and four-component cat states and characterised their error-correction properties under the dominant bosonic noise processes of photon loss and dephasing.

For the two-component construction, DRCC-1, we showed that single-photon-loss events are converted from logical errors into detectable leakage events, while dephasing errors become asymptotically correctable in the large-$\alpha$ limit. For the four-component construction, DRCC-2, we identified parameter regimes in which both single-photon loss and first-order dephasing errors satisfy the Knill--Laflamme conditions. We further developed a set of logical operations for DRCC-1, including single-qubit rotations, entangling gates, and controlled operations that preserve the noise bias.

A central result of this work is the demonstration that leakage can be treated deterministically within the DRCC framework. By concatenating DRCC-1 with an outer repetition code, detectable leakage events can be converted into erasure information, enabling efficient correction. This approach exploits a key distinction between DRCCs and conventional cat codes: single-photon loss  generates leakage rather than logical Pauli errors. As a consequence, DRCCs can benefit from lower-overhead repetition-code constructions while maintaining compatibility with dissipative stabilisation. In particular, the parity-check operation used for syndrome extraction commutes with the stabilising dynamics, thereby allowing continuous protection of the encoded state during error correction.

More broadly, we identified a structural property shared by dual-rail rotationally symmetric bosonic codes: the absence of intrinsic geometric phase accumulation during logical operations. In contrast to single-mode cat codes, repetition cat codes, pair-cat codes, and several other rotationally symmetric bosonic encodings, DRCCs do not acquire relative geometric phases that must subsequently be tracked or compensated. This feature simplifies the implementation of fault-tolerant logical gates.

These results establish DRCCs as a promising bosonic code for hardware-efficient quantum error correction. By simultaneously providing bias-preserving logical operations, deterministic leakage detection and correction that does not require entanglement with a two-level auxiliary system, compatibility with continuous stabilisation, and freedom from intrinsic geometric-phase accumulation, DRCCs offer a distinct route toward scalable fault-tolerant quantum computation in bosonic architectures.

An important open question arising from this work concerns the experimental realisation of the constraints discussed in Sec.~\ref{sec:dissipator} for stabilising the DRCC-1 using the dissipator $\hat \cD$ in Eq.~\eqref{eq:disp_dyn_2}. Ref.~\cite{maryam_diss} demonstrated how $\hat \cD$ can be engineered in a circuit quantum electrodynamics architecture. Extending that framework to incorporate the additional constraints required for DRCC-1 remains an important challenge.

We also found that the dissipator $\hat \cD$ in Eq.~\eqref{eq:disp_dyn_2} stabilises a broad class of codes, Eq.(\ref{eq:reccurence}), which exhibits features reminiscent of both the dual-rail cat code and the pair-cat code~\cite{Albert_pair_cat}, see also footnote [48]. A systematic characterisation of the error-correcting capabilities of such codes under photon-loss and dephasing noise would be of considerable interest. Such a study could include an analysis of its logical gate set, fault-tolerance properties, and performance when concatenated with outer codes. 

Although we have identified a universal set of logical gates compatible with continuous stabilisation, their gate times exhibit different dependences on the cat-state amplitude. For a fixed interaction strength, the gate time of the logical $\hat R_X$ rotation decreases with increasing $\alpha$, whereas those of the logical $\hat R_Y$ and $\hat R_Z$ rotations increase because the corresponding projected couplings are exponentially suppressed. For example, at $\alpha=1.5$, which lies within the optimal operating regime of the DRCC-1 under the Petz recovery map for the noise strengths considered, the gate times are $t_X\sim 0.35/K_{\rm BS}$ and $t_Y\sim t_Z\sim 15.7/K_{\rm BS}$ for the beam-splitter-based construction discussed in Sec.~\ref{sec:1q-gates}. As quantified in Appendix~\ref{sec:leak_app}, these gate Hamiltonians also contain off-codespace contributions. In particular, the fractional off-codespace Hamiltonian weight of the beam-splitter-based $\hat R_Z$ gate approaches unity in the large-$\alpha$ limit, implying that the dissipative stabilisation rate must be much larger than $K_{\rm BS}$ to confine the dynamics to the codespace. The alternative construction presented in Appendix~\ref{sec:proof_gate} provides a faster implementation of the logical Pauli $\hat Z_L$ gate without requiring the dissipator to remain active during the gate operation, but it does not allow the implementation of fast arbitrary $\hat R_Z(\theta)$ rotations. This limitation is addressed in Appendix~\ref{sec:fast_Rz}, which introduces an auxiliary-qubit-assisted implementation of fast arbitrary $\hat R_Z(\theta)$ rotations. In this approach, the dissipative stabilisation must be turned off during the controlled individual-parity operations, and errors in the auxiliary qubit may affect the gate performance. Achieving fast logical rotations while maintaining continuous stabilisation remains an open challenge.

As a further point for future extensions, in our construction of bias-preserving parity measurements for DRCC-1 we considered a two-level auxiliary system. A natural extension is to investigate whether replacing such an auxiliary qubit with a bosonic mode can further enhance noise-bias preservation by reducing the probability that auxiliary-system errors propagate to the encoded oscillators~\cite{Yale_dual,mehta2025bias}. More broadly, it would be worthwhile to explore alternative noise-protected auxiliary architectures, including pairs of giant atoms~\cite{Ariadna_Anton}, for syndrome extraction.

Also, while the Knill-Laflamme conditions have been analysed for general loss rates, it would be interesting to also perform a numerical analysis of the code's performance in the case of unequal loss rate in the two modes defining the codewords. 

Beyond repetition codes, an important avenue is concatenating DRCCs with more powerful outer codes, including low-density parity-check (LDPC) codes, surface codes, and high-rate constructions such as elevator codes~\cite{shanahan2026elevatorcodesconcatenationresourceefficient}. The logical gate and measurement protocols developed here may serve as useful primitives for constructing scalable, fault-tolerant architectures based on such concatenated schemes.

Finally, we analysed the quantum error-correcting properties of the DRCC-2, which can correct photon-loss errors without concatenation and can further suppress dephasing errors in the presence of a beam splitter. Several key aspects of this code remain unexplored, including stabilisation mechanisms, noise-bias preservation, and logical gate implementations. Assessing these properties, studying concatenated DRCC-2 architectures, and benchmarking their performance against related bosonic encodings, such as pair-cat codes and four-component single-mode cat codes~\cite{ruiz_ldpc_4cat}, constitute promising directions for future research.

\section{Acknowledgement}
We acknowledge Armelle Célarier, Axel Eriksson, Benjamin Huard, David DiVincenzo, Gerard Milburn, Ingrid Strandberg, Lukas Splitthoff, Martin Jirelow, Maryam Khanahmadi, Mazyar Mirrahimi, Niccolò Laurora, Rabsan Galib Ahmed, Steven Girvin and Timo Hillmann for useful discussions and comments on our manuscript. G.F.\ acknowledges funding from the European Union’s Horizon Europe Framework Programme (EIC Pathfinder Challenge project Veriqub) under Grant Agreement No.\ 101114899 and the Swedish Research Council through the project grant VR DAIQUIRI. G.F., A.U.,  R.W. and M.G acknowledge funding from the Knut and Alice Wallenberg Foundation through the Wallenberg Center for Quantum Technology (WACQT). G.F. and D.B. acknowledge funding from the Olle Engkvist Foundation. N. S. and A. S. acknowledge funding from Chalmers Area of Advance Nano, as well as resources at the Chalmers Centre for Computational Science and Engineering (C3SE). {We also acknowledge that we have used an AI tool such as ChatGPT-5.5  to correct language errors, improve the logical flow of some paragraphs and for some LaTeX commands, and also for the suggestion on the caption of some figures. We also declare that neither the calculations nor the plots were generated by any AI tools.}

\bibliography{refs,refs2}
\appendix

\section{DRCC-1 with beam-splitted modes against pure dephasing errors}\label{sec:DRCC-bs}
The BS operation $\hat U_{\text{BS}}$ commutes with photon loss errors, resulting in no improvement in the QEC  performance against photon losses. In this Appendix, we instead consider dephasing noise and analyse the QEC properties of the DRCC-1 when we allow a beam-splitter operation between the two modes. We consider, therefore, the following code vectors, which belong to the two-mode RSB code in Eq.\eqref{eq:m-RSB} and are obtained by fixing the $f_{mn}$ coefficients as in Eq.\eqref{eq:f_mn_cat} with $N=1$:
\begin{align}\label{eq:drcc-bs-1}
\begin{split}
|0_L\rangle & = \hat U_{\BS} |\bar{0}\rangle,\\
|1_L\rangle & = \hat U_{\BS} |\bar{1}\rangle.
\end{split}
\end{align}

The dephasing Kraus operators acting on the DRCC-1 qubits, up to the first order in the noise strength $\gamma_i t$, are given in Eqs.\eqref{eq:d0}, \eqref{eq:d1}, and \eqref{eq:d2}.

To check the KL condition for the first-order dephasing noise, we first study the transformation of the number operators $\hat{n}_i$ under $\hat U_{\BS}$. Using Eq.\eqref{eq:BS_transformation}, the number operators acting on different modes transform as
\begin{align}\label{eq:b_1_dag_b_1}
\hat b_1^{\dagger}\hat b_1 &= \hat n_1 \cos^2\delta + \hat n_2 \sin^2\delta\nonumber\\
& +(e^{-j \phi }\hat a_1^{\dagger}\hat a_2 + e^{j \phi} \hat a_2^{\dagger}\hat a_1 ) \sin{\delta}\cos\delta,\\
\label{eq:b_2_dag_b_2}
\hat b_2^{\dagger}\hat b_2 &=\hat  n_2 \cos^2\delta + \hat n_1 \sin^2\delta \nonumber\\
&  -(e^{-j \phi }\hat a_1^{\dagger}\hat a_2 + e^{j \phi} \hat a_2^{\dagger}\hat a_1 ) \sin{\delta}\cos\delta.
\end{align}

Therefore, by analysing the action of $\hat E_k^{D\dagger}\hat E_l^D$ up to the first order in $\gamma_i t$ on the codespace defined by the code vectors in Eq.\eqref{eq:drcc-bs-1}, we obtain
\begin{align}\label{eq:n_100}
\langle 0_L|\hat{n}_1|0_L\rangle
&= \left[\cos^2\delta \left(\frac{N_-}{N_+}\right)^2+
\sin^2\delta \left(\frac{N_+}{N_-}\right)^2\right] |\alpha|^2,\\
  \label{eq:n_111}
  \langle 1_L|\hat{n}_1|1_L\rangle
  &= \left[\cos^2\delta \left(\frac{N_+}{N_-}\right)^2 +\sin^2\delta \left(\frac{N_-}{N_+}\right)^2\right] |\alpha|^2,\\
  \label{eq:n_101}
  \langle 1_L|\hat{n}_1|0_L\rangle
  &= \frac{\sin 2\delta}{2}
  \left[e^{-j \phi}\left(\frac{N_+}{N_-}\right)^2 +e^{j \phi}\left(\frac{N_-}{N_+}\right)^2 \right]|\alpha|^2,\\
  \label{eq:n_110}
  \langle 0_L|\hat{n}_1|1_L\rangle
  &= \frac{\sin 2\delta}{2}
  \left[e^{-j \phi}\left(\frac{N_-}{N_+}\right)^2+e^{j \phi}\left(\frac{N_+}{N_-}\right)^2 \right]|\alpha|^2.
  \end{align}

To obtain the above expressions, we use the transformations in Eqs.\eqref{eq:b_1_dag_b_1} and \eqref{eq:b_2_dag_b_2}. We note that, for $\delta=\pi/4$, the diagonal matrix elements satisfy
\begin{align}
\langle0_L|\hat{n}_1|0_L\rangle = \langle1_L|\hat{n}_1|1_L\rangle.
\end{align}
To satisfy the KL condition, we further require
\begin{align}
\langle1_L|\hat{n}_1|0_L\rangle
=
\langle0_L|\hat{n}_1|1_L\rangle
=0.
\end{align}
For $\phi=\pi/2$, this off-diagonal term becomes exponentially suppressed with increasing $\alpha$, as $N_+ \to N_-$. Therefore, at finite $\alpha$, the DRCC-1, with or without the BS operation, does not exactly satisfy the KL condition and hence cannot exactly correct first-order dephasing errors.

Regarding error detectability, we note from Eqs.\eqref{eq:n_100}-\eqref{eq:n_110} that, in the small-$\alpha$ limit, the code in Eq.\eqref{eq:drcc-bs-1} cannot detect single-mode first-order dephasing noise either. 

\section{Effects of the higher-order photon loss and dephasing errors for DRCC-1}\label{sec:noise_bias}

In Sec.~\ref{sec:QEC_properties}, we showed that DRCC-1 does not satisfy the full KL conditions for photon loss. Nevertheless, the structure of the projected error operators reveals a strong effective noise bias, which is crucial for choosing an appropriate outer code.

Like the standard dual-rail code in Eq.\eqref{eq:dual_rail_stdrd}, DRCC-1 also converts the single-photon loss acting on the individual modes into erasures. However, the standard dual-rail code does not convert higher-order photon-loss events into detectable leakage. Here, we investigate the effects of higher-order photon loss and dephasing errors on the DRCC-1 qubits and show that odd-order photon loss or gain in a single mode leads to leakage. The joint photon loss events either act as a logical identity operation onto the codespace or cause a $\hat X_L$ in the large $\alpha$ limits. The higher-order dephasing errors induce an exponentially suppressed logical $\hat Z$ to the DRCC-1 codespace in the large $\alpha$ limits.

\subsection{Photon loss}

The Knill--Laflamme (KL) conditions for single-photon loss are given in Eqs.~\eqref{eq:Pa_iP}--\eqref{eq:n2_w/obs}. More generally, for an $\ell$-th order photon loss acting on the $i$-th mode, one finds
\begin{align}\label{eq:Pa^lP}
   \hat{ \bar{P}}\hat a_i^{\ell}\hat{\bar{P}}
    \propto
    \begin{cases}
        0,
        & \ell \ \text{odd},
        \\
        \hat{\bar{P}},
        & \ell \ \text{even}.
    \end{cases}
\end{align}

We now consider joint photon-loss processes acting on both cavities. For $\ell$-th order loss in mode 1 and $k$-th order loss in mode 2, we obtain
\begin{align}\label{eq:loss-gain}
    \hat{\bar{P}}
    \hat a_1^{\dagger\ell}
   \hat a_2^k
    \hat {\bar{P}}
    \propto
    \begin{cases}
       \hat{  \bar{P}},
        & \ell+k \ \text{even},
        \\       
        (A_X \hat X_L+iA_Y\hat Y_L),
        & \ell+k \ \text{odd} \ \&\, k \neq 0 ,
    \end{cases}\\
  \label{eq:loss-loss}   
    \bar{P}
    \hat a_1^{\ell}
   \hat  a_2^k
    \bar{P}
    =
    \begin{cases}
        |\alpha|^{\ell+k}\bar{P},
        & \ell \, \& \, k \ \text{even},
        \\
        |\alpha|^{\ell+k-2}
        (A_X\hat X_L+iA_Y\hat Y_L),
        & \ell \,\& \, k \ \text{odd}.
    \end{cases}
\end{align}

These relations show that, in the large-$\alpha$ regime, joint photon-loss events predominantly induce logical $\hat X$ errors, while logical $\hat Y$ errors are exponentially suppressed due to Eq.~\eqref{eq:A_Z}. 

Furthermore, Eq.~\eqref{eq:loss-gain} provides insight into photon-gain processes. In particular, for $k=0$, it describes the effect of an $\ell$-photon gain error on the first cavity. For even $\ell$, such gain errors preserve the logical subspace of the DRCC qubit. In contrast, odd-order photon gain on the $i$-th cavity leads to leakage, as can be verified from
\[
\bar{P}\hat a^{ \dagger\ell}_i\bar{P}
= \alpha^{\ell-1}(\bar{P}\hat a^{\dagger}\bar{P}) = 0.
\]

Overall, these results demonstrate that the DRCC-1 exhibits a strong noise bias towards logical $\hat X$ errors under higher-order photon-loss processes, while simultaneously suppressing logical $\hat Y$ errors in the large $\alpha$ limit. We note that the DRCC-1 also converts odd-order photon-gain events into leakage outside the logical subspace, and hence facilitates us to flag them as an erasure. 

\subsection{Dephasing}

We have seen, by analysing the KL condition for first-order dephasing in either cavity, that dephasing errors induce a logical $\hat Z$ error in the DRCC-1 qubit. 

To analyse higher-order dephasing processes, we use the following identity
\begin{align}\label{eq:n^l}
    \hat n^{\ell}
    =
    \sum_{p=1}^{\ell}
    S(\ell,p)
    \hat a^{\dagger p}\hat a^p,
\end{align}
where
\begin{align}
    S(\ell,p)
    =
    \frac{1}{p!}
    \sum_{k=0}^{p}
    (-1)^{p+k}
    \binom{p}{k}
    k^{\ell}
\end{align}
are the Stirling numbers of the second kind~\cite{normal_ordering_2003}. Separating even and odd powers gives
\begin{align}\label{eq:n^l_2}
   \hat  n_i^{\ell}
    &=
    \sum_m
    S(\ell,2m)
    \hat a^{\dagger 2m}\hat a^{2m}
   +
    \sum_m
    S(\ell,2m+1)
    \hat a^{\dagger(2m+1)}\hat a^{2m+1}.
\end{align}

Using Eq.~\eqref{eq:Pa^lP}, we obtain
\begin{align}
    \bar{P}\hat n_i^{\ell}\bar{P}
    &=
    \bar{P}
    \sum_m
    S(\ell,2m)
    |\alpha|^{4m}
    \nonumber\\
    &\quad+
    (A_0\bar{P}
    +
    (-1)^{i\oplus1}A_Z\hat Z_L)
    \sum_m
    S(\ell,2m+1)|\alpha|^{4m},
\end{align}
where $\oplus$ denotes modulo-two addition. Therefore, for small $\alpha$, higher-order dephasing processes induce residual logical $\hat Z$ noise. However, in the large limit $\alpha$, the logical contribution $\hat Z$ becomes exponentially suppressed, and the projected operators become approximately proportional to the identity due to Eq.\eqref{eq:A_Z}.

\section{Alternative dissipation stabilisation}
\label{sec:alt_diss}

In the main text, we discussed the dissipative stabilisation of the DRCC-1 using the dissipator in Eq.~\eqref{eq:disp_dyn_2}. In this appendix, we present two alternative dissipative-stabilisation schemes that can also stabilise the DRCC-1 manifold.

We first exploit the observation in Eq.~\eqref{eq:cat_eig_state}, namely that the logical states are simultaneous eigenstates of the operators $\hat a_i^2$. This property allows us to stabilise the DRCC-1 manifold through the dissipative dynamics
\begin{align}\label{eq:disp_dyn_0}
\dot{\hat{\rho}}
&=
-j[\hat{H}_{\rm sys},\hat{\rho}]
+d_1^2\,\hat{\cD}[\hat{a}_1^2-\alpha^2]\rho
+d_2^2\,\hat{\cD}[\hat{a}_2^2-\alpha^2]\rho,
\end{align}
where $d_i$ denotes the dissipation rate associated with the $i^{\rm th}$ mode. In this approach, the two dissipators independently stabilise the individual cavity modes into cat states.

We can also stabilise the same manifold using the single dissipator
\begin{align}\label{eq:disp_dyn_1}
    \mathcal{\hat{D}}[\hat{a}_1^2+\hat{a}_2^2-2\alpha^2].
\end{align}
Using Eq.~\eqref{eq:cat_eig_state}, we find that the dissipator in Eq.~\eqref{eq:disp_dyn_1} stabilises the two-mode cat manifold in Eq.~\eqref{eq:cat_manifold}, provided that
\begin{align}
    \alpha_1^2+\alpha_2^2
    =
    2\alpha^2.
\end{align}

The stabilisation mechanisms underlying Eqs.~\eqref{eq:disp_dyn_0} and \eqref{eq:disp_dyn_1} differ from each other. Equation~\eqref{eq:disp_dyn_0} employs two independent dissipators that separately stabilise each cavity mode, whereas Eq.~\eqref{eq:disp_dyn_1} uses a single collective dissipator acting on both modes simultaneously. Consequently, Eq.~\eqref{eq:disp_dyn_1} provides a single-dissipator route to generating the DRCC-1 manifold.

 \section{Alternative logical gates for DRCC}\label{sec:proof_gate}

In Sec.~\ref{sec:logical gates}, we introduced Hamiltonian-based logical gates for the DRCC-1 qubits. However, that construction requires the dissipator to remain active during the gate operation. Here, we consider an alternative approach for constructing logical single-qubit gates and propose explicit unitary operations that implement the logical $\hat Y_L$, $\hat Z_L$, $\hat S$, and $\hat T$ gates within the DRCC-1 codespace. Unlike Sec.\ref{sec:logical gates}, in this alternative construction, we do not need the dissipator to remain active during the gate operation.

\subsection{Logical $Y$}

We show that the operator $\hat G_{12}^{-}$ generates the logical gate $\hat X_L$ for even $N$ and the logical gate $i\hat Y_L$ for odd $N$ in the multimode RSB code. We begin by proving the following lemma.

\begin{lemma}
For the operator $\hat G_{12}^{-}$, the following relation holds:
\begin{align}
e^{ j\pi \hat G_{12}^{-}/2}\hat a_i^{\dagger}e^{ -j\pi \hat G_{12}^{-}/2}
=
(-1)^i \hat a_{i+1}^{\dagger}.
\end{align}
\end{lemma}

\begin{proof}
Using the Baker--Campbell--Hausdorff (BCH) formula for two operators $(\hat U,\hat V)$,
\begin{align}
e^{\chi \hat U} \hat V e^{-\chi \hat U}
=
\hat V + \chi [\hat U,\hat V]
+\frac{\chi^2}{2!}[\hat U,[\hat U,\hat V]]
+ \cdots,
\end{align}
we choose
\begin{align}
\hat U = \hat G_{12}^{-}, \qquad
\hat V = \hat a_1^{\dagger}, \qquad
\chi = j\theta.
\end{align}

Using the bosonic commutation relation $[\hat a,\hat a^{\dagger}]=\hat I$, we obtain
\begin{align}
[\hat U,\hat V] &= j \hat a_2^{\dagger},\
[\hat U,[\hat U,\hat V]] = - \hat a_1^{\dagger}.
\end{align}
Substituting these relations into the BCH expansion gives
\begin{align}
e^{ j\theta \hat G_{12}^{-}}
a_1^{\dagger}
e^{ -j\theta \hat G_{12}^{-}}
=
\cos{\theta} \hat a_1^{\dagger}
-
\sin{\theta} \hat a_2^{\dagger}.
\end{align}
Similarly, choosing $\hat V=\hat a_2^{\dagger}$ gives
\begin{align}
e^{ j\theta \hat G_{12}^{-}}
\hat a_2^{\dagger}
e^{ -j\theta \hat G_{12}^{-}}
=
\cos{\theta} \hat a_2^{\dagger}
+
\sin{\theta} \hat a_1^{\dagger}.
\end{align}
Setting $\theta=\pi/2$ proves the lemma.
\end{proof}

Using this lemma, for odd $N$ we obtain
\begin{align}
e^{ j\pi \hat G_{12}^{-}/2}
\hat a_1^{\dagger 2mN}
\hat a_2^{\dagger (2n+1)N}
e^{ -j\pi \hat G_{12}^{-}/2}
&=
\hat a_2^{\dagger 2mN} \hat a_1^{\dagger (2n+1)N},\\
e^{ j\pi \hat G_{12}^{-}/2}\hat a_2^{\dagger 2mN}\hat a_1^{\dagger(2n+1)N}e^{ -j\pi \hat G_{12}^{-}/2}
&=
-\hat a_1^{\dagger 2mN} \hat a_2^{\dagger (2n+1)N}.
\end{align}
For even $N$, we instead obtain
\begin{align}
e^{ j\pi \hat G_{12}^{-}/2}
\hat a_1^{\dagger 2mN}
\hat a_2^{\dagger (2n+1)N}
e^{ -j\pi \hat G_{12}^{-}/2}
&=
\hat a_2^{\dagger 2mN} \hat a_1^{\dagger (2n+1)N},\\
e^{ j\pi \hat G_{12}^{-}/2}
\hat a_2^{\dagger 2mN}
\hat a_1^{\dagger (2n+1)N}
e^{ -j\pi \hat G_{12}^{-}/2}
&=
\hat a_1^{\dagger 2mN} \hat a_2^{\dagger (2n+1)N}.
\end{align}

Therefore, for odd $N$,
\begin{align}
\hat U_{\rm BS}
e^{ j\pi \hat G_{12}^{-}/2}
\hat U_{\rm BS}^{\dagger}
|i_L\rangle
=
(-1)^i |(i+1)_L\rangle.
\end{align}
Consequently, in the logical basis,
\begin{align}
\hat U_{\rm BS}
e^{ j\pi \hat G_{12}^{-}/2}
\hat U_{\rm BS}^{\dagger}
=
\begin{pmatrix}
0 & 1 \\
-1 & 0
\end{pmatrix}
=
j\hat Y_L.
\end{align}

For even $N$, we instead have
\begin{align}
\hat U_{\rm BS}
e^{ j\pi \hat G_{12}^{-}/2}
\hat U_{\rm BS}^{\dagger}
|i_L\rangle
=
|(i+1)_L\rangle,
\end{align}
which corresponds to the logical operator $\hat X_L$.

\subsection{Logical $\hat Z$}

We now discuss an alternative implementation of the logical $\hat Z_L$ gate and its extension to arbitrary rotational symmetry order $N$. In particular, we revisit the gate construction introduced in Ref.~\cite{multimode_RSB}, but with a code vector that excludes the beam splitter. 

\begin{lemma}\label{lem:ZL_alt}
For a two-mode RSB code with arbitrary rotational symmetry order $N$, the logical operator $\hat Z_L$ can be expressed as
\begin{align}\label{eq:Z_LN}
\hat Z_L^{(N)}
= e^{j\pi \hat n_1/N}.
\end{align}
\end{lemma}

\begin{proof}
We examine the action of the operator in Eq.~\eqref{eq:Z_LN} on the two-mode RSB code vectors. Acting on $|0_L\rangle$, we obtain
\begin{align}
\hat Z_L^{(N)}|\bar 0\rangle
&=\sum_{m,n}
f_{mn}
e^{2jm\pi}
|2mN\rangle |(2n+1)N\rangle \nonumber\\
&=|\bar 0\rangle,
\end{align}
since $e^{2jm\pi}=1$. Similarly, for $|1_L\rangle$,
\begin{align}
\hat Z_L^{(N)}|\bar 1\rangle
&=\sum_{m,n}
f_{mn}
e^{j(2n+1)\pi}|(2n+1)N\rangle |2mN\rangle \nonumber\\
&=-|\bar 1\rangle,
\end{align}
because $e^{j(2n+1)\pi}=-1$. Thus, the operator $\hat Z_L^{(N)}$ acts as the logical Pauli-$Z$ operator for arbitrary $N$.
\end{proof}

\subsection{$\hat S$ and $\hat T$ gate}
Deviating from the gate teleportation protocol for generating the $\hat {S} $ and $\hat {T} $ gates discussed in~\cite{multimode_RSB}, in this Appendix we explore other possible methods to realize these gates. We investigate whether the unitary
\begin{align}
\hat U_G& = e^{j\theta \hat n_1} 
\end{align}
can generate other logical single-qubit gates. Previously, we showed that choosing $\theta=\pi$ implements the logical $\hat Z_L$ gate. Here, we further examine whether the choices $\theta=\pi/2$ and $\theta=\pi/4$ can yield the logical $\hat {S} $ and $\hat {T} $ gates, respectively.

To determine whether $\hat U_G$ acts as a logical gate within the DRCC-1 codespace, we evaluate its action on the logical basis states:
\begin{align}
 e^{j\theta \hat n_1} |\bar 0\rangle
&=
\frac{1}{N_+}\left(
|\alpha e^{j\theta}\rangle
+
|-\alpha e^{j\theta}\rangle
\right)
|C^{-}\rangle,\\
 e^{j\theta \hat n_1} |\bar 1\rangle
&=\frac{1}{N_-}\left(|\alpha e^{j\theta}\rangle-|-\alpha e^{j\theta}\rangle\right)|C^{+}\rangle.
\end{align}

From these expressions,
\begin{align}
\langle \bar \mu|
 e^{j\theta \hat n_1} |\bar \nu\rangle=0,\qquad\mu\neq\nu.
\end{align}
For $\mu=\nu$, we obtain

\begin{align}
\langle \bar \mu| e^{j\theta \hat n_1}|\bar{\mu}\rangle=
\begin{cases}
\frac{2e^{-|\alpha|^2}}{N_+^2}
\cosh\left(|\alpha|^2 e^{j\theta}\right),
& \mu=0,\\ \frac{2e^{-|\alpha|^2}}{N_-^2}
\sinh\left(|\alpha|^2 e^{j\theta}\right),
& \mu=1.
\end{cases}
\end{align}

In particular, for $\theta=\pi$,
\begin{align}
\langle \bar \mu|
 e^{j\pi \hat n_1} 
|\bar \mu\rangle
=\begin{cases}
1, & \mu=0,\\
-1, & \mu=1.
\end{cases}
\end{align}
Thus, for $\theta=\pi$, the operator $e^{j\theta \hat n_1} $
implements the logical $\hat Z$ gate.

More generally, to see whether the operator $e^{j \theta \hat n _1}$ produces a $\hat R_Z$ rotation onto the DRCC-1 subspace, we look into the projected action of it onto the codespace
\begin{align}
\hat {\bar P} e^{j\theta \hat n_1} \hat {\bar P}=f(\alpha,\theta)\begin{pmatrix}
1 & 0 \\
0 & \Phi(\alpha,\theta)
\end{pmatrix},
\end{align}
where $f(\alpha,\theta)$ is a global phase factor and
\begin{align}
\Phi(\alpha,\theta)
=\frac{\tanh\left(|\alpha|^2 e^{j\theta}\right)}{\tanh(|\alpha|^2)}.
\end{align}

We note that $|\Phi(\alpha,\theta)|\neq 1$ for arbitrary values of $\alpha$ and $\theta$. Therefore, $\hat U_G$ does not correspond to a logical gate for all choices of parameters. However, when $\theta=\pi/2$ and
\begin{align}
|\alpha|
\approx
{1.98, 2.656, 4.06,\ldots},
\end{align}
The operator realises the logical $\hat S$ gate. Furthermore, for small coherent-state amplitudes, namely $\alpha\lesssim 0.1$, choosing $\theta=\pi/4$ implements the logical $\hat T$ gate, also known as the $\pi/8$ gate.

\section{Leakage induced by  Hamiltonians implementing logical gates}\label{sec:leak_app}

In Sec.\ref{sec:logical gates}, we considered the beam-splitter Hamiltonian in Eq.\eqref{eq:gen_H} and constructed arbitrary single-qubit gates on the DRCC-1 codespace by applying an additional beam-splitter transformation to $\hat H_{\rm BS}$, as shown in Eq.\eqref{eq:BS-hamiltonian}. However, $\hat H_{\rm BS}$ does not act entirely within the DRCC-1 codespace and may therefore map the code vectors outside the codespace. For example, a 50:50 beam-splitter action on $|\bar 0\rangle= |C^+_{\alpha}C^-_{\alpha}\rangle$ maps it to a state $\sim |C^{+}_{\sqrt{2}\alpha}\rangle|0\rangle$, which does not belong to the codespace. Thus, to suppress this type of leakage, the dissipative stabilisation rate must be much larger than the strength of the beam-splitter interaction. In this Appendix, we quantify the weight of the beam-splitter Hamiltonian projected outside the codespace relative to the total Hamiltonian weight acting on the codespace through the following fraction: 
\begin{align}\label{eq:leak_frac}
W_{\text{leak}} = \frac{||\bar Q \hat{\tilde{ H}}_{\text{BS}}\bar  P||_{\text{HS}}^2}{||\hat{\tilde{ H}}_{\text{BS}} \bar P||_{\text{HS}}^2},
\end{align}
where $ \bar Q= \hat I - \bar{P}$, and $\hat I$ is the identity operator defined on the infinite-dimensional Hilbert space of the two encoding oscillators. We refer to $W_{\text{leak}}$ as the fractional off-codespace Hamiltonian weight. The norm $||(.)||_{\rm HS}$ denotes the Hilbert--Schmidt norm and is defined for an operator $\hat O$ as
\begin{align*}
||\hat O||_{\rm HS} & = \sqrt{\tr (\hat O^{\dagger}\hat O)}.
\end{align*}

Expanding $||\bar Q \hat{\tilde{ H}}_{\rm BS}\bar P||^2_{\rm HS}$, i.e., the numerator of Eq.\eqref{eq:leak_frac}, in terms of the projector $\bar{P}$, we obtain
\begin{align}
||\bar Q\hat{\tilde{ H}}_{\rm BS} \bar P||_{\rm HS}^2& = \tr(\bar P \hat{\tilde{ H}}_{\rm BS} \bar Q \hat{\tilde{ H}}_{\rm BS} \bar P )\nonumber\\
& = \tr (\bar P \hat{\tilde{ H}}_{\rm BS}^2\bar P)-\tr (\bar P\hat{\tilde{ H}}_{\rm BS}\bar P\hat{\tilde{ H}}_{\rm BS} \bar P) \nonumber\\
& = ||\hat{\tilde{ H}}_{\rm BS} \bar P ||_{\rm HS}^2 - ||\bar P\hat{\tilde{ H}}_{\rm BS}\bar P||_{\rm HS}^2.
\end{align}

To calculate the fractional off-codespace Hamiltonian weight, we first evaluate $\langle \bar \mu | \hat{\tilde{H}}^2(\delta,\phi) |\bar \mu \rangle$. For this purpose, we note that
\begin{align}\label{eq:crosst-1}
\langle \bar \mu |  (A_{\delta,\phi}\hat a_1^{\dagger}\hat a_2 +h.c) (\hat{n}_2-\hat{n}_1) |\bar \mu \rangle &=0 \\
\label{eq:cross-2}   \langle \bar \mu | (\hat{n}_2-\hat{n}_1) (A_{\delta,\phi}\hat a_1^{\dagger}\hat a_2 +h.c)  |\bar \mu \rangle &=0.
\end{align}
The number operators ($\hat {n} _ {i} $) preserve the parity sectors of the individual modes of the DRCC-1 code vectors, whereas the term $(A_{\delta,\phi}\hat {a} _ {1} ^ {\dagger}\hat {a} _ {2} +h.c)$ flips the parity of both modes. Hence, the expectation values in Eqs.\eqref{eq:crosst-1}-\eqref{eq:cross-2} vanish.

The expectation value $\langle \bar \mu | \hat{\tilde{H}}^2(\delta,\phi) |\bar \mu \rangle$ therefore takes the following simplified form:
\begin{align}
\langle \bar \mu | \hat{\tilde{H}}^2(\delta,\phi) |\bar \mu \rangle & = K_{\rm BS}^2\bigg(\langle\bar \mu | (A_{\delta,\phi} \hat a_1^{\dagger}\hat a_2 +h.c )^2 |\bar \mu \rangle \nonumber\\
& \qquad +\sin^2(2\delta)\cos^2\phi\langle \bar \mu |(\hat n_2 -\hat n_1)^2|\bar \mu \rangle\bigg).
\end{align}
For the first term, $\langle\bar \mu | (A_{\delta,\phi} \hat a_1^{\dagger}\hat a_2 +h.c )^2 |\bar \mu \rangle$, we obtain
\begin{align}
\langle\bar \mu | (A_{\delta,\phi} \hat a_1^{\dagger}\hat a_2 +h.c )^2 |\bar \mu \rangle 
& = 4|\alpha|^4{\rm Re}(A_{\delta,\phi})^2 \nonumber\\
&+ 2 |\alpha|^2|A_{\delta,\phi}|^2 \coth(2|\alpha|^2).
\end{align}
For $\langle\bar \mu |(\hat n_2 -\hat n_1)^2|\bar \mu\rangle$, we have
\begin{align}
\langle \bar \mu |(\hat n_2 -\hat n_1)^2|\bar \mu\rangle & = 2 |\alpha|^2 \coth(2 |\alpha|^2).
\end{align}
Therefore, using $|A_{\delta,\phi}|^2+ \sin^2 2\delta\cos^2\phi = 1$, we obtain 
\begin{align}
||\hat {\tilde H}(\delta,\phi)\bar P ||^2_{\rm HS} & = 2K_{\rm BS}^2 (4 |\alpha|^4{\rm Re} (A_{\delta,\phi})^2 +2 |\alpha|^2\coth(2|\alpha|^2)).
\end{align}

To evaluate $\tr [(\bar P\hat{ \tilde H}\bar P)^2]$, we first note that $\tr (\hat X_L)= \tr (\hat Y_L)= \tr (\hat Z_L)=0 $ and $\tr (\hat X_L^2)= \tr (\hat Y_L^2)= \tr (\hat Z_L^2)=\tr \bar P=2$. Using these identities, we obtain
\begin{align}
\tr [(\bar P\hat{ \tilde H}\bar P)^2] & =8 K_{\rm BS}^2 |\alpha|^4({\rm Re}(A_{\delta,\phi})^2+\csch^2 (2 |\alpha|^2)).
\end{align}
Therefore, the off-codespace Hamiltonian weight is
\begin{align}
||\bar Q \hat{\tilde{H}}(\delta,\phi)\bar P||^2_{\rm HS} & = 4K_{\rm BS}^2 (|\alpha|^2\coth(2|\alpha|^2)\nonumber\\
&- 2|\alpha|^4\csch^2 (2|\alpha|^2)).
\end{align}

The above expression shows that the off-codespace Hamiltonian weight is independent of the angles $(\delta,\phi)$. Consequently, $ ||\bar Q \hat{\tilde{H}}(\delta,\phi)\bar P||^2_{\rm HS}$ is independent of the particular gate operation generated by the choice of $(\delta,\phi)$. However, the fractional off-codespace Hamiltonian weight $W_{\rm leak}$ depends on the angles $\delta,\phi$, and therefore on the corresponding gate operation, as follows:
\begin{align}
W_{\rm leak}(\delta,\phi)& = \frac{|\alpha|^2\coth(2|\alpha|^2)- 2|\alpha|^4\csch^2 (2|\alpha|^2)}{2 |\alpha|^4{\rm Re} (A_{\delta,\phi})^2 + |\alpha|^2\coth(2|\alpha|^2)}.
\end{align}

For the $\hat R_X$ gate, obtained for $\delta= 0,\phi =\pi/2$, the fractional off-codespace Hamiltonian weight can be written as
\begin{align}
W_{\rm {leak}}(0,\pi/2) =1-  \frac{2|\alpha|^4\coth^2 (2|\alpha|^2)}{2|\alpha|^4 + |\alpha|^2\coth (2|\alpha|^2)}\overset{\alpha \rightarrow \infty}{\longrightarrow}\frac{1}{1+2|\alpha|^2}.
\end{align}
For the $\hat R_Z$ gate, which is obtained by setting $\delta=\pi/4$ and $\phi=0$ in $\hat{\tilde{ H}}$, the fractional off-codespace Hamiltonian weight is given by
\begin{align}
    W_{\rm {leak}}(\pi/4,0) &= 1- \frac{2|\alpha|^2\csch^2(2|\alpha|^2)}{\coth{(2|\alpha|^2)}}.
 \end{align}
From the above expression, we note that, in the large-$\alpha$ limit, $W_{\rm leak}(\pi/4,0) \rightarrow 1$, indicating that the dissipation rate should be much larger than $K_{\rm BS}$ in order to confine the dynamics to the codespace. 

\section{Faster $\hat R_{Z}(\theta )$ gates} \label{sec:fast_Rz}

\begin{figure}
    \centering
    \includegraphics[width=1\columnwidth]{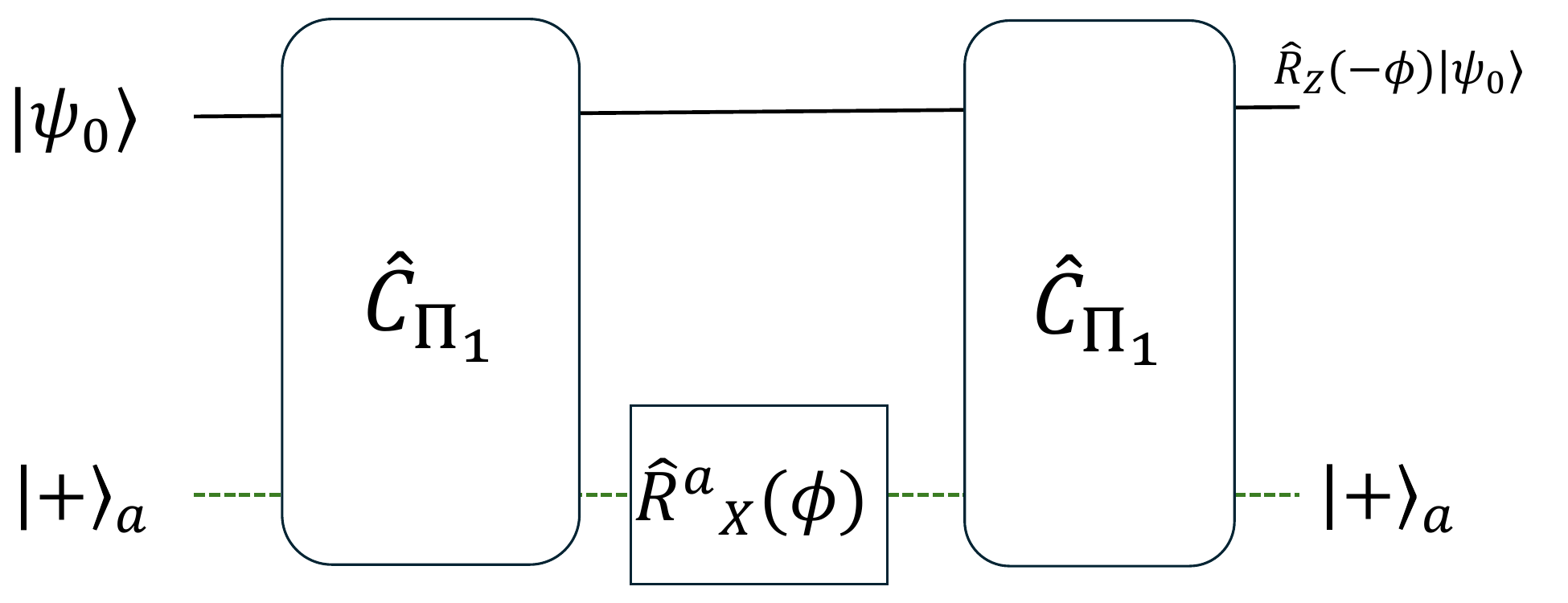}
    \caption{Auxiliary qubit-assisted implementation of the logical $\hat R_Z(\theta)$ gate on the DRCC-1 subspace. The auxiliary qubit is initialised in $|+\rangle_a$, and an auxiliary $\hat R_X^a(\theta)$ rotation is applied between two controlled individual-parity operations $\hat C_{\Pi_1}$. The sequence returns the auxiliary qubit to $|+\rangle_a$ while applying the required relative phase to the logical DRCC-1 state}
    \label{fig:fast_RZ}
\end{figure}
As noted in the main text and the preceding Appendices, implementing a faster $\hat R_Z(\theta)$ gate is challenging. In the Zeno-gate construction, the dissipation rate stabilising the DRCC-1 manifold must be much higher than $K_{\BS}$ because the $\hat R_Z(\theta)$ gate induces substantial leakage into the off-code space. Alternatively, the logical-$\hat Z$ gate can be implemented without dissipative stabilisation. However, this approach does not allow arbitrary rotations about the $Z_L$ axis. Here, we show that a fast $\hat R_Z(\theta )$ gate can be implemented on the DRCC-1 subspace using an auxiliary system. The implementation procedure is illustrated in Fig \ref{fig:fast_RZ}.

The gate sequence shown in Fig. \ref{fig:fast_RZ} can be understood as follows:

\begin{enumerate}
    \item First, we prepare the auxiliary system in the $|+\rangle$ state. The initial state of the joint system, consisting of the DRCC-1 system and the auxiliary system, is
    \begin{align}
        |\Psi_0\rangle & = |\psi\rangle|+\rangle_a,
    \end{align}
    where $|\psi\rangle = c_0 |\bar 0\rangle+ c_1 |\bar 1\rangle$. 
    \item We then apply the controlled operation 
    \begin{align}
        \hat C_{\Pi_1} & = \bar P\otimes |g\rangle\langle g| + e^{-j \pi \hat n_1} \otimes |e\rangle\langle e|.
    \end{align}
    
    \item Next we apply the $\hat R^a_X(\phi)$ gate -- an $\hat X$ rotation on the auxiliary qubit, and $\phi$ is the angle of rotation. The resulting state is then 
    \begin{align}
        |\Psi_1\rangle& =  c_0 |\bar 0\rangle e^{j \phi \sigma_x/2}|+\rangle_a+ c_1 |\bar 1\rangle e^{-j \phi \sigma_x/2} |-\rangle_a\nonumber\\
        & = e^{-j \phi /2}(c_0 |\bar 0\rangle|+\rangle_a+ c_1 |\bar 1\rangle e^{-j \phi} |-\rangle_a).
    \end{align}
    \item We then apply the operation $\hat C_{\Pi_1}$ once more, resulting in the state 
    \begin{align}
         |\Psi_2\rangle&=( c_0 |\bar 0\rangle+ c_1 e^{-j \phi}|\bar 1\rangle  )|+\rangle_a
    \end{align}
    \item Finally, we discard the auxiliary system. The total operation depicted in the Fig. \ref{fig:fast_RZ} therefore implements the $\hat R_{Z}(\theta)$ on the state $|\psi_0\rangle$ when the auxiliary rotation angle is fixed to $\phi = - 2\theta$.
\end{enumerate}

Thus, using an auxiliary qubit, we can implement a fast $\hat R_Z(\theta)$ gate without resorting to gate teleportation. We further note that, for the physical implementation considered here, the dissipative stabilisation of the DRCC-1 must be turned off during the application of the $\hat C_{\Pi_1}$ operation, as the number operator $\hat n_1$ does not commute with any of the dissipators used to stabilise the DRCC-1. Furthermore, a faulty auxiliary qubit can affect the gate performance.  

\section{Quantum error correction with a standard dual-rail code and an outer repetition code}
\label{sec:qec_drc}

The standard dual-rail code converts a single-photon-loss event into a detectable erasure by mapping the single-photon codespace to the vacuum state~$|00\rangle$~\cite{Yale_dual}. Knowledge of the erasure location does not change the distance~$d$ of the outer code, but allows a distance-$d$ code to correct up to $d-1$ located erasures, compared with $\lfloor(d-1)/2\rfloor$ arbitrary Pauli errors at unknown locations~\cite{S_puri_erasure_threshold,chadwick2025erasureminesweeperexploringhybriderasure}.

We consider two dual-rail blocks, $A=(1,2)$ and $B=(3,4)$, concatenated with an outer two-qubit repetition code:
\begin{align}
|0_{\rm rep}\rangle
&=|+^{\rm dr}\rangle_A|+^{\rm dr}\rangle_B,\\
|1_{\rm rep}\rangle
&=|-^{\rm dr}\rangle_A|-^{\rm dr}\rangle_B,
\label{eq:drc-rep-code}
\end{align}
where
\begin{align}
|\pm^{\rm dr}\rangle
=\frac{|01\rangle\pm|10\rangle}{\sqrt{2}}.
\end{align}
An arbitrary encoded state is
\begin{align}
|\psi_0^{\rm dr}\rangle
=c_0|0_{\rm rep}\rangle+c_1|1_{\rm rep}\rangle.
\end{align}

Suppose that a photon is lost from block $A$. Losses from its two modes produce
\begin{align}\label{eq:a_1drc}
\hat a_1|\psi_0^{\rm dr}\rangle
&=\frac{|00\rangle_A}{\sqrt{2}}
\left(c_0|+^{\rm dr}\rangle_B-c_1|-^{\rm dr}\rangle_B\right),\\
\hat a_2|\psi_0^{\rm dr}\rangle
&=\frac{|00\rangle_A}{\sqrt{2}}
\left(c_0|+^{\rm dr}\rangle_B+c_1|-^{\rm dr}\rangle_B\right).
\label{eq:drc-rail-dependent-loss}
\end{align}

Eqs.~\eqref{eq:a_1drc} and~\eqref{eq:drc-rail-dependent-loss}
show that $\hat a_1$ and $\hat a_2$ map the erased block to the same vacuum
state but produce opposite relative phases in the surviving block. The
corresponding expressions for losses from block $B$ are obtained by
exchanging blocks $A$ and $B$.

Let $\hat P_{\rm rep}$ denote the projector onto the repetition-code
subspace. The corresponding projected error products are
\begin{align}
\hat P_{\rm rep}\hat a_1^\dagger\hat a_1\hat P_{\rm rep}
&=
\hat P_{\rm rep}\hat a_2^\dagger\hat a_2\hat P_{\rm rep}
=\frac{1}{2}\hat P_{\rm rep},\\
\hat P_{\rm rep}\hat a_1^\dagger\hat a_2\hat P_{\rm rep}
&=
\hat P_{\rm rep}\hat a_2^\dagger\hat a_1\hat P_{\rm rep}
=\frac{1}{2}\hat Z_{\rm rep},
\end{align}
where
\begin{align}
\hat Z_{\rm rep}
=|0_{\rm rep}\rangle\langle0_{\rm rep}|
-|1_{\rm rep}\rangle\langle1_{\rm rep}|.
\end{align}
Since $\hat Z_{\rm rep}$ is not proportional to $\hat P_{\rm rep}$, the error set $\{\hat a_1,\hat a_2\}$ does not satisfy the KL conditions and therefore cannot be exactly corrected by a common recovery based only on the block-level syndrome. The following calculation illustrates the physical origin of this obstruction.

A block-level joint-parity check identifies block $A$ as erased but does not distinguish whether the loss was generated by $\hat a_1$ or $\hat a_2$. After syndrome detection, we reinitialise block $A$ in $|+^{\rm dr}\rangle_A$. The state conditioned on a particular loss process is
\begin{align}
|\psi^{\rm dr}_{1,\mu}\rangle
&=
|+^{\rm dr}\rangle_A
\left(
c_0|+^{\rm dr}\rangle_B
+s_\mu c_1|-^{\rm dr}\rangle_B
\right),\\
s_1&=-1,\qquad s_2=+1,
\nonumber
\end{align}
where $\mu=1,2$ labels the mode from which the photon was lost. In the ideal setting, $|+^{\rm dr}\rangle$ can be prepared deterministically by resetting the block to $|01\rangle$ and applying a suitably phased $50-50$ beam splitter.

One possible recovery operation is
\begin{align}
\hat U_{\rm dr}^{A\leftarrow B}
&=
\hat I_A\otimes|+^{\rm dr}\rangle_B\langle+^{\rm dr}|
+\hat Z_{\rm dr}^{(A)}
\otimes|-^{\rm dr}\rangle_B\langle-^{\rm dr}|,
\end{align}
where $\hat Z_{\rm dr}
=
|+^{\rm dr}\rangle\langle-^{\rm dr}|
+|-^{\rm dr}\rangle\langle+^{\rm dr}|.$
The action of the recovery unitary $\hat U^{A\leftarrow B}_{\rm dr}$ on the re-intialised state $|\psi_{1,\mu}^{\rm dr}\rangle$ gives
\begin{align}
\hat U_{\rm dr}^{A\leftarrow B}
|\psi_{1,\mu}^{\rm dr}\rangle
=
\begin{cases}
\hat Z_{\rm rep}|\psi_0^{\rm dr}\rangle, & \mu=1,\\
|\psi_0^{\rm dr}\rangle, & \mu=2.
\end{cases}
\end{align}

Thus, after the erased block is reinitialised, the same recovery operation recovers the original state following an $\hat a_2$ loss but leaves a logical phase error following an $\hat a_1$ loss. Because the block-level joint-parity syndrome does not distinguish the $\hat a_1$ and $\hat a_2$ loss processes, the required conditional phase correction cannot be applied. This explicitly illustrates the obstruction identified by the Knill--Laflamme conditions.

This limitation highlights a difference between the standard dual-rail code and DRCC-1 for the particular outer code and syndrome considered here. In the standard dual-rail code, losses from both modes map the erased block to the same vacuum state, while the resulting mode-dependent relative phase is not revealed by the block-level joint-parity check. In DRCC-1, by contrast, the error spaces generated by different single-mode losses remain distinguishable through the orthogonal leakage states $|\Phi^+\rangle$ and $|\Phi^-\rangle$. The protocol in
Sec.~\ref{sec:conc_code} exploits this leakage-space information, together
with the outer-code syndrome, to correct the dominant single-photon loss
events in the large-$\alpha$ regime.

 \section{Alternative QEC protocol with the DRCC qubits}\label{sec:alt_qec}

In Sec.\ref{sec:conc_code}, we discussed QEC with DRCC-1 qubits concatenated with the two-qubit and three-qubit repetition code. Here, we outline an alternative approach to executing the syndrome extraction and the subsequent error-correction protocol. 

To correct single-order photon-loss errors, we proceed with the following steps.
\begin{itemize}

\item[S1:] Suppose a single-photon loss  error occurs on the $i{-\rm th}$ mode. Thus, the states (in the large-$\alpha$ limit) after the noise are the following 
\begin{align}
    |\psi_1\rangle &\sim \begin{cases}
        (c_0|\Phi_-\rangle|\bar{0}\rangle+ c_1|\Phi_+\rangle|\bar{1}\rangle) & \mbox{if the error is } \, a_1\\
        (c_0|\Phi_+\rangle|\bar{0}\rangle+ c_1|\Phi_-\rangle|\bar{1}\rangle) & \mbox{if the error is } \, a_2\\
        (c_0|\bar{0}\rangle|\Phi_-\rangle+ c_1|\bar{1}\rangle|\Phi_+\rangle) & \mbox{if the error is } \, a_3\\
        (c_0|\bar{0}\rangle|\Phi_+\rangle+ c_1|\bar{1}\rangle|\Phi_-\rangle) & \mbox{if the error is } \, a_4
    \end{cases}
\end{align}
    \item[S2:] We perform the JPC to detect whether any single-photon loss error has occurred on the DRCC-1 qubits. If the auxiliary qubit stays is in the state $|g\rangle -|e\rangle$ after the JPC operation, it flags a photon loss error $\{\hat a_1,\hat a_2,\hat a_3,\hat a_4\}$ as these errors map the codespace to an even parity subspace (see  Fig. \ref{fig:DRCC_code_pic}).
    
    In the case of false flagging, the errors are second-order or third-order photon loss, leading to a logical error. On the other hand, with the simple dual-rail code, the JPC results in false negatives due to some errors that are not logical errors, as noted in Fig.~ \ref{fig:DRCC_code_pic}. For example, if  errors  $\hat a_1 \hat a_2$ and $\hat a_3 \hat a_4$ occur, due to the JPC, one obtains the following combined cavity-qubit state: 
    \begin{align}
        |\psi_2\rangle \sim \begin{cases}
           (\hat X_L \otimes \hat I) |\psi_{0}\rangle|g\rangle_{A_1}|g\rangle_{A_2} & \mbox{for} \quad \hat a_1\hat a_2\\
            (\hat I \otimes \hat X_L)|\psi_{0}\rangle|g\rangle_{A_1}|g\rangle_{A_2} & \mbox{for}  \quad \hat a_3 \hat a_4.
        \end{cases} 
    \end{align}
    Thus, these two errors share the same syndrome as the identity error and therefore remain undetected by the JPC. We note that a $ Z$-error on the qubit can also trigger false negatives. 
\begin{figure}
    \centering
    \includegraphics[width=1.0\columnwidth]{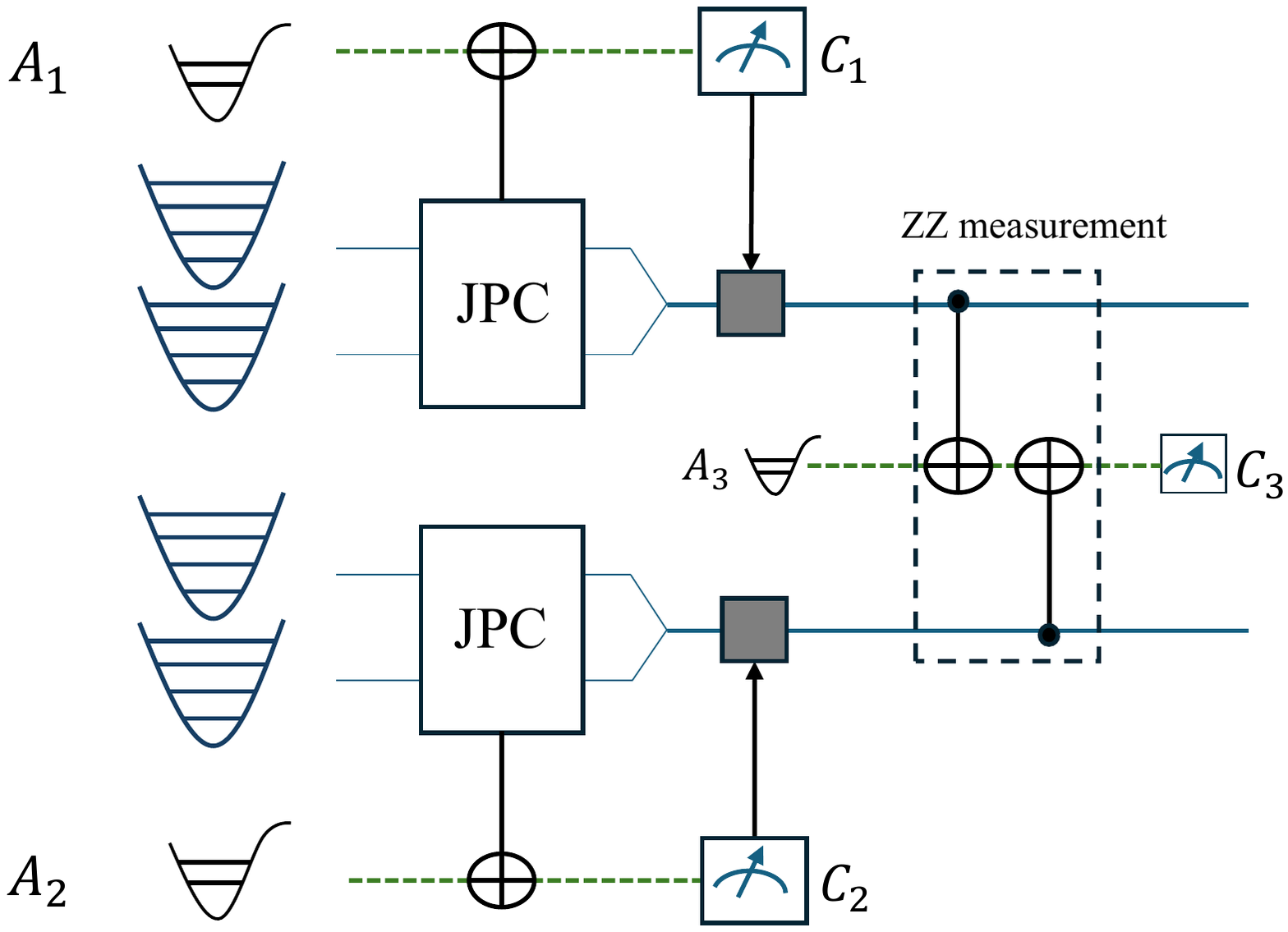}
    \caption{A schematic of the QEC protocol using the two-qubit repetition code with an inner DRCC-1 code without BS. The auxiliary systems $(A_1,A_2,A_3)$ are qubits  - all of them are initialized in the $|g\rangle - |e\rangle$ state. The JPC denotes the joint-parity-check module that acts on the DRCC-1 qubits. If any of the $(C_1,C_2)=0$, we apply the $\hat X_c$ operation defined in Eq.\eqref{eq:X_c} on the first or third cavity. The grey boxes represent these operations.} 
    \label{fig:two-qubit_rep_qec}
\end{figure}
    \item[S3:] Once we detect the loss on the $i^{\rm th}$ mode, we apply the $\hat  X_c = |C_+\rangle\langle C_-|+|C_-\rangle\langle C_+|$ to either the first mode  or the third mode. Note that even if the error occurs on the second or fourth mode, we apply the $\hat X_c$ correction on the first or third mode only. 

    After the $\hat X_c$ gate on the first or on the third mode, we have the following states
    \begin{align}
        |\psi_3\rangle & \sim \begin{cases}
            c_0|\bar{0}\rangle|
            \bar{0}\rangle + c_1 |\bar{1}\rangle|\bar{1}\rangle & \mbox{if the error is } \, \hat a_1\\
            c_0|\bar{1}\rangle|
            \bar{0}\rangle + c_1 |\bar{0}\rangle|\bar{1}\rangle & \mbox{if the error is } \, \hat a_2\\
             c_0|\bar{0}\rangle|
            \bar{0}\rangle + c_1 |\bar{1}\rangle|\bar{1}\rangle & \mbox{if the error is } \, \hat a_3\\
            c_0|\bar{0}\rangle|
            \bar{1}\rangle + c_1 |\bar{1}\rangle|\bar{0}\rangle & \mbox{if the error is } \, \hat a_4
        \end{cases}
    \end{align}

    \item[S4:] After applying  $\hat X_c$, we perform a parity check between the DRCC-1 qubits using a qubit auxiliary qubit (indicated as $A_3$ in  Fig. \ref{fig:two-qubit_rep_qec}), which is initialized in the $|g\rangle$ state. To measure the parity, we apply a pair of cavity-qubit $CX$ (see  Fig.\ref{fig:two-qubit_rep_qec}) defined in the Sec.\ref{sec:logical gates}. If, after applying the gate $\hat X_c$, the parity between the DRCC-1 qubits is odd, meaning that the state of the DRCC-1 qubits is either of the states $\{|\bar{0}\bar{1}\rangle,|\bar{1}\bar{0}\rangle\}$, the qubit flips to the state $|e\rangle$. If the DRCC-1 qubits have the same parity, then the auxiliary qubit $A_3$ remains in the state $|g\rangle$. Therefore, we have the following states after the parity measurement 
    \begin{align}
        |\psi_4\rangle &\sim \begin{cases}
           ( c_0|\bar{0}\rangle|
            \bar{0}\rangle + c_1 |\bar{1}\rangle|\bar{1}\rangle)|g\rangle & \mbox{for the error} \, \hat a_1\\
            (c_0|\bar{1}\rangle|
            \bar{0}\rangle + c_1 |\bar{0}\rangle|\bar{1}\rangle) |e\rangle & \mbox{for the error} \, \hat a_2\\
            ( c_0|\bar{0}\rangle|
            \bar{0}\rangle + c_1 |\bar{1}\rangle|\bar{1}\rangle)|g\rangle & \mbox{for the error } \, \hat a_3\\
           ( c_0|\bar{0}\rangle|
            \bar{1}\rangle + c_1 |\bar{1}\rangle|\bar{0}\rangle)|e\rangle & \mbox{for the error } \, \hat a_4.
        \end{cases}
    \end{align}
    Note that the parity check does not hamper the logical data. Thus, it behaves like a non-demolition measurement which is essential for a QEC protocol~\cite{chuang_nielsen}.

    \item[S5:] If the parity of both DRCC-1 qubits is the same, the qubit remains in the ground state $|g\rangle$. If the parity of the DRCC-1 qubits is odd, the qubit changes its state from $|g\rangle \rightarrow |e\rangle$. The measurement results on the qubit  $C_1,C_2,C_3$ give the syndromes to decode which error amongst $\{\hat a_1,\hat a_2,\hat a_3,\hat a_4\}$ has occurred. Table~\ref{tab:syn_tab} shows the complete syndrome used to correct errors.   
    \item[S6:] Once we know which error has occurred, we apply the logical $\hat X$ gate to either of the DRCC-1 qubits, which is nothing but a 50-50 beam-splitter, as shown in  Table \ref{tab:syn_tab}. 
\end{itemize}

A schematic of the QEC protocols with the two-qubit repetition code is shown in Fig.\ref{fig:two-qubit_rep_qec}. The grey boxes in Fig.\ref{fig:two-qubit_rep_qec} are the $\hat X_c$ gates applied on the first and the third mode depending on the values of the outcomes on $C_1, C_2$. 

\begin{table}[t!]
    \centering
    \caption{Syndrome table for the error correction with the two-qubit repetition code with an inner DRCC-1. }
    \begin{tabular}{c c | c |c c | c}
    \midrule
       $C_1$  & $C_2$ & Conditional operation & $C_3$ & Error & Correction\\
       \midrule
       0  & 1    &  & 0 & $\hat a_1$ & Identity \\
        &     & $\hat X_{c}$ on the first cavity &  & &\\
       0 & 1  &  & 1 & $\hat a_2$ & ${\rm BS}(\pi/4)_{12}$\\
       \midrule
        1  & 0    &  & 0 & $\hat a_3$ & Identity \\
        &     & $\hat X_{c}$ on the third cavity &  & &\\
       1 & 0  &  & 1 & $\hat a_4$ & ${\rm BS}(\pi/4)_{34}$\\
       \midrule
        1 & 1  & Nothing & 0 & No errors & Identity \\
        \midrule
    \end{tabular}
    
    \label{tab:syn_tab}
\end{table}

\section{QEC properties of the three-qubit repetition code with inner DRCC}\label{sec:three-qubit_rep_qec}

We know that the standard three-qubit repetition code is capable of correcting a single bit-flip or a single phase-flip error. In this Appendix, we show that the three-qubit repetition code with an inner DRCC-1 can additionally correct all photon-loss errors up to order $\mathcal{O}(\kappa_i^2 t^2)$, along with all first-order photon-loss errors.

As discussed in Sec.~\ref{sec:conc_code}, the two-qubit repetition code can correct all first-order photon-loss events in the large-$\alpha$ limit. Therefore, it follows that the three-qubit repetition code also corrects all first-order photon-loss events.

To analyse the QEC properties of the DRCC-1 against second-order photon-loss events, we first divide the possible second-order loss operators into the following two sets:
\begin{align}
   G_1
    &=
    \{\hat a_1^2,\hat a_2^2,\hat a_3^2,\hat a_4^2,\hat a_5^2,\hat a_6^2\},\\
   G_2
    &=
    \{\hat a_i \hat a_j\},
    \qquad
    i\neq j.
\end{align}

The operators in $G_1$ commute with the dissipator of the DRCC-1, and the code vectors $\{|\bar{0}\rangle,|\bar{1}\rangle\}$ are simultaneous eigenstates of these operators. Consequently, the errors in $G_1$ are not harmful to the encoded information.

Next, consider the operators of the form $\hat a_i \hat a_{i+1}$ belonging to $G_2$. These operators act as logical bit-flip errors on the DRCC-1 code vectors. Since the outer code is the three-qubit repetition code, such logical bit-flip errors are correctable.

For errors in $G_2$ involving the first and second DRCC-1 blocks, namely $ \{\hat a_1\hat a_3,\hat a_1\hat a_4,\hat a_2\hat a_3,\hat a_2\hat a_4\},$ we obtain the following action on the encoded states:
\begin{eqnarray}
\label{eq:a13,a14}
     \{\hat a_1\hat a_3,\hat a_1\hat a_4\}
     |\bar{0}_3/\bar{1}_3\rangle
     &\propto&
     \{|\Phi_{--}\rangle,|\Phi_{-+}\rangle\}
     \otimes
     |\bar{0}/\bar{1}\rangle,\\
\label{eq:a23,a24}
     \{\hat a_2\hat a_3,\hat a_2\hat a_4\}
     |\bar{0}_3/\bar{1}_3\rangle
     &\propto&
     \{|\Phi_{+-}\rangle,|\Phi_{++}\rangle\}
     \otimes
     |\bar{0}/\bar{1}\rangle,
\end{eqnarray}
where $ |\Phi_{ij}\rangle
    =
    |\Phi_i\rangle\otimes|\Phi_j\rangle.$ Similarly, for the operators $\{\hat a_5\hat a_3,\hat a_5\hat a_4,\hat a_6\hat a_3,\hat a_6\hat a_4\},$ we obtain
\begin{eqnarray}
\label{eq:a53,a63}
     \{\hat a_5\hat a_3,\hat a_6\hat a_3\}
     |\bar{0}_3/\bar{1}_3\rangle
     &\propto&
     |\bar{0}/\bar{1}\rangle
     \otimes
     \{|\Phi_{--}\rangle,|\Phi_{-+}\rangle\},\\
\label{eq:a54,a64}
     \{\hat a_5\hat a_4,\hat a_6\hat a_4\}
     |\bar{0}_3/\bar{1}_3\rangle
     &\propto&
     |\bar{0}/\bar{1}\rangle
     \otimes
     \{|\Phi_{+-}\rangle,|\Phi_{++}\rangle\}.
\end{eqnarray}

Finally, for the operators $ \{\hat a_1\hat a_5,\hat a_1\hat a_6,\hat a_2\hat a_5,\hat a_2\hat a_6\},$
the erroneous states are given by
\begin{align}
\label{eq:a15}
   \hat a_1\hat a_5
   |\bar{0}_3/\bar{1}_3\rangle
   &\propto
   |\Phi_-\rangle
   |\bar{0}/\bar{1}\rangle
   |\Phi_-\rangle,\\
\label{eq:a16}
   \hat a_1\hat a_6
   |\bar{0}_3/\bar{1}_3\rangle
   &\propto
   |\Phi_-\rangle
   |\bar{0}/\bar{1}\rangle
   |\Phi_+\rangle,\\
\label{eq:a25}
   \hat a_2\hat a_5
   |\bar{0}_3/\bar{1}_3\rangle
   &\propto
   |\Phi_+\rangle
   |\bar{0}/\bar{1}\rangle
   |\Phi_-\rangle,\\
\label{eq:a26}
   \hat a_2\hat a_6
   |\bar{0}_3/\bar{1}_3\rangle
   &\propto
   |\Phi_+\rangle
   |\bar{0}/\bar{1}\rangle
   |\Phi_+\rangle.
\end{align}

Using Eqs.~\eqref{eq:a13,a14}--\eqref{eq:a26}, we conclude that the three-qubit repetition code defined in Eq.~\eqref{eq:three-reps} corrects all second-order photon-loss errors.

\end{document}